\documentclass[reprint,
superscriptaddress,
bibnotes,
 amsmath,amssymb,
 aps,
prx,
floatfix,
]{revtex4-2}
\usepackage[dvipsnames]{xcolor}
\usepackage{booktabs}
\usepackage{comment}
\usepackage{etoolbox}
\usepackage{stmaryrd}
\usepackage[T1]{fontenc}
\usepackage{graphicx,braket,hyperref,amsthm,cleveref,amssymb,amsmath,mathtools}
\usepackage{dcolumn}
\usepackage{bm}
\usepackage{makecell}
\setcellgapes{4pt} 
\makegapedcells
\usepackage{booktabs}   
\usepackage{dcolumn}    
\usepackage{graphicx}
\usepackage{adjustbox}
\usepackage[caption=false]{subfig}
\theoremstyle{theorem} 
\newtheorem{theorem}{Theorem}
\newtheorem*{theorem*}{Theorem} 
\newtheorem*{criterion}{Criterion}
\newtheorem{proposition}[theorem]{Proposition}
\newtheorem{lemma}[theorem]{Lemma}
\newtheorem{definition}{Definition}[section]

\newtheorem{remark}{Remark}[section]
\newtheorem{corollary}[theorem]{Corollary}

\usepackage{graphicx}
\usepackage{booktabs}
\usepackage[percent]{overpic}

\usepackage[dvipsnames]{xcolor}
\usepackage{hyperref}
\hypersetup{
colorlinks=true,
linkcolor=BrickRed,
filecolor=blue,
citecolor=MidnightBlue,  
urlcolor=MidnightBlue,
}

\usepackage{todonotes}

{\par}
\newcommand{\supp}{\mathrm{supp}}
\newcommand{\vdef}{\coloneqq}
\newcommand{\ckz}{\mathrm{C}^{k-1}{Z}}
\newcommand{\ccz}{\mathrm{CC}{Z}}
\newcommand{\dimasd}{\mathrm{dim}_{\mathrm{Assouad}}}

\begin{document}



\title{Intrinsic locality dimension of quantum codes}
\date{\today}
\author{Yimin Lu}
\affiliation{Yau Mathematical Sciences Center, Tsinghua University, Beijing, 100084, China }

\author{Esther Xiaozhen Fu}
\affiliation{Joint Center for Quantum Information and Computer Science, University of Maryland, College Park,
Maryland 20742, USA}

\author{Zi-Wen Liu}
\affiliation{Yau Mathematical Sciences Center, Tsinghua University, Beijing, 100084, China }

\begin{abstract}

    Quantum error-correcting codes are a cornerstone of quantum computing, with broad and profound connections to physics and mathematics. 
In this work, we introduce the notion of intrinsic locality dimension of stabilizer codes, which is independent of the underlying geometry of quantum codes and naturally extends to non-integer values. Drawing on mathematical tools from fractal geometry and geometric measure theory, the intrinsic locality dimension accommodates flexible architectures and provides a quantitative measure of code connectivity, encompassing both topological codes and algebraic constructions such as bivariate-bicycle-type codes. We show how the intrinsic dimension serves as a fundamental organizing parameter that unifies code properties.  In particular, we prove general limitations on code parameters and compatible fault-tolerant logical gates induced by the intrinsic dimension, generalizing the Bravyi--Poulin--Terhal and Bravyi--K\"{o}nig bounds for regular topological codes, respectively. Furthermore, we consider implications on thermal properties: toward fully characterizing the geometry requirement for self-correcting quantum memories (SCQMs), we present a conditional no-go result for SCQMs in dimension $3-\epsilon$ and take stock of existing results on low-dimensional SCQMs. Our theory provides a unifying mathematical framework for understanding the fundamental capabilities and geometric implementations of quantum error correction and fault tolerance.


\end{abstract}
\maketitle


\section{Introduction}\label{sec:introduction}

The realization of reliable and scalable quantum computation relies on quantum error correction (QEC) to protect quantum information from diverse noise effects including environmental decoherence and operational errors. A central challenge in QEC is not only to suppress logical errors, but to do so with as low as possible overhead on certain realistic quantum hardware. This has led to a broad and intensive search for quantum codes that simultaneously achieve good parameters, versatile operations, efficient decoding, and high compatibility with practical architectures.

At one end of this landscape are the topological codes~\cite{kitaev2003fault,dennis2002topological}, whose geometric locality makes them particularly appealing for experimental implementation. However, their locality comes at a cost of suboptimal code parameters. At the other end are the asymptotically good quantum low-density parity-check (qLDPC) codes~\cite{panteleev2022asymptotically,dinur2023good,Golowich_2024} which achieve optimal scaling of rate and distance with bounded-weight checks, but are far from ideal for practical realization because they cannot be embedded into standard finite-dimensional lattices without incurring unmanageable long-range connectivity. These two extremes expose a fundamental tension between error correction performance and spatial locality~\cite{Dai2025LocalityvsQuantumCodes,Baspin_2022_2}. 

The difficulty in reconciling this tension has spurred sustained efforts across both theoretical and experimental domains. On the theoretical side, several works have established parameter bounds for geometrically local codes and more general variants~\cite{haah2012logical,Baspin_2022,Baspin_2022_2,fu2025nogotheoremslogicalgates,Dai2025LocalityvsQuantumCodes,dai_et_al:LIPIcs.TQC.2025.4} while new code families have been developed that (nearly) saturate  existing parameter bounds~\cite{portnoy2023localquantumcodessubdivided,Williamson_2024,Li_2024_GeometricallyLocal,lin2023geometrically,Jeong_PhysRevA.83.042330}. In parallel, substantial effort has been devoted to finding new code families and syndrome extraction schemes that allow a controlled amount of long-range connectivity in exchange for improved rate, distance, or overhead~\cite{Berthusen_2025,Eberhardt_2024_BB_logical,Shaw_2025_morphing,Voss_2025_MB,mianMultivariateMulticycleCodes2026,Hong_2024_LRESC,Old_2024_LCS}. These developments are further supported by emerging hardware platforms such as modular architectures and reconfigurable neutral atom arrays, where more flexible connectivity makes it feasible to implement codes that have higher connectivity requirements~\cite{Strikis_2023,Xu_2024_reconfigurable_atoms,Pecorari_2025_longrange,Poole_2025_rydberg_qldpc,Ye_2025_trapped_ions,Bravyi_2024_highthreshold}. Taken together, these works reveal a rich spectrum of quantum codes between local topological codes and highly nonlocal qLDPC codes. 

However, these works also reveal that our current understanding of quantum codes remains largely fragmented, highlighting the need for a general overarching theory that provides a systematic organizing principle.
A natural unifying perspective is geometry.
For topological codes, the relevant geometric structure is explicitly specified by an ambient spatial manifold or lattice, so the associated notion of dimension is readily manifested. However, for general quantum codes, it is much less clear how their geometric nature should be properly characterized.
One natural approach is to characterize a code by the lowest-dimensional manifold into which it can be locally embedded. This minimum ambient dimension is of great practical interest as it quantifies the geometric requirements for physical realization on hardware with local connectivity. 
This perspective underlies the work of Freedman and Hastings~\cite{freedman2021buildingmanifoldsquantumcodes} which reverse-engineers interesting manifolds from given quantum codes so that their stabilizer generators satisfy the spatial locality constraints imposed by the resulting manifold. This makes explicit how sparse code connectivity can arise from locality in an underlying finite-dimensional geometry.
However, the construction does not itself determine the minimum dimension for a general code, and any characterization based on an ambient realization remains extrinsic. This motivates a more intrinsic geometric framework defined directly from code connectivity.

Here, we develop such an intrinsic theory by introducing the framework of \emph{intrinsic locality dimension}, which provides a unifying geometric principle for systematically characterizing stabilizer quantum codes.  This concept is defined based on the Assouad dimension, a bi-Lipschitz invariant from fractal and metric geometry that captures the multiscale growth of metric neighborhoods and provides a natural measure of effective geometric connectivity. 
For regular topological codes living on finite-dimensional lattices, the intrinsic dimension $\beta$ recovers the usual spatial dimension. More generally, however, $\beta$ can take arbitrary values in $\left[1,+\infty\right]$, including noninteger real numbers and $+\infty$. 

The natural emergence of noninteger dimension places our framework within a broad and historically fundamental theme across mathematics and physics. Beginning with Hausdorff's measure-theoretic formulation, noninteger notions of dimension have become central to fractal and metric geometry~\cite{Hausdorff1918}. In physics, they have played an influential role in critical phenomena on fractal lattices~\cite{GefenMandelbrotAharony1980}, the Wilson--Fisher $\epsilon$ expansion~\cite{WilsonFisher1972}, and dimensional regularization in quantum field theory~\cite{BolliniGiambiagi1972,tHooftVeltman1972}. Although these geometric and analytic uses are conceptually distinct, together they demonstrate the broad power of treating dimension as a continuously varying quantity. Our framework brings this metric-geometric perspective to quantum error correction by introducing an intrinsic, potentially noninteger dimension for code connectivity.

To show that the noninteger regime is concrete rather than merely formal, we also provide systematic constructions of quantum codes with noninteger $\beta$  derived from Laakso spaces~\cite{Laakso2000AhlforsQS} and self-similar quasi-convex fractals. 
As we demonstrate throughout the paper, this intrinsic dimension governs a range of fundamental code properties of central theoretical and practical importance, including basic code parameters, fault-tolerant logical gates, and thermal properties such as passive self-correction, in a unified manner. 
Importantly, the intrinsic dimension is defined solely from the code, constituting a fundamental parameter independent of any background geometry. 
Beyond its theoretical interest, our framework enables the active design of code geometries to optimize code capabilities in versatile scenarios, giving it broad practical appeal.


This paper is organized as follows. In Section~\ref{sec:intrinsic locality dimension}, we introduce the intrinsic locality dimension and elucidate both its mathematical properties and the underlying physical intuition, with a wide range of representative examples. 
As a representative application of our framework, Section~\ref{sec:intrinsic dimension constrains code parameters} establishes our intrinsic code parameter bounds. In Section~\ref{sec:intrinsic dimension bridges code symmetry and indistinguishability}, we introduce the concept of graph-local gates alongside the Assouad--Nagata reduction technique, employing them to prove intrinsic fault-tolerant gate bounds for both single and multiple-code-block settings. Finally, in Section~\ref{sec:towards the nonexistence},  we use this reduction to prove a directional geometric obstruction to self-correcting quantum memories in any dimension below $3$ and discuss various related recent results.
  All detailed background, definitions, proofs, and extended discussions can be found in the appendix.

\section{Intrinsic locality dimension of stabilizer codes}\label{sec:intrinsic locality dimension}

\subsection{Assouad dimension and geometric classification of stabilizer codes}\label{subsec:Assouad dimension}

 To properly characterize the intrinsic geometrical degrees of freedom associated with codes, we introduce a key concept from fractal and metric geometry called \emph{Assouad dimension}.
 Originally introduced in the study of metric embeddings, this notion has been widely used in mathematics to study the feasibility of embedding an arbitrary metric space $(K,d_{K})$ into a Euclidean space~\cite{assouad1977espaces,assouad1983plongements} and plays an important role in, e.g., Kakeya-type problems~\cite{Wang_2025}. Here, we demonstrate its fundamental role in quantum error correction. 

The key starting point is that every stabilizer generator set naturally defines an intrinsic metric geometry through its corresponding connectivity graph or Tanner graph. 
Then the Assouad dimension of this graph formally measures the maximal growth of local connectivity over all locations and length scales. 
We show that for conventional topological stabilizer codes, this intrinsic dimension agrees with the usual integer dimension of the underlying lattice or manifold. 
Beyond this setting, however, the same definition applies directly to general sparse stabilizer code families, including those with flexible architectures or noninteger-dimensional connectivity. 
In this way, the Assouad dimension provides an intrinsic notion of dimension for quantum codes, independent of any ambient embedding space.

We will show that this intrinsic dimension is not merely a geometric descriptor: it controls wide-ranging code properties. In the following sections, we show that  $\beta$ imposes rate--distance tradeoffs, restricts fault-tolerant logical gates, and constrains thermal behaviors. 

We proceed by first defining the Assouad dimension and introducing some important properties of the Assouad dimension that will be useful later. 

\begin{definition}[Assouad dimension] \label{Assouad dimension short} (See also definition~\ref{Assouad dimension})
    The \emph{Assouad dimension} of a family of metric spaces 
\(\mathcal{X}=\{(X_{\lambda},d_{\lambda})\}_{\lambda\in I}\)
is defined as the infimum over all \(\beta\) for which the following uniform covering condition holds: there exists a constant \(C>0\), independent of \(\lambda\), \(x\), \(r\), and \(R\), such that for every \(\lambda\in I\), every \(x\in X_{\lambda}\), and all \(0<r<R\), the ball \(N_{x}(R)\) can be covered by at most $C\cdot(\frac{R}{r})^{\beta}$
balls of radius \(r\).
 \end{definition}

A stabilizer code can be represented by a Tanner or connectivity graph. The former is a bipartite graph defined on qubits and check nodes (see Definition~\ref{tanner graph}). The latter is defined only on the set of qubits, with an edge between two vertices if the corresponding qubits share a common stabilizer generator (see also Definition~\ref{Connectivity graph}). As we will show, the two representations give rise to the same intrinsic dimension (see Corollary~\ref{Tanner/connectivity unified}). We will simply use connectivity graphs throughout the rest of the paper.
 \begin{definition}[Intrinsic dimension]\label{code_dimension}

Let \(\{C_{\lambda}(\mathcal{S}_{\lambda})\}_{\lambda\in I}\) be a code family with specified stabilizer generator sets, where the index $\lambda$ parametrizes the system size.
Let \(\mathcal{X}=\{G_{\lambda}\}_{\lambda\in I}\) be the associated family of connectivity graphs.
     The \emph{intrinsic locality dimension} $\mathrm{dim}(C_{I}(\mathcal{S}_{I}))$ is defined as the Assouad dimension of  $\mathcal{X}$ at all scales:
     \begin{equation}
         \mathrm{dim}(C_{I}(\mathcal{S}_{I}))\vdef\dimasd(\mathcal{X})=\beta\in \{0\}\cup \left[1,+\infty\right].
     \end{equation}
Here each \(G_{\lambda}\) is required to be a finite connected graph (corresponding to a single finite-size code block). The case \(\beta=0\) occurs if and only if \(|G_{\lambda}|\leq M\) for some constant \(M\) independent of \(\lambda\).

     For brevity, we refer to the intrinsic locality dimension simply as the intrinsic dimension throughout the rest of the paper.
 \end{definition}

Naturally, the intrinsic dimension captures the growth rate of code connectivity across different scales. See e.g.~Fig.~\ref{Assouad condition}(a) for an illustration. For the leaf-shaped graph, a large ball in the graph can be covered by a uniformly bounded number of smaller balls, independent of the scale and the system size.

 We note that the intrinsic dimension is ultimately associated with a specific stabilizer generator set family (although we may sometimes speak of it as being associated with codes for simplicity of exposition). 
The choice of generators provides an important layer of flexibility, giving rise to a highly versatile framework that allows us to, e.g., incorporate architecture or experiment-specific structures. 
This flexibility is particularly useful in the study of fault-tolerant logical gates, as discussed in Section~\ref{sec:intrinsic dimension bridges code symmetry and indistinguishability}.  The minimal intrinsic dimension over all generating sets is itself a profound mathematical object, as we will also discuss.

Employing the Assouad dimension affords maximum theoretical flexibility in establishing our theorems for general connectivity graphs, as it naturally upper bounds other conventional fractal dimensions:
 \begin{equation}
     \mathrm{dim}_{\mathrm{Hausdorff}}X\leq \mathrm{dim}_{\mathrm{Minkowski}}X\leq \dimasd X.
 \end{equation}
 \begin{figure}[htbp]
    \vspace{-1em}
    \centering
    \subfloat{
        \includegraphics[width=\linewidth]{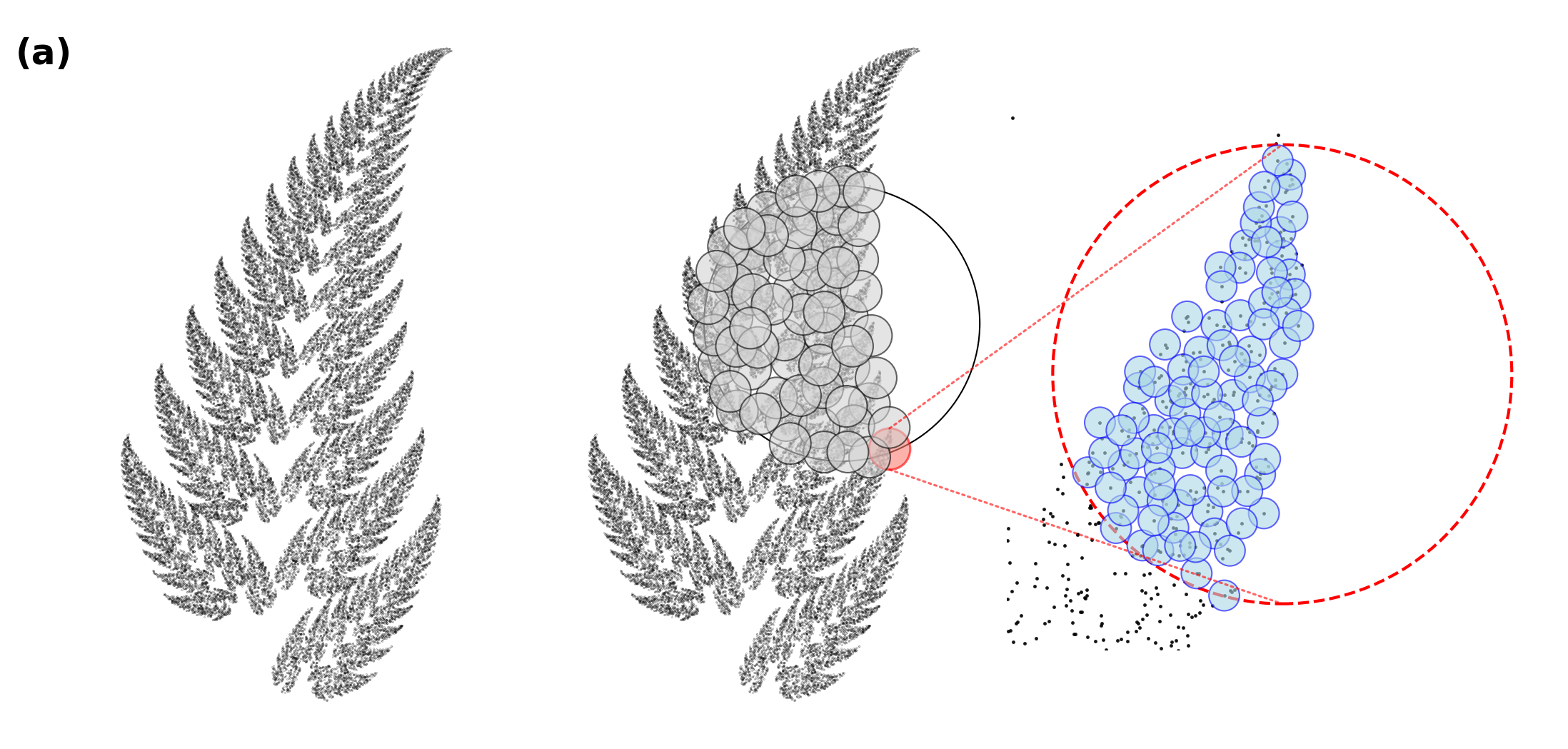}
    }
    \\ 
    \subfloat{
        \includegraphics[width=\linewidth]{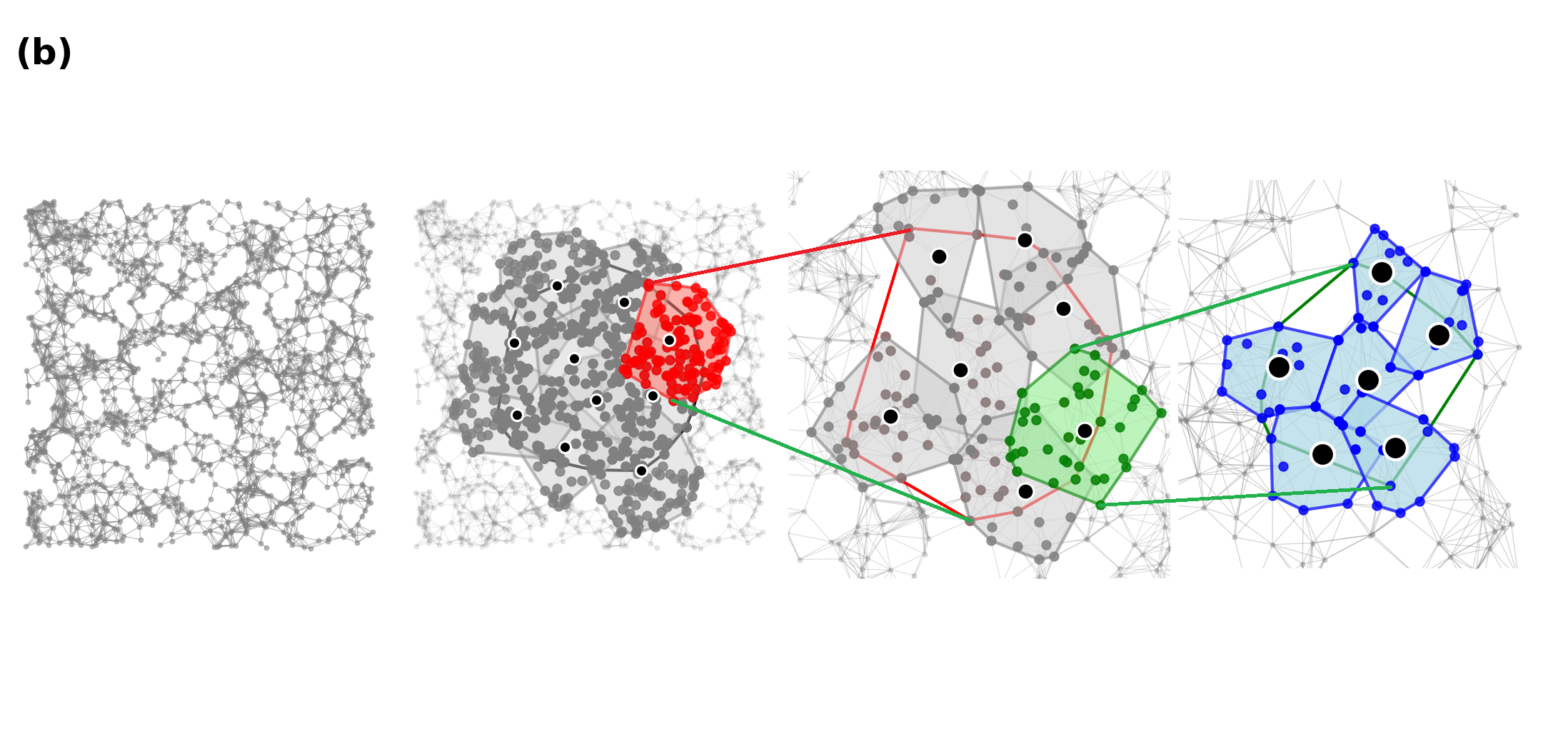}
    }
    
    \caption{Assouad dimension. (a) After ``zooming out'' so that a ball of radius $r$ becomes a ``unit'' ball. The number of such $r$-balls required to cover a ball $N_{x}(R)$ is at most $C(R/r)^\beta$ where $\beta$ is the Assouad dimension given by the infimum over all such exponents for all scales $r>0$. (b) Doubling condition: An example of the doubling condition for $R=4,2,1$ on a random unweighted graph with 1500 vertices, where each polygon is a discrete ball of radius $R$ centered at a black dot.}
    \label{Assouad condition}
\end{figure}

By definition, every family $\{C_{\lambda}(\mathcal{S}_{\lambda})\}_{\lambda\in I}$ such that $n$ goes to infinity as $\lambda$ goes to infinity falls into one of the following two cases:
\begin{enumerate}
    \item $\mathrm{dim}(C_I(\mathcal{S}_{I}))=\beta\in \left[1,+\infty\right)$;
    \item $\mathrm{dim}(C_I(\mathcal{S}_{I}))=+\infty$.
\end{enumerate}
The first case consists of code families whose connectivity remains a uniform polynomial growth with index $\beta$ at every scale. Examples include topological codes and multivariate multicycle codes with long-range connectivity. For the infinite case, examples include asymptotically good qLDPC codes, and more general stabilizer code families with unbounded check weights or unbounded degree. See Table~\ref{Assouad_dim_common_codes} for a discussion of various common codes and their corresponding intrinsic dimensions.

An equivalent characterization of having finite Assouad dimension is the \emph{doubling condition}:
\begin{proposition}\label{doubling_condition}
    A family of graphs $\mathcal{X}$  has finite Assouad dimension if and only if it is \emph{doubling}, i.e., there exists some constant $C$, such that for all $R$ every ball of radius $2R$ in $\mathcal{X}$ can be covered by at most $C$ balls of radius $R$.
\end{proposition}See Proposition~\ref{doubling condition} for proof details. 
Intuitively, this condition says that the connectivity is uniformly sparse at every scale, i.e., any ball of radius $R$ in $G_{\lambda}\in\mathcal{X}$ is connected to at most a constant number of neighboring balls of radius $R$ (See Fig.~\ref{Assouad condition}(b)). In particular, for $\mathcal{X}$ with finite dimension, each vertex on $G_{\lambda}$ has a bound degree $C$ by taking $R=\frac12$ in Proposition~\ref{doubling_condition}. This is a natural geometric constraint to have for physical realization, since unbounded connectivity would be difficult to realize experimentally.

\begin{table*}[t]
\caption{Intrinsic dimensions of various representative quantum code families: a summary.}
\label{Assouad_dim_common_codes}
\begin{ruledtabular}
\begin{tabular}{l l c}
\parbox[c]{0.27\textwidth}{\centering \textrm{Code family}}
&
\parbox[c]{0.47\textwidth}{\centering \textrm{Setting information}}
&
\textrm{Intrinsic dimension}
\\
\colrule

\parbox[t]{0.27\textwidth}{\raggedright Regular topological codes}
&
\parbox[t]{0.47\textwidth}{\raggedright Defined by good triangulation on manifold \(\mathcal{M}^{D}\)}
&
\(D\)
\\[1.0ex]

\parbox[t]{0.27\textwidth}{\raggedright Fractal topological codes}
&
\parbox[t]{0.47\textwidth}{\raggedright Defined on quasi-convex, self-similar fractal of dimension \(\beta\)}
&
\(\beta\)
\\[1.0ex]

\parbox[t]{0.27\textwidth}{\raggedright Multivariate multicycle codes\\(constant-range)}
&
\parbox[t]{0.47\textwidth}{\raggedright Defined by translation-invariant constant-range stabilizer generators in
\(\mathbb{F}_2[x_1,\ldots,x_D]/(x_1^{L_1}-1,\ldots,x_D^{L_D}-1)\)}
&
\(D\)
\\[1.0ex]

\parbox[t]{0.27\textwidth}{\raggedright Multivariate multicycle codes\\(growing range)}
&
\parbox[t]{0.47\textwidth}{\raggedright Defined by \(t\) long-range translation-invariant stabilizers of weight at most \(M\)}
&
\(\beta\leq t\binom{M}{2},\ \beta\in\mathbb{Z}\)
\\[1.0ex]

\parbox[t]{0.27\textwidth}{\raggedright Hypergraph product codes}
&
\parbox[t]{0.47\textwidth}{\raggedright Constructed from homological product of classical codes with Ahlfors \(\beta_1,\beta_2\)-regularity}
&
\(\beta_1+\beta_2\)
\\[1.0ex]

\parbox[t]{0.27\textwidth}{\raggedright Asymptotically good qLDPC codes}
&
\parbox[t]{0.47\textwidth}{\raggedright Must have unbounded geometric locality}
&
\(+\infty\)
\\[1.0ex]

\parbox[t]{0.27\textwidth}{\raggedright Non-qLDPC codes}
&
\parbox[t]{0.47\textwidth}{\raggedright Have unbounded weight or degree}
&
\(+\infty\)

\end{tabular}
\end{ruledtabular}
\end{table*}

We can establish a sufficient condition based on the regularity conditions of Ahlfors~\cite{ahlfors1935theorie,david1997fractured,david1993analysis} to determine the intrinsic dimension $\beta$. In fractal geometry, the Ahlfors regularity expresses the spatial homogeneity of a metric space. This criterion naturally encompasses the case of translation-invariant stabilizer code families, since the volume growth centered at any fixed qubit suffices to determine its Ahlfors regularity condition, and hence its intrinsic dimension. 
\begin{criterion}[Ahlfors $\beta$-regularity]\label{Asdim_criterion}
Let $C_I=\{C_\lambda\}_{\lambda\in I}$ be a code family with associated
connectivity graphs $\mathcal{X}=\{G_\lambda\}_{\lambda\in I}$.  Suppose
that $\mathcal{X}$ is uniformly Ahlfors $\beta$-regular: namely, there
exist constants $0<v_*<v^*$, independent of $\lambda$, such that for all
$\lambda\in I$, all $q\in V(G_\lambda)$, and all
$0<R\leq \operatorname{diam}(G_\lambda)$,
\[
    v_* R^\beta
    \leq
    \mathrm{Vol}(N_q(R))
    \leq
    \max\{v^*R^\beta,1\}.
\]
where we denote $\mathrm{Vol}(N_{q}(R))\coloneqq\sharp\{x\in N_{q}(R)\subseteq G\}$ as the volume of a ball $N_{q}(R)$ centered at qubit $q$.

Then the intrinsic dimension of $C_I$ is given by $dim(C_I)=\beta $.
\end{criterion}
 \vspace{-0.7em}
  The proof is given in Appendix~\ref{Ahlfors to Assouad}.

Next, we show the existence of a graph for any value of $\beta$ in $\left[1,+\infty\right)$. Our construction is based on the discretization of the family of Laakso space in fractal geometry~\cite{Laakso2000AhlforsQS,ErikssonBique_2025,capolli2024overviewlaaksospaces}:
 \begin{proposition}
     For any $\beta\in \left[1,\infty\right)$, there exists a family of finite connected graphs $\mathcal{X}=\{G_{\delta}\}_{\delta\in \mathbb{R}^{+}}$ with shortest-path metric, which satisfies 
     \begin{equation}
         \dimasd\mathcal{X}=\beta.
     \end{equation}
 \end{proposition}

 \begin{proof}
     Given $\beta>1$, there exists a compact Laakso space $(X,d_{X})$ with Ahlfors $\beta$-regularity and rectifiable geodesics connecting any two points in $X$~\cite{Laakso2000AhlforsQS}. For any scale $\delta>0$, consider the $\delta$-net $\mathrm{Net}(\delta)\subseteq X$, such that for all $u,v \in \mathrm{Net}(\delta), d_{X}(u,v)\geq \delta$ and for all $ x\in X,d_{X}(x,\mathrm{Net}(\delta))\leq \delta$. $\mathrm{Net}(\delta)$ is a discrete metric space equipped with the metric inherited from $X$. We first prove that the family $\widetilde{\mathcal{X}}=\{\mathrm{Net}(\delta), \delta\in \mathbb{R}^{+}\}$ has Assouad dimension $\beta$. $\mathrm{Net}(\delta)\subseteq X$ implies $\dimasd \widetilde{\mathcal{X}}\leq \beta$. It remains to show that  $\dimasd \widetilde{\mathcal{X}}\geq \beta$. Consider the ball $N_u(R)$ of radius $R$ in $X$, and write $k\vdef | N_u(R)\cap \mathrm{Net(\delta)}=\{u_1,u_2,\cdots,u_k\}|$. By the property of $\delta$-net, $N_u(R)\subseteq \mathop{\cup}\limits_{i=1}^k N_{u_i}(\delta)$. Taking the volume and applying Ahlfors regularity in $X$, we have
     \[v_{*} R^{\beta}\leq \mathrm{Vol}(N_{u}(R))\leq k \mathrm{Vol}(N_{u_i}(\delta))\leq k v^{*}\delta ^{\beta}.\] Therefore, $k\geq \frac{v_{*}}{v^{*}}(\frac{R}{\delta})^{\beta}$. Set $r=\frac{\delta}{3}$. By the $\delta$ separation of the net, any ball of radius $r$ contains at most one point in $\mathrm{Net(\delta)}$. Therefore, one requires at least $\frac{v_{*}}{v^{*}}(\frac{R}{\delta})^{\beta}=\frac{v_{*}3^{\beta}}{v^{*}}(\frac{R}{r})^{\beta}$ balls of radius $r$ to cover $N_u(R)$. Hence, $\dimasd \widetilde{\mathcal{X}}\geq \beta$. 
     
     For the last step, to obtain a family of graphs with shortest-path metric, we construct a family $\mathcal{X}=\{G_{\delta}\vdef (V_{\delta},E_{\delta})\}_{\delta\in \mathbb{R}^{+}}$, where $V_{\delta}=\mathrm{Net(\delta)},E_{\delta}=\{(u,v)|d_{X}(u,v)\leq 3\delta\}.$ Any two adjacent vertices $u,v$ have distance $\delta\leq d_X(u,v)\leq 3\delta$ in Laakso space $X$. So by a scaling $\delta$ on the weight of the edges, there exist a uniformly bi-Lipschitz map between $\mathcal{X}$ and $\widetilde{\mathcal{X}}$. Applying proposition~\ref{factor multiplication}, we conclude that $\mathcal{X}$ has Assouad dimension $\beta$ with shortest-path metric.
 \end{proof}
 
Different generating sets of the same stabilizer group $\langle\mathcal{S}\rangle=\langle\mathcal{S}^{\prime}\rangle$ may give rise to non-isomorphic Tanner or connectivity graphs $G(\mathcal{S})\not\simeq G(\mathcal{S}^{\prime})$, and hence the Assouad dimensions may also differ due to possibly non-local products of stabilizer generators. For this reason, the minimal intrinsic dimension of code families are meaningfully defined by taking the infimum over all possible stabilizer generating sets:
\begin{definition}[]
    \begin{equation}
        \beta_0(C_I)\vdef \inf_{\mathcal{S}_I} \dimasd (\{G_I(\mathcal{S}_I)\}),
    \end{equation}
    where $C_I$ is the code family defined by generating set family $\mathcal{S}_I$.
\end{definition} 
 From a heuristic standpoint, because the growth function $R^{\beta}$ is convex with respect to the radius $R$, Jensen's inequality indicates that achieving the minimum dimension $\beta_0$ for the overall connectivity of the stabilizer group requires the strongest spatial homogeneity in entire space. Then by assuming a strengthened Ahlfors regularity condition, which is fairly physical, Theorem~\ref{Intrinsic BPT bound} provides a lower bound for the minimal intrinsic dimension,
 \begin{equation}
     \limsup_{d_{\lambda}\to +\infty} \frac{2}{\log_{d_{\lambda}}(\frac{n_{\lambda}}{k_{\lambda}})}+1\leq \beta_0.
 \end{equation}

An equivalent way of obtaining the Assouad dimension of a family of finite-diameter graphs is to realize the family inside a single infinite graph. Let $\mathcal{X}=\{G_{i} \mid i\in \mathbb{N}^{+}\}$ be a countable family of finite-diameter graphs. We construct an infinite graph, denoted by $\widetilde{G}$ as follows: For each \(i\), choose a base point \(x_i\in G_i\), and glue the sequence of points $x_i\in G_i$ along an infinite ray $[0,\infty)$ at designated base points $l_i$, such that the separation satisfies $l_{i+1}-l_i>2\mathrm{diam}(G_{i})$. And then we define 
\begin{equation}
    \widetilde{G}=\bigcup_{i\in \mathbb{Z}} G_{i}\bigcup \{l\in\left[0,+\infty\right)\}/(x_i\sim l_i).
\end{equation}
By construction, each \(G_i\) is realized as a subgraph of \(\widetilde G\). Moreover, $\dimasd \widetilde{G}=\dimasd \mathcal{X}$.

    The Assouad embedding theorem~\cite{assouad1983plongements,naor2010assouadstheoremdimensionindependent,Non-Probabilistic13} guarantees that $\widetilde{G}$ admits a bi-Hölder embedding $\tau$ of index $\gamma\in \left(\frac{1}{2},1\right)$ into a Euclidean space $\mathbb{R}^{N}, c^{-1}d^{\gamma}_{\widetilde{G}}(x,y)\leq d_{\mathbb{R}^N}(\tau(x),\tau(y))\leq cd^{\gamma}_{\widetilde{G}}(x,y)$. The ambient dimension $N$ for Theorem 1.2 in Ref.~\cite{naor2010assouadstheoremdimensionindependent} can be chosen to be $O(\beta)$, which depends only on the doubling constant $F$, even for $\gamma\to 1$. In general, there are metric spaces—such as the Heisenberg group~\cite{pansu1989metriques}—that fail to admit a bi-Lipschitz embedding into $\mathbb{R}^{N}$. We conjecture that for connectivity graphs with finite Assouad dimension $\beta$, there exists a tight bi-Hölder embedding.

\subsection{Representative code families and their intrinsic dimensions}\label{subsec:representative code family}
Now, we examine the intrinsic dimensions of several representative code families in detail, illustrating how the abstract definition recovers familiar geometric notions while also extending beyond them. 
For  topological codes, we show that their intrinsic dimensions agree with the conventional notion of integer dimensions. Furthermore, we provide a procedure for constructing fractal codes and discuss known examples of fractal codes with noninteger intrinsic dimensions. We also examine the intrinsic dimensions of multivariate multicycle (MM) or group algebra codes, hypergraph product codes and good qLDPC codes as representative instances. As a non-example,  we show in Appendix~\ref{app:Fractalized stabilizer code} that the fractalized quantum codes~\cite{Devakul2021fractalizingquantum} actually have integer intrinsic dimensions and thus cannot exhibit improved code parameters over integer-dimensional codes. 

Key information about the intrinsic dimensions of these code families is summarized in Table~\ref{Assouad_dim_common_codes}.
 
 \subsubsection{Topological codes on manifolds} For a code family derived from an ambient manifold, the intrinsic dimension captures the information of its geometric structure. The following result shows that our definition coincides with the conventional notion of dimension for standard topological codes.

 \begin{theorem}\label{Example1}
     Given a good triangulation (See Definition~\ref{Good triangulation}) $\mathcal{C}_{\mathcal{M,\delta}}^{\bullet}$ at scale $\delta$ of a connected compact Riemannian manifold $(\mathcal{M}^{D},g)$, the associated CSS code $\{C_{\delta},\delta=\frac{1}{L}\}$ obtained from any three consecutive terms in $\mathcal{C}_{\mathcal{M,\delta}}^{\bullet}$ has intrinsic dimension:
     \begin{equation}
         \mathrm{dim}(C_{\delta})=\mathrm{dim}\mathcal{M}=D.
     \end{equation}
 \end{theorem}
 \begin{proof} \emph{(Sketch, details in Appendix~\ref{Dimension_unification_for_classical_code})}

 We first sketch the general construction of a topological code on a compact manifold. Let $\mathcal{M}$ be a connected compact Riemannian manifold that is locally bi-Lipschitz to $\mathbb{R}^D$. Then, we know that its topological dimension is $D$. We can always find a good triangulation for $\mathcal{M}$, for example via the Riemannian Delaunay triangulation (See Appendix~\ref{Delaunay triangulation}). A triangulation is good if its simplexes have a uniformly bounded radius-edge ratio and thickness, and there exists a family of bi-Lipschitz homeomorphisms $H_\delta$ from $\mathcal{C}^{\bullet}_{\mathcal{M},\delta}$ to $\mathcal{M}$ at scale $\delta$. The triangulation induces a simplicial chain complex over the base field $\mathbb{F}$ of $\mathrm{characteristic}~2$. By taking any $3$ consecutive terms $\mathcal{C}_{\mathcal{M},\delta}^{k+1},\mathcal{C}_{\mathcal{M},\delta}^{k},\mathcal{C}_{\mathcal{M},\delta}^{k-1}$ of the chain complex, we obtain a CSS code whose parity check matrices are given by the two boundary maps.

 Next, to show that these two dimensions coincide, our proof proceeds by first showing that the Assouad dimension of $(\mathcal{M}, d_g)$ is $D$ (See Theorem~\ref{Dimension unification}). This is done by showing that the Ahlfors Regularity condition for a compact manifold (see lemma~\ref{Volume on manifold}) holds for $\mathcal{M}$. By lemma~\ref{Uniform net approximation lemma}, the Assouad dimension will be inherited by an arbitrary family of $\delta$-net approxiemation for $\mathcal{M}$. Then, consider the following map restricted on the barycenters $\delta$-net $\mathcal{M}_{\delta}(\mathcal{C}^{k}_{\mathcal{M},\delta})\coloneqq \{H_{\delta}(b(\sigma))\mid \sigma\in \mathcal{C}^{k}_{\mathcal{M},\delta}\}$:
\[ (\mathcal{M},d_{g}) \xrightarrow{H_{\delta}^{-1}} (\mathcal{C}^{\bullet}_{\mathcal{M},\delta},d_{PL}) \xrightarrow{\mathcal{S}_{\delta}} (G_{\delta},d_{\delta}). \]
where $G_\delta$ is the connectivity graph formed from the 3-term chain complex. $d_{PL}$ is the piecewise linear Euclidean metric on $C^\bullet_{\mathcal{M},\delta}$ and $d_\delta(q_i,q_j) = \delta \cdot |\text{shortest path between } q_i, q_j|$ is the natural metric defined on $G_\delta$. Theorem~\ref{Dimension unification for code} shows that $\mathcal{S}_\delta\circ H_\delta^{-1}|_{\mathcal{M}_{\delta}}$ is uniformly bi-Lipschitz between $(\mathcal{M}_{\delta},d_{g})$ and $(G_{\delta},d_{\delta})$ for sufficiently small $\delta$. 
Then, we use proposition~\ref{bi-Lipschitz reduction}, which states that these two families of metric spaces that are uniformly bi-Lipschitz have the same Assouad dimension. Consequently, $\mathrm{dim}(C_{\delta})=\dimasd(G_\delta)=\dimasd(\mathcal{M}_{\delta})=\dimasd{\mathcal{M}} =D$, concluding our proof. 
 \end{proof}

 \subsubsection{Fractal topological codes}

 Here we provide an explicit construction of stabilizer-code families with noninteger intrinsic dimension using quasi-convex self-similar fractals embedded in compact Riemannian manifolds $\mathcal{M}^{D}$. The resulting intrinsic dimension equals the fractal dimension of the attractor. After this construction, we list known examples of codes in the literature with noninteger intrinsic dimensions.
 
  A fractal set $K$ is said to be \emph{quasi-convex}~\cite{heinonen1998quasiconformal,durand2012poincare} if it admits a metric defined by the shortest rectifiable path and the metric on $K$ is bounded by the geodesic distance on ambient manifold $\mathcal{M}$ by a constant factor $\gamma$:
 \begin{equation}
     d_K(x,y)\leq \gamma d_{\mathcal{M}}(x,y).
 \end{equation}
 \begin{remark}
Self-similar fractal sets on manifolds have been extensively studied by Ngai and Xu~\cite{ngai2022separationconditionsiteratedfunction,ngai2022existencelqdimensionentropydimension,liu2024iteratedrelationsystemsriemannian}. Such fractal sets can be generated by \emph{conformal iterated function systems} $\{R_i\}$, satisfying the \emph{weak separation condition}. The Ahlfors regularity of self-similar set unifies all fractal dimensions of the attractor:
 \begin{equation}\label{fractal_dimension_coincide}
\mathrm{dim}_{\mathrm{Hausdorff}}K=\mathrm{dim}_{\mathrm{Minkowski}}K=\dimasd K.
 \end{equation}
 \end{remark} 

We now explain the construction of a general fractal code. Let $\{R_i\}_{i\in[N]}$ be a conformal iterated function system on \(\mathcal M\), satisfying the weak separation condition, and let \(K\subset \mathcal M\) be its compact attractor. 

Start with a good triangulation of the compact manifold $\mathcal{M}^D$, $\mathcal{C}^{\bullet}_{\mathcal{M},\delta}$, as described above for the case of topological codes, we obtain the $D$-simplex packing of $K$ in $\mathcal{C}^{\bullet}_{\mathcal{M},\delta}$ defined by
     \begin{equation}
         {\kappa}_{\mathcal{C},\delta}(K)=\{\sigma\in \mathcal{C}^{D}_{\mathcal{M},\delta}\mid\sigma \cap H_{\delta}^{-1}(K)\neq \emptyset\}.
     \end{equation}
The $k$-cells in the triangulation $\mathcal{C}^{\bullet}_{\mathcal{M},\delta}(K)$ at scale $\delta$ are given by:
     \begin{equation}
         \mathcal{C}^{k}_{\mathcal{M},\delta}(K)=\left\{\sigma\in\mathcal{C}^{k}_{\mathcal{M},\delta}
        \mid\exists\,\tau\in\kappa_{\mathcal{C},\delta}(K)\text{ such that }\sigma\preceq\tau\right\}.
     \end{equation}
The fractal stabilizer code \(C_\delta(K,k)\) is the CSS code associated with three consecutive terms $\mathcal{C}_{\mathcal{M},\delta}^{k+1}(K),\mathcal{C}_{\mathcal{M},\delta}^{k}(K),\mathcal{C}_{\mathcal{M},\delta}^{k-1}(K)$ in the induced chain complex $\mathcal C^\bullet_{\mathcal M,\delta}(K)$. Its connectivity graph $G_{\delta}(K,k)$ is determined by the two boundary maps corresponding to the $X$-checks and $Z$-checks. 

 The intrinsic dimension of fractal codes and the geometric dimension of their underlying self-similar sets are unified by the following theorem.
 \begin{theorem}
     For a quasi-convex self-similar set $K\subset\mathcal{M}$, the intrinsic dimension of $C_{\delta}(K,k)$ agrees with all the fractal dimensions (by Eq.~\eqref{fractal_dimension_coincide}, the Assouad, Hausdorff, and Minkowski dimensions of $K$ coincide) of attractor $K\subset \mathcal{M}$. More precisely,
     \begin{equation}
         \mathrm{dim}\big(C_{\delta}(K,k)\big)=\dimasd K,\quad\forall\,1\leq k\leq D-1.
     \end{equation}
 \end{theorem}
 The proof is similar in spirit to that in Theorem~\ref{Example1} and is given in Appendix~\ref{Appendix Dim unification fractal case}.
 
 The quasi-convexity of the fractal set is a crucial prerequisite for our construction to work. Sufficient conditions for fractal sets to be quasi-convex have been discussed extensively in Refs.~\cite{heinonen1998quasiconformal,durand2012poincare,hajlasz2000sobolev,semmes1996finding}. 
 
 Below, we list several standard quasi-convex fractal sets and their figures in Fig.~\ref{fig:fractalset} (See also examples in Section 6 of~\cite{tyson2002conformal}). Their iterated function systems admit explicit descriptions by holomorphic maps on the complex plane. 
 \begin{itemize}
    \item \emph{Sierpiński gasket and Sierpiński simplex} 
    \begin{equation}
        f_{j}(z)=\frac{1}{2}z+a_j,a_1=0,a_2=\frac{1}{2},a_3=\frac{1+\sqrt{3}\mathrm{i}}{4}.
    \end{equation}
    \item \emph{Sierpiński carpet and Menger sponge}
    \begin{equation}
        \begin{split}
            f_{i,j}(z)=\frac{1}{3}z+a_{i,j},a_{i,j}=\frac{i+j\cdot \mathrm{i}}{3},\\
            0\leq i,j\leq 2,(i,j)\neq(1,1).
        \end{split}
    \end{equation}
    \item \emph{Dendrites}
    \begin{equation}
        \begin{split}
             f_{j}(z)=\frac{1}{3}z+\frac23a_j,a_1=-1+\mathrm{i},a_2=1+\mathrm{i},\\
        a_3=0,a_4=-1-\mathrm{i},a_5=1-\mathrm{i}.
        \end{split}
    \end{equation}
    \item \emph{Hexagasket}
    \begin{equation}
        f_{j}(z)=\frac{1}{3}z+\frac23e^{\frac{\pi \mathrm{i}j}{3}}, 0\leq j\leq 5.
    \end{equation}
 \end{itemize}

\begin{figure}[t]
    \centering

    \begin{minipage}[b]{0.23\columnwidth}
        \centering
        \includegraphics[width=\linewidth]{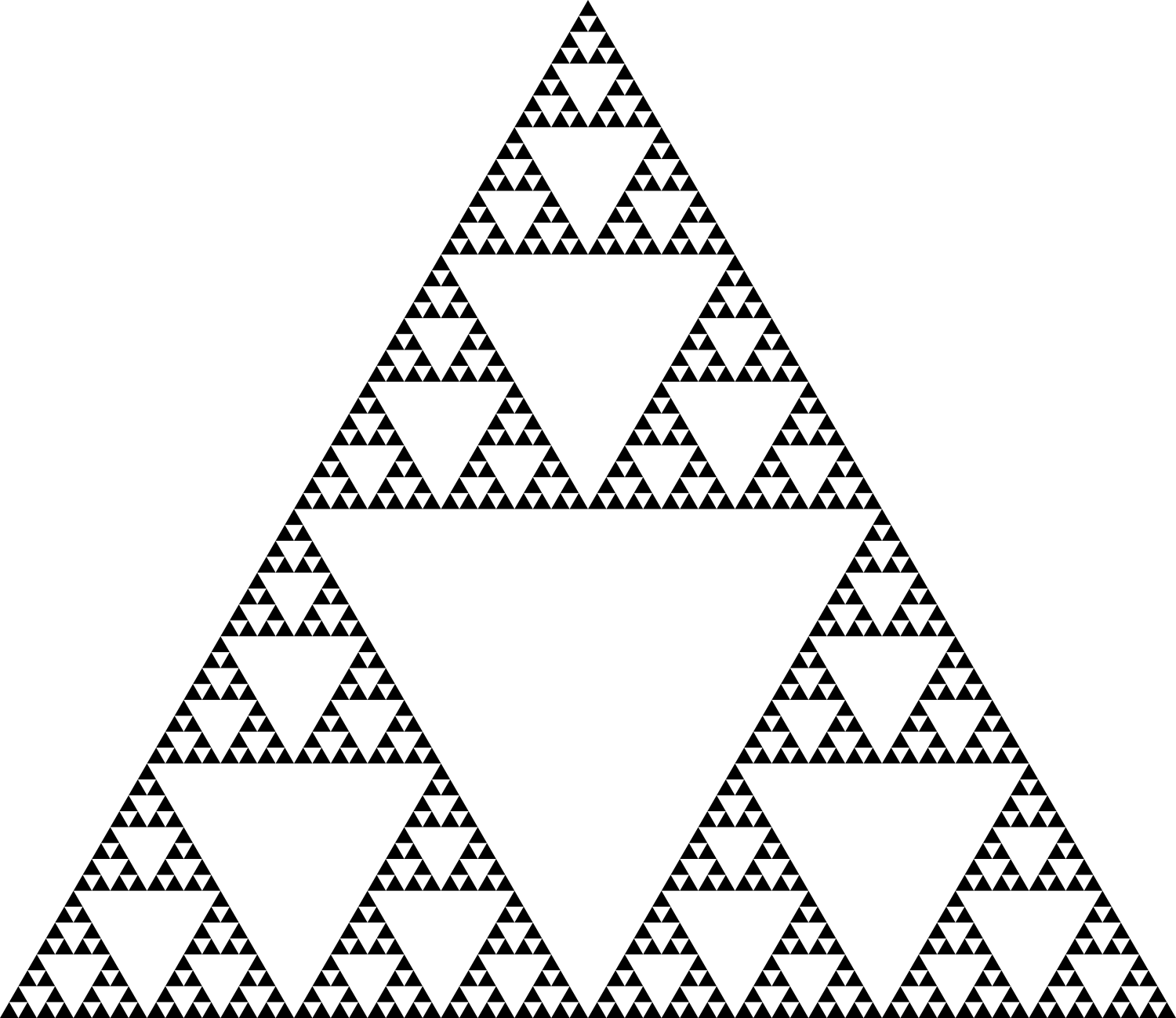}\\[-0.2em]
        {\scriptsize\bfseries (a)}
    \end{minipage}
    \hfill
    \begin{minipage}[b]{0.23\columnwidth}
        \centering
        \includegraphics[width=\linewidth]{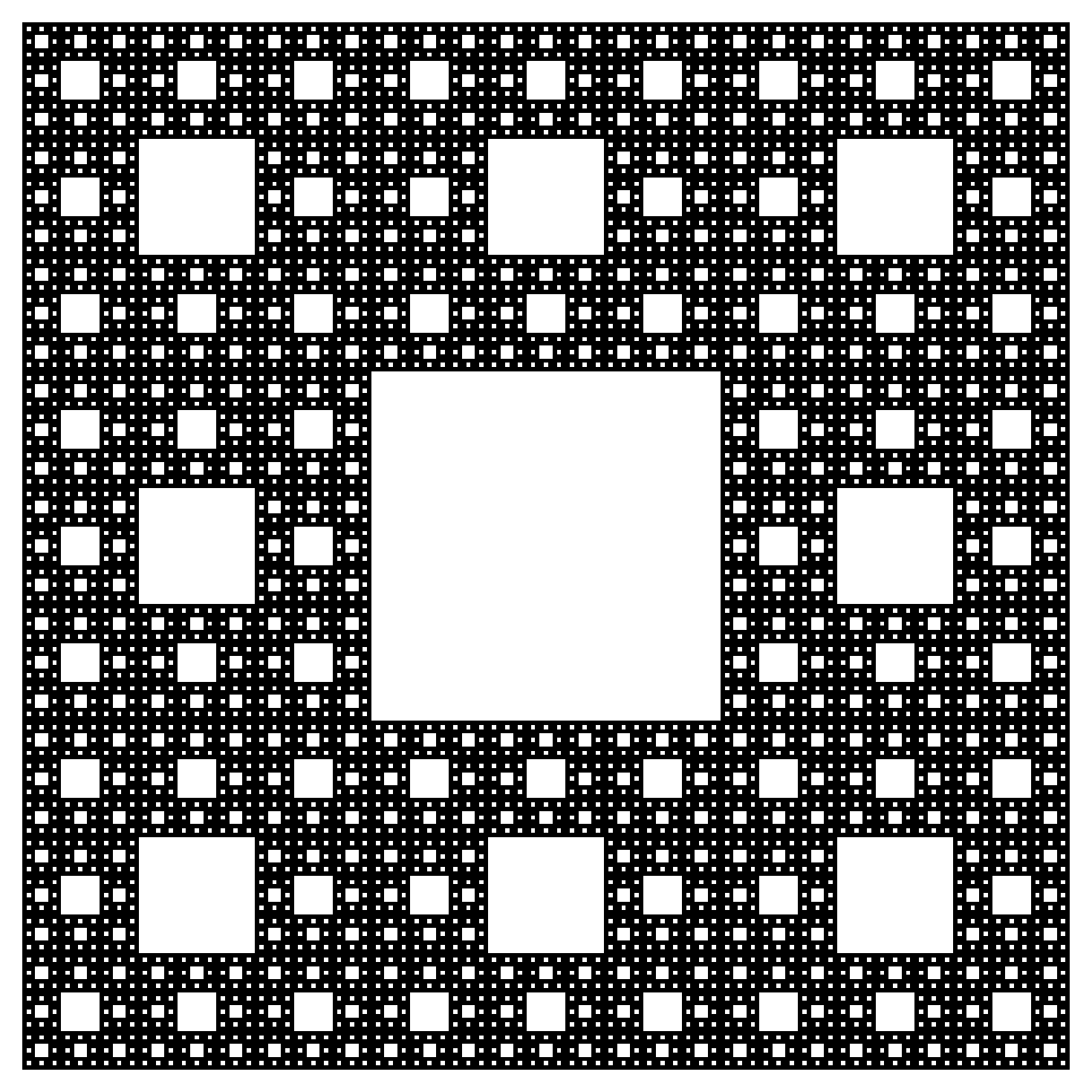}\\[-0.2em]
        {\scriptsize\bfseries (b)}
    \end{minipage}
    \hfill
    \begin{minipage}[b]{0.23\columnwidth}
        \centering
        \includegraphics[width=\linewidth]{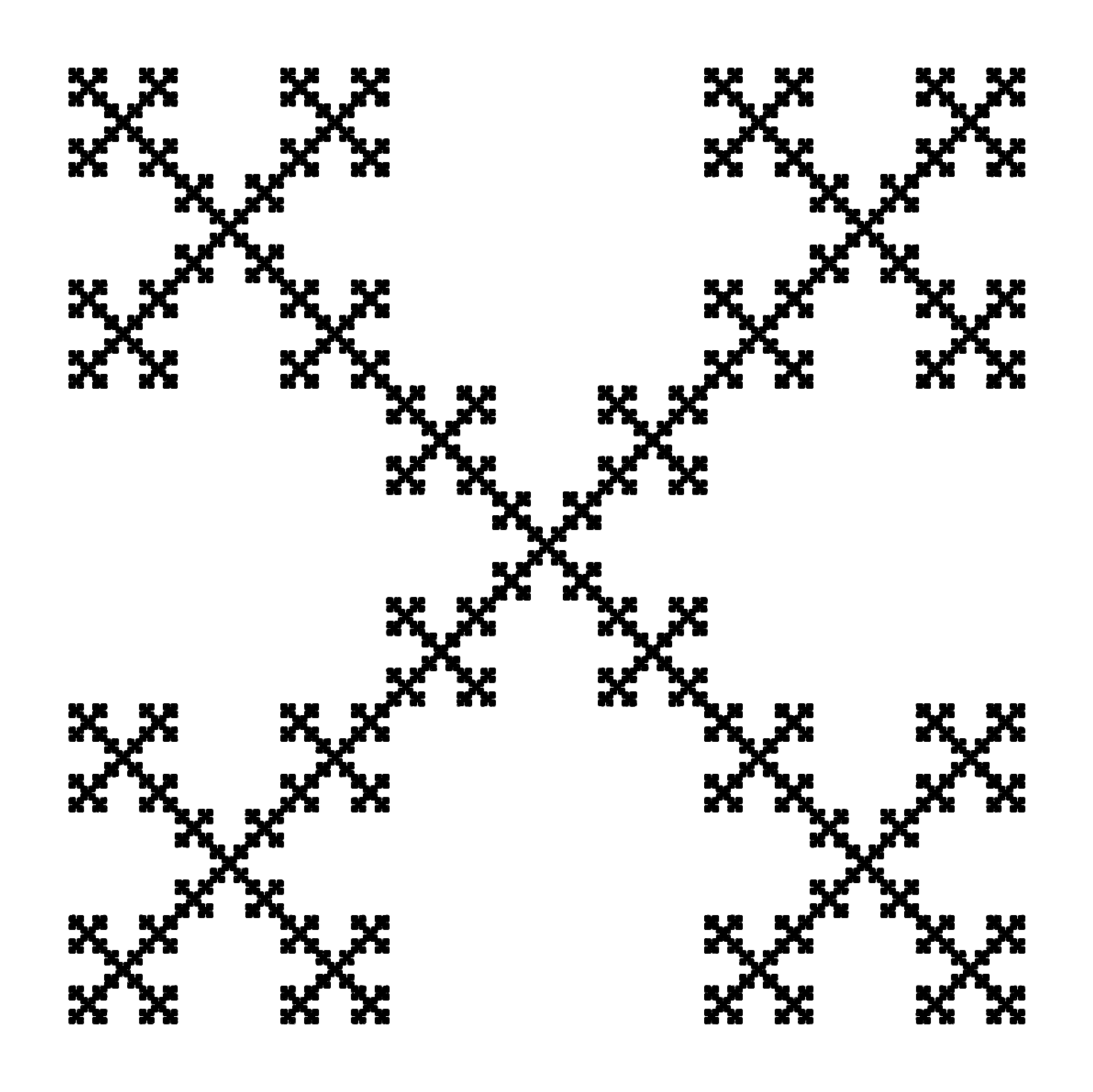}\\[-0.2em]
        {\scriptsize\bfseries (c)}
    \end{minipage}
    \hfill
    \begin{minipage}[b]{0.23\columnwidth}
        \centering
        \includegraphics[width=\linewidth]{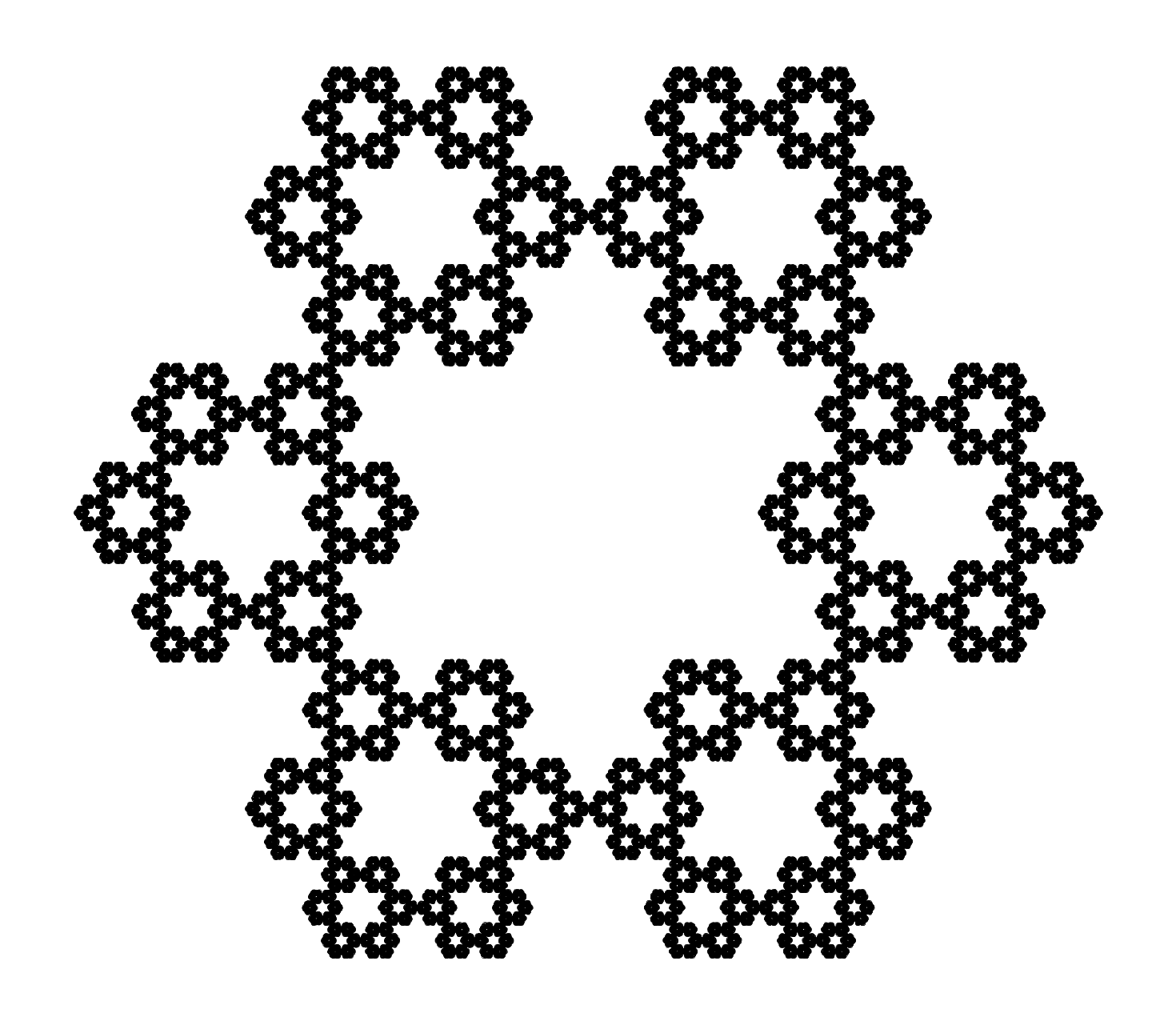}\\[-0.2em]
        {\scriptsize\bfseries (d)}
    \end{minipage}

    \vspace{0.7em}

    \begin{minipage}[b]{0.30\columnwidth}
        \centering
        \includegraphics[width=\linewidth]{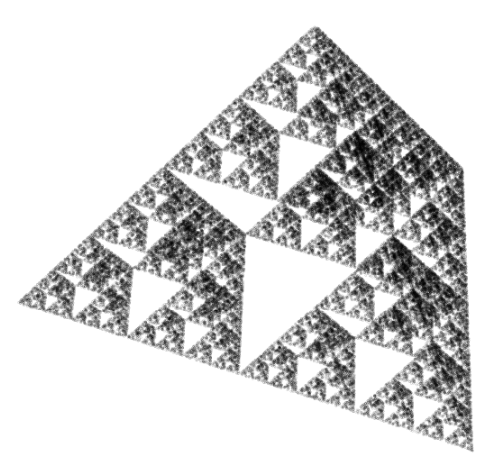}\\[-0.2em]
        {\scriptsize\bfseries (e)}
    \end{minipage}
    \hfill
    \begin{minipage}[b]{0.30\columnwidth}
        \centering
        \includegraphics[width=\linewidth]{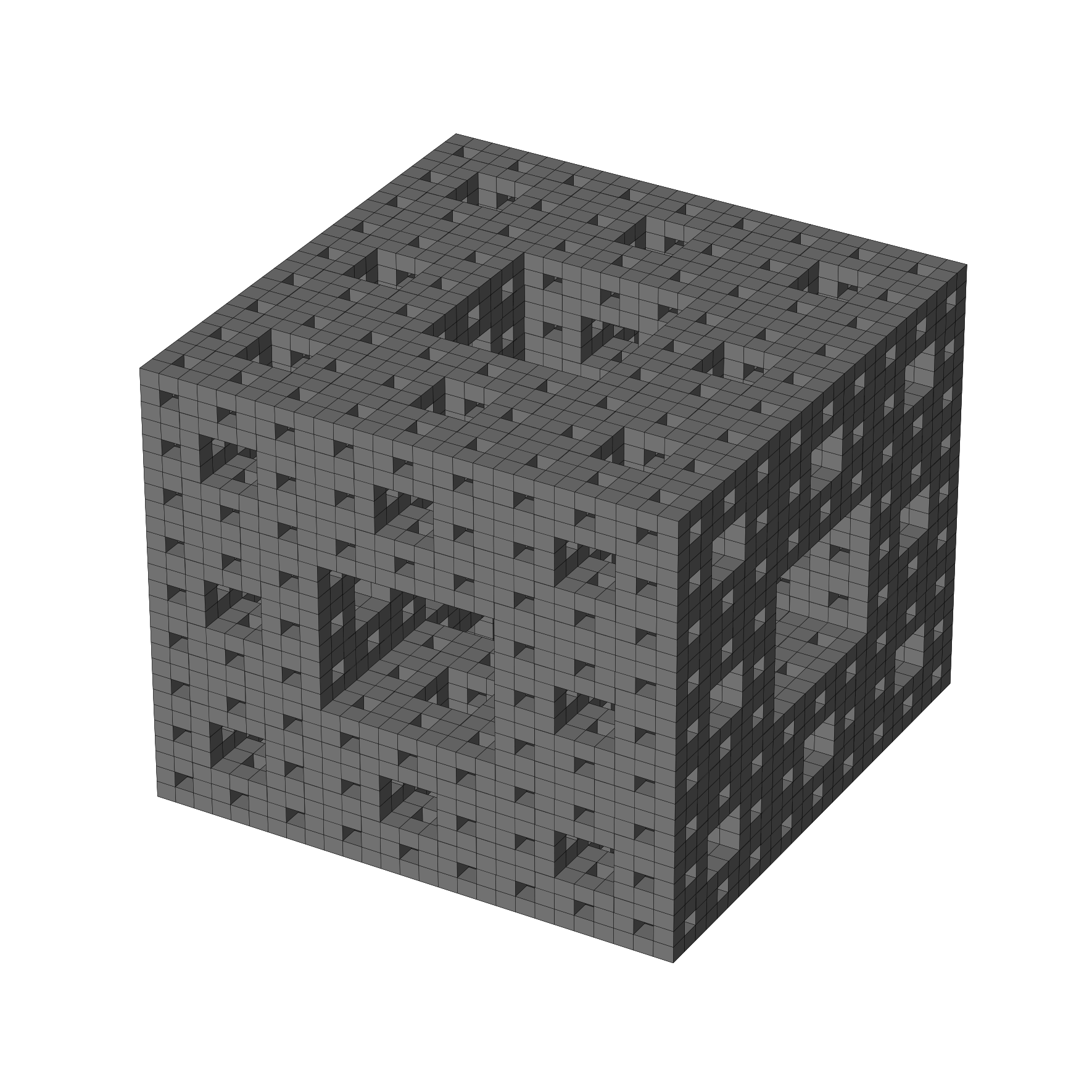}\\[-0.2em]
        {\scriptsize\bfseries (f)}
    \end{minipage}
    \hfill
    \begin{minipage}[b]{0.30\columnwidth}
        \centering
        \includegraphics[width=\linewidth]{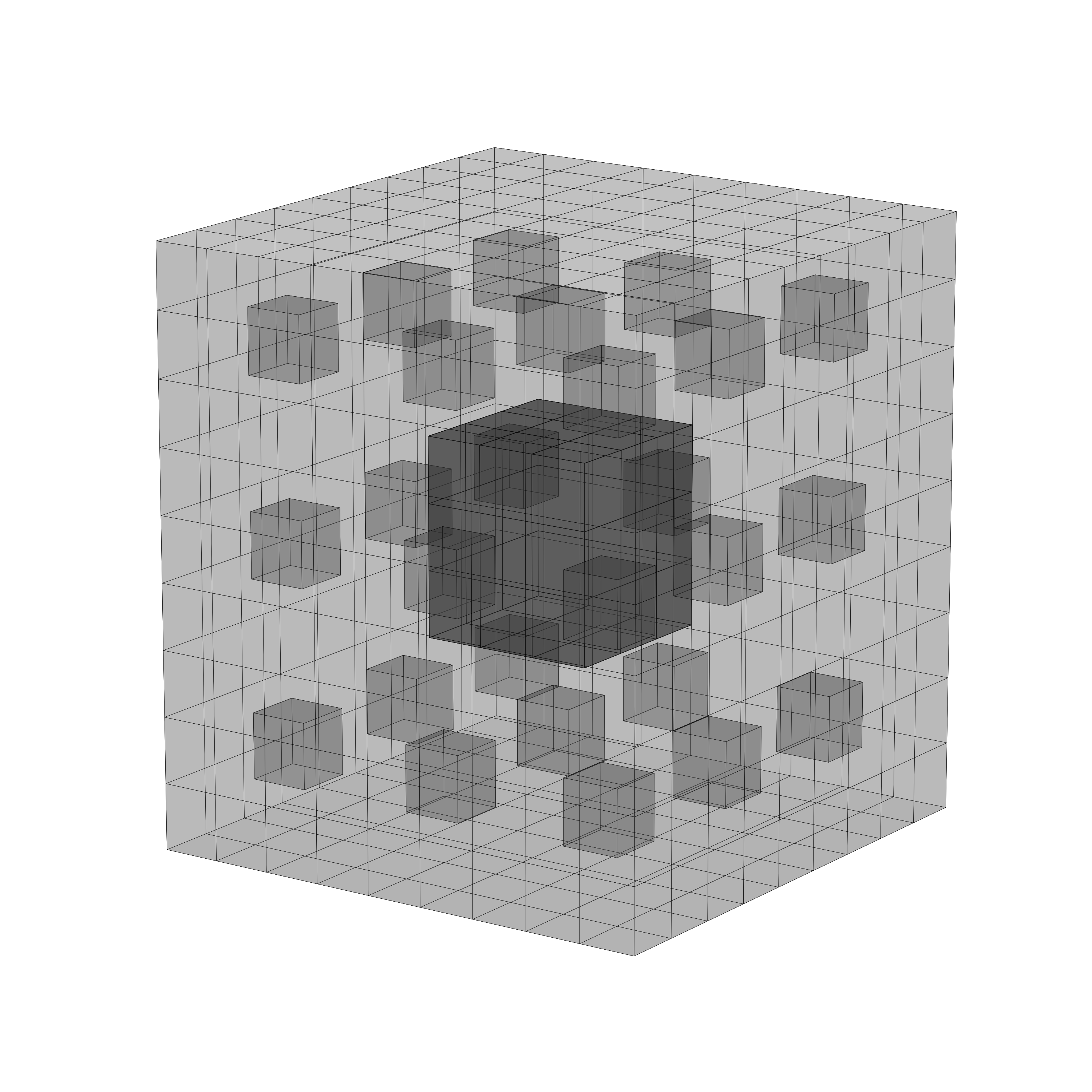}\\[-0.2em]
        {\scriptsize\bfseries (g)}
    \end{minipage}

    \vspace{0.9em}

    {\footnotesize
    \begin{tabular}{lc}
        \toprule
        \textbf{Fractal model} & \textbf{Assouad dimension $\beta$} \\
        \midrule
        (a) Sierpiński gasket   & $\log_2 3 \approx 1.58$ \\
        (b) Sierpiński carpet   & $\log_3 8 \approx 1.89$ \\
        (c) Dendrites           & $\log_3 5 \approx 1.46$ \\
        (d) Hexagasket          & $\log_3 6 \approx 1.63$ \\
        (e) Sierpiński simplex  & $\log_2 4 = 2$ \\
        (f) Menger sponge       & $\log_3 20 \approx 2.73$ \\
        (g) Cantor cube         & $\log_3 26 \approx 2.97$ \\
        \bottomrule
    \end{tabular}
    }

    \caption{
    Representative self-similar quasi-convex fractal models and their Assouad dimensions.
    }
    \label{fig:fractalset}
\end{figure}

The fractional 3D toric code~\cite{Zhu_2022,Dua_2023} is an example of a fractal code with macroscopic distance and noninteger intrinsic dimension. Using our construction outlined above, this code can be generated by the \emph{Cantor cube}, a sponge-like fractal formed by iteratively puncturing the central core in $\mathbb{R}^3$. However, in general it can be difficult to guarantee the distance of the constructed code as it can depend on the specific triangulation on the manifold and analyzing the construction can be highly nontrivial. 

\subsubsection{Multivariate multicycle (group algebra) codes} 

Motivated by recent progress on qLDPC code constructions with improved finite-size performance and increasingly flexible experimental architectures~\cite{Bravyi_2024_highthreshold,Berthusen_2025,Eberhardt_2024_BB_logical,Shaw_2025_morphing}, we discuss stabilizer codes defined by polynomials over quotient rings. Here we focus on code families with finite long-range connectivity, corresponding to defining polynomials of uniformly bounded degree. The case of unbounded-degree polynomials, where the interaction range grows with system size, will be treated separately in the next example. This class includes bivariate bicycle (BB) codes~\cite{Liang_2025,jacob2026singleshotdecodingfaulttolerantgates,Bravyi_2024}, multivariate bicycle (MB) codes~\cite{ll5p-z88p}, and trivariate tricycle (TT) codes~\cite{Menon_2026}. For such families, the intrinsic dimension is $D$, where $D$ corresponds to the dimension of the lattice.

To analyze the intrinsic dimension in a unified framework, we use the formulation for MM codes~\cite{mianMultivariateMulticycleCodes2026}. Consider a code family with unbounded system size $l_1,l_2,\cdots,l_D\to +\infty$ and translation-invariant stabilizers with bounded range $\mathrm{diam}(\mathrm{S_i})\leq \eta$. Qubit locations are identified by monomials in the quotient ring $\bm{x}^{\bm{v}}\in\mathbb{F}_2\left[x_1,x_2,\cdots,x_{D}\right]/(x_1^{l_1}-1,x_2^{l_2}-1,\cdots,x_D^{l_D}-1)$. The stabilizers are defined by finitely many terms with bounded degree $\eta$:
\begin{equation}
    \mathrm{S}_i: f_i\vdef\sum_{\bm{v}\in\left[-\eta,\eta\right]^{D}}c_{\bm{v}}\bm{x}^{\bm{v}},c_{\bm{v}}\in \{0,1\}; i=1,\cdots,t.
\end{equation}
 Here $\mathbb{F}_2\left[x_1^{\pm1},x_2^{\pm1},\cdots,x_{D}^{\pm1}\right]$ is naturally identified with an infinite ambient lattice $\mathbb{Z}^{D}$. The monomials that appear in $f_1+f_2+\cdots+f_t\in \mathbb{Z}\left[x_1^{\pm1},x_2^{\pm1},\cdots,x_{D}^{\pm1}\right]$ with non-zero coefficients correspond to the vertices connected to the origin by a stabilizer generator $\mathrm{S}_i$. Therefore, on the connectivity graph, the qubits adjacent to the origin $\mathcal{O}$ can be denoted by the set of differences between indices of monomials: \begin{equation}\Gamma=\{\bm{x}^{\bm{v}}\mid\bm{v}=\bm{a}-\bm{b};\bm{a},\bm{b}\in \supp(f_i) \text{ for some }i\},\end{equation} where the set of directions of adjacent vertices is denoted by $V=\{\bm{v}_1,\bm{v}_2\cdots,\bm{v}_r\},\bm{v}_i\in \left[-2\eta,2\eta\right]^{D}\subseteq \mathbb{Z}^D$ and $r$ is the degree of the origin in the connectivity graph. 
 
Since the connectivity graph is undirected and any pair of opposite directions is contained in $V$, the volume of the ball
centered at the origin is given by
\begin{equation}
    \mathrm{Vol}(N_{O}(R))=\sharp\left\{
        \sum_{i=1}^r n_i\bm v_i
        \mid n_i\in\mathbb Z,\ \sum_{i=1}^r|n_i|\leq R
    \right\}
\end{equation}
 on the infinite ambient lattice $\mathbb{Z}^D$. By the structure theorem for $\mathbb{Z}$-Mod, there exist a basis $\widetilde{\bm{v}}_1,\widetilde{\bm{v}}_2,\cdots,\widetilde{\bm{v}}_{r^{\prime}}$ such that
\begin{align}\label{Decomposition_stabilizers}
    &F\vdef\mathbb{Z}\langle\bm{v}_1,\bm{v}_2,\cdots,\bm{v}_r\rangle\nonumber\\
    &\simeq\mathbb{Z}\langle\widetilde{\bm{v}}_1\rangle\oplus\mathbb{Z}\langle\widetilde{\bm{v}}_2\rangle\oplus\cdots \oplus\mathbb{Z}\langle\widetilde{\bm{v}}_{r^{\prime}}\rangle\subseteq \mathbb{Z}^{D}.
\end{align}
 Observe that $r^{\prime}=\mathrm{rank}_{\mathbb{Z}}F\leq \mathrm{rank}_{\mathbb{Z}}\mathbb{Z}^D=D$, and the cardinality of the residual torsion part $F_{\mathrm{Torsion}}\coloneqq |\mathrm{Tor}(\mathbb{Z}^D/F)|$ is bounded by some constant $M_0$ depending on the range $\eta$ and $D$. In standard contexts~\cite{Liang_2025,jacob2026singleshotdecodingfaulttolerantgates,Bravyi_2024,mianMultivariateMulticycleCodes2026,Menon_2026}, the vectors generate the whole lattice , thus $F$ spans a full-rank sublattice of $\mathbb{Z}^D,r^{\prime}=D$. We show that after the basis change $(\bm{v}\to\widetilde{\bm{v}})$, the volume growth of $N_{O}(R)$ only differs from that of a standard ball in $\mathbb{Z}^D$ by a finite factor depending on the determinant of the transition matrix. By Eq.~(\ref{Decomposition_stabilizers}), the vectors $\bm{v}_i$ can be represented in the basis $\widetilde{\bm{v}_j}$ using a transition matrix $A_{(\bm{v}\to\widetilde{\bm{v}})} $
\begin{equation}\label{MMcode_coordinate_change}
     \begin{pmatrix}
         \bm{v}_1\\
         \bm{v}_2\\
         \vdots\\
         \bm{v}_r\\
     \end{pmatrix}=\mathrm{A}_{(\bm{v}\to\widetilde{\bm{v}})} 
     \begin{pmatrix}
         \widetilde{\bm{v}}_1\\
         \widetilde{\bm{v}}_2\\
         \vdots\\
         \widetilde{\bm{v}}_{D}\\
     \end{pmatrix},
 \end{equation}
  where $\mathrm{A}_{(\bm{v}\to\widetilde{\bm{v}})}=(a_{i,j})\in \mathrm{M}_{r\times D}(\mathbb{Z})$ depending on $\eta$ and $D$. The entries in $\mathrm{A}_{(\bm{v}\to\widetilde{\bm{v}})}$ are controlled by $\max_{i,j}|a_{i,j}|\leq M_1$. We use hypercubes in the connectivity graph to bound the ball $N_O(R)$:
    \begin{equation}
    \left[-\frac{R}{r},\frac{R}{r}\right]_{\mathbb Z}^{r}
    \bm{v}\subseteq N_{O}(R)\subseteq[-R,R]_{\mathbb Z}^{r}\bm v .
\end{equation}
  For any $\bm{v}=\sum\limits_{i=1}^{r}R_i\bm{v}_i, (R_1,R_2,\cdots,R_r)\in \left[-R,R\right]^{r}$, its coordinates in $\widetilde{\bm{v}}_i$ are given by $(R_1,R_2,\cdots,R_r)\mathrm{A}_{(\bm{v}\to\widetilde{\bm{v}})}\in \left[-rM_1R,rM_1R\right]^{D}$. Considering the enlarged hypercube, the volume of $N_O(R)$ has an upper bound given by
  \begin{align}
      \mathrm{Vol}(N_{O}(R))&\leq \sharp\{\left[-R,R\right]^{r}\bm{v}\}\leq \sharp\{\left[-rM_1R,rM_1R\right]^{D}\}\nonumber\\
      &=(2rM_1R+1)^{D}= O(R^{D}).
  \end{align}
  For the lower bound, since $\mathrm{rank}_{\mathbb{Z}}F=D$, after relabeling the generators if necessary, we may assume that $\bm v_1,\ldots,\bm v_D$ are linearly independent over $\mathbb{Q}$. Equivalently, the first $D$ rows of
 $\mathrm A_{(\bm v\to\widetilde{\bm v})}$ form a matrix $\mathrm A_1\in\mathrm M_{D\times D}(\mathbb Z)$ satisfying $\det(\mathrm A_1)\neq0$. Now we consider the hypercube
 \begin{equation}
     Q_R\coloneqq \left\{\sum_{i=1}^D n_i\bm v_i \mid n_i\in\mathbb{Z}, |n_i|\leq \frac{R}{D}\right\}
 \end{equation}
 Every element in $Q_R$ can be reached from the origin $O$ by a path of length at most $\sum\limits_{i=1}^D|n_i|\leq D\cdot \frac{R}{D}=R$, and hence $Q_R\subseteq N_{O}(R)$. Moreover, $\det(\mathrm A_1)\neq0$ implies that the map $(n_1,\ldots,n_D)\longmapsto \sum\limits_{i=1}^Dn_i\bm{v_i}$
 is injective. Therefore, 
 \begin{equation}\label{lower_bound_MM_code}
    \mathrm{Vol}(N_{O}(R))\geq|Q_R|=\left(\frac{2R}{D}+1\right)^D=\Omega(R^D).
\end{equation}
  The Ahlfors $D$-regularity implies that the MM code family associated with $\{f_i\}$ or equivalently with $V$ has intrinsic dimension $D$. Its code parameters may benefit from a ``twist'' by algebraic relation in the finite regime, but they are asymptotically constrained by the intrinsic code parameter bound (Theorem~\ref{Intrinsic BPT bound} in the next section). 
  

Then, we discuss the intrinsic dimensions for MM code families when their corresponding polynomials can have unbounded degree. A key result is that if each stabilizer generator contains at most $M$ monomials and there are $t$ generator types, then the intrinsic dimension is uniformly bounded by
$\beta\leq t\binom{M}{2}=\frac{M(M-1)}{2}t$.

Given a family of MM codes $\{C_{\lambda}\}_{\lambda\in \mathbb{N}}$ whose stabilizer generators are defined by finitely many polynomial terms with unbounded degree $\eta_{\lambda}$ over $F_{\lambda}\vdef\mathbb{F}_2\left[x_1,x_2,\cdots,x_{D}\right]/(x_1^{l^{\lambda}_1}-1,x_2^{l^{\lambda}_2}-1,\cdots,x^{l^{\lambda}_D}-1)$, 
\begin{equation}
    \mathrm{S}^{\lambda}_i: f^{\lambda}_i\vdef\sum_{\bm{v}\in\left[-\eta_{\lambda},\eta_{\lambda}\right]^{D}}c^{\lambda}_{\bm{v}}\bm{x}^{\bm{v}},c_{\bm{v}}\in \{0,1\},i\in\left[t\right].
\end{equation}
 where for every $ \lambda\in \mathbb{N}^{+}$, each $\mathrm{S}^{\lambda}_i$ has uniformly bounded connectivity $ \sharp\{c^{\lambda}_{\bm{v}}\neq 0| \bm{v}\in \left[-\eta_{\lambda},\eta_{\lambda}\right]^{D}\}\leq M$. We denote the set of edges corresponding difference of monomials by $\Gamma_{\lambda}=\{\bm{x}^{\bm{v}^{\lambda}_i}|i=1,2,\cdots,m_{\lambda},\bm{v}_i^{\lambda}\in\left[-2\eta_{\lambda},2\eta_{\lambda}\right]^{D} \}, |\Gamma_{\lambda}|\leq t\binom{M}{2}$. Thus, at most $m_{\lambda}\leq t\binom{M}{2}$ terms appear in the adjacency polynomial $f^{\lambda}\coloneqq f^{\lambda}_1+f^{\lambda}_2+\cdots+f^{\lambda}_t=\sum\limits_{i=1}^{m_{\lambda}}c_i \bm{x}^{\bm{v}^{\lambda}_i}\in\mathbb{Z}\left[x_1^{\pm1},x_2^{\pm1},\cdots,x_{D}^{\pm1}\right]$. The connectivity graph $G_{\lambda}$ is exactly the Cayley graph $\mathrm{Cay}(\langle\text{Monomials in }F_{\lambda}\rangle,\Gamma_{\lambda})$, where the group action is given by multiplication by monomials $\bm{x}^{\bm{v}^{\lambda}}$ in $\Gamma_{\lambda}$. 
 
 \begin{remark}
     The volume growth of a Cayley graph is closely related to the algebraic structure, or more precisely, the nilpotency of the underlying group. This has been a central topic in geometric group theory for over sixty years~\cite{Bass1972,10.4310/jdg/1214428658,gromov1981groups,trofimov1985graphs,shalom2010finitary,breuillard2014diophantine}.
 \end{remark} 
 
 In our scenario, $F_{\lambda}$ is an abelian group generated by $D$ elements. We will prove that $\mathrm{dim}(C_{\lambda\in I})\leq t\binom{M}{2}<+\infty$.
 By definition, the volume of $N_{O}(R)$ in the connectivity graph $G_{\lambda}$ is bounded by the number of terms in $({f^{\lambda})^R+(f^{\lambda}})^{R-1}+\cdots+1=\sum c_{\bm{v}}x^{\bm{v}}, \bm{v}=R_1\bm{v}_1+R_2\bm{v}_2+\cdots+R_{m_{\lambda}}\bm{v}_{m_{\lambda}}, 0\leq R_1+R_2+\cdots+R_{m_{\lambda}}\leq R$. Thus a basic combinatorial counting argument implies 
 \begin{equation}\label{Combinatorial_upper_growth_bound}
     \sharp N_{O}(R)\leq \binom{R+m_{\lambda}}{m_{\lambda}}= O(R^{m_{\lambda}}).
 \end{equation}
 The Theorem 2 in Ref.~\cite{Bass1972} and Theorem 3.2 in Ref.~\cite{10.4310/jdg/1214428658} prove that, for a finitely generated abelian group $\Gamma$, there exists an integer $\beta\in \mathbb{Z}$, such that $c_{*}R^{\beta}\leq \mathrm{Vol}(N_{O}(R))\leq c^{*}R^{\beta}$. Equivalently, the family $\{C_{\lambda}\}_{\lambda\in \mathbb{Z}}$ has Ahlfors $\beta$-regularity and intrinsic dimension $\beta\in \mathbb{Z}$. Comparing this with Eq.~(\ref{Combinatorial_upper_growth_bound}), we obtain the upper bound $\beta\leq m_{\lambda}\leq t\binom{M}{2}$. Note that we cannot apply the argument based on Eq.~(\ref{Decomposition_stabilizers}) here, since the transition matrix between basis $\mathrm{A}^{\lambda}_{(\bm{v}\to\widetilde{\bm{v}})}$ may depend on the unbounded range $\eta_{\lambda}$.
 \begin{figure*}[htbp]
    \centering
        \includegraphics[width=\textwidth]{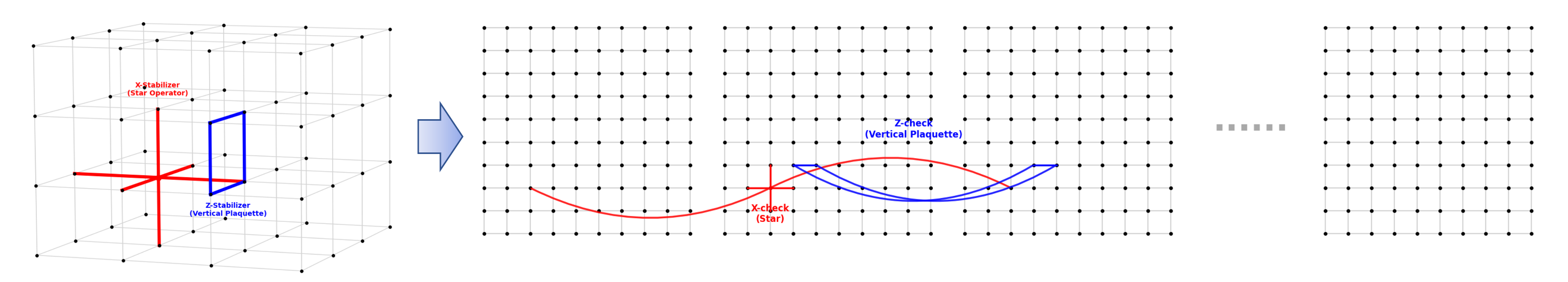}
        \caption{\raggedright\small Illustration of the mapping from 3D toric code to 2D planar code with long range connectivity: the $X$-star check $1+x^{-1}+x+y^{-1}+y+z^{-1}+z\in \mathbb{F}_2\left[x,y,z\right]/(x^{L}-1,y^{L}-1,z^{L}-1)$ maps to $1+x^{-1}+x+y^{-1}+y+x^{-L}+x^{L}\in\mathbb{F}_2\left[x,y\right]/(x^{L^2}-1,y^{L}-1)$, and the vertical $Z$-plaquette check $1+x+z+xz$ maps to $1+x+x^{L}+x^{L+1}$.}
        \label{expanded3D}
\end{figure*}

 Once long-range connectivity from polynomials with unbounded degrees is allowed, the intrinsic dimension $\beta$ can be an arbitrary integer value. Consider the translation invariant $\mathrm{Z}$-check $1+x+x^{L}+x^{L^2}+\cdots+x^{L^{\beta-1}}$ on $\mathbb{F}_2\left[x\right]/(x^{L^{\beta}}-1)$. It has at least $\beta$ effective directions $1,L,\ldots,L^{\beta-1}$ for $R<\log{L}\ll L$, therefore the non-overlapping $L$-adic expression indicates
 \begin{align}
     &\mathrm{Vol}(N_{O}(R))\geq \sharp\{R_0+R_1L+\cdots+R_{\beta-1}L^{\beta-1}\mid\nonumber\\&R_0+R_1+\cdots+R_{\beta-1}\leq R, 0\leq R_i\leq R\}\nonumber\\&=\binom{R+\beta}{\beta}= \Omega(R^{\beta}).
 \end{align}
 On the other hand, a path of length $R$ can be decomposed as an alternating sum of $L^i$ with total absolute weight less than $2R$, then 
 \begin{align}
     &\mathrm{Vol}(N_O(R))\nonumber\\
     &\leq \sharp\left\{n_0+n_1L+\cdots+n_{\beta-1}L^{\beta-1}\mid \sum_{i=1}^{\beta-1}|n_i|\leq 2R\right\}\nonumber\\
     &\leq (4R+1)^{\beta}= O(R^{\beta}).
 \end{align}
 So its volume growth satisfies the Ahlfors $\beta$-regularity at scale $R<\log L$.
 
 Another explicit example is the expanded 3D toric code on the bivariate quotient ring $\mathbb{F}_2\left[x,y\right]/(x^{L}-1,y^{L^2}-1)$ (See Fig.~\ref{expanded3D}), which can be seen as flattening one spatial direction of the 3D toric code. The local $X$-star (vertical $Z$-plaquette) stabilizers in the 3D toric code become stretched stabilizers with macroscopic range $L$ across consecutive 2D layers. Since it has the same connectivity graph, its intrinsic dimension remains unchanged as $\beta =3$.

It is worth specifically noting that the logical gates supported by MM constructions are of interest in current developments in fault tolerance. Given the above analysis of their intrinsic dimensions, we can apply the intrinsic logical gate bounds in Section~\ref{sec:intrinsic dimension bridges code symmetry and indistinguishability} to  systematically study this problem. In particular, TT codes can support fault-tolerant non-Clifford gates, a feature that has been leveraged in recent work to improve fault tolerance~\cite{jacob2026singleshotdecodingfaulttolerantgates,Menon_2026}, while BB codes even with modest long-range connectivity cannot have non-Clifford gates. 
 
 \subsubsection{Hypergraph product codes} 
 We now turn to code constructions that are algebraic in nature. As the prime example, we consider hypergraph product code families and show that the intrinsic dimension is at most $\beta_1 +\beta_2$, where $\beta_i$ is the intrinsic dimension of its constituent classical codes $\{C_i\}$.
 
 Given two families of classical error correcting codes whose Tanner graph families $\mathcal{X}_1,\mathcal{X}_2$ have Ahlfors $\beta_1$-regularity and $\beta_2$-regularity respectively, the Tanner graph family of the hypergraph product code is defined by the product of the Tanner graphs of the constituent classical codes $\widetilde{\mathcal{X}}=\{\widetilde{G_{\lambda}}\vdef G_{1,\lambda}\times G_{2,\lambda},\lambda\in I\}$. For any two vertices $q=(v_1,v_2),q^{\prime}=(v_1^{\prime},v_2^{\prime})$ in the product graph $\widetilde{G}=G_1\times G_2$, their distance in the new graph is given by $d_{G}(q,q^{\prime})=d_{G_1}(v_1,v_1^{\prime})+d_{G_2}(v_2,v_2^{\prime})$. Therefore, we have 
 \begin{equation}\label{product_space_shrinking}
     N_{v_1}\left(\frac{R}{2}\right)\times N_{v_2}\left(\frac{R}{2}\right)\subseteq N_{q}(R) \subseteq N_{v_1}(R)\times N_{v_2}(R).
 \end{equation}
 Taking the volume and applying the Ahlfors regularity of $G_1,G_2$, we obtain 
 \begin{equation}
     \frac{v_{1,*}v_{2,*}}{2^{\beta_1+\beta_2}}R^{\beta_1+\beta_2}\leq \mathrm{Vol}(N_q(R))\leq v_1^{*}v_2^{*}R^{\beta_1+\beta_2}.
 \end{equation} Therefore, the Tanner graph family of the derived quantum codes has Ahlfors $(\beta_1+\beta_2)$-regularity. 
 
 In general, if the component codes have intrinsic dimensions $\dimasd\mathcal{X}_i=\beta_i,i=1,2$ but are not known to satisfy Ahlfors regularity, then we can only obtain a weaker lower and upper bound:
 \begin{equation}
     \max\{\beta_1,\beta_2\}\leq \dimasd \widetilde{\mathcal{X}}\leq \beta_1+\beta_2.
 \end{equation}
 We show that the upper bound holds. First, by~(\ref{product_space_shrinking}), $N_{q}(R) \subseteq N_{v_1}(R)\times N_{v_2}(R)$. Each $N_{v_i}(R)$ can be covered by $O((\frac{R}{r})^{\beta_i})$ balls of radius $r$. Thus, any $N_q(R)$ can be covered by $O((\frac{R}{r})^{\beta_1+\beta_2})$ balls of radius $r$ in $\widetilde{\mathcal{X}}$. Thus, a hypergraph product code has intrinsic dimension at most $\beta_1 +\beta_2$.

 \subsubsection{``Good'' qLDPC and optimal finite-dimensional codes}

First, known asymptotically  ``good'' quantum low-density parity-check (qLDPC) codes~\cite{panteleev2022asymptotically,leverrier2022quantum,dinur2023good} which achieve the optimal scaling $\llbracket n, k=\Omega(n), d=\Omega(n) \rrbracket$ rely on expanding connectivity and hence have infinite intrinsic dimension.
To see this, recall that the Cheeger constant $h(G)$ for the connectivity graph $G \in \mathcal{X}$ is strictly lower bounded:
\begin{equation}
    h(G) = \min_{0 < |S| \leq \frac{|G|}{2}} \frac{|\partial S|}{|S|} \geq \epsilon > 0.
\end{equation}
Let $S = N_x(r)$ denote a ball of radius $r$ on the connectivity graph. Assuming $|S| \leq \frac{|G|}{2}$, the volume of the enlarged ball $N_x(r+1)$ satisfies:
\begin{align}
    \mathrm{Vol}(N_x(r+1)) &= |S \cup \partial S| \nonumber \\
    &\geq (1+\epsilon)|S| = (1+\epsilon)\mathrm{Vol}(N_x(r)). 
\end{align}
Iterating this relation gives exponential volume growth, $\mathrm{Vol}(N_x(r)) \geq C (1+\epsilon)^r$. This implies that the Assouad dimension is infinite, i.e., $\dim_{\mathrm{A}}(\mathcal{X}) = +\infty$, which indicates that the intrinsic dimension is infinite. This is consistent with the intrinsic code parameter bound (Theorem~\ref{Intrinsic BPT bound}) in the subsequent section.

Moreover, we note that code parameters arbitrarily close to asymptotically good can be achieved in finite dimensions.  Recently, Lin {et al.}~\cite{lin2023geometrically,Li_2024_GeometricallyLocal} constructed the qLDPC code families with parameters $\llbracket n, \Omega(n^{\frac{D-2}{D}}), \Omega(n^{\frac{D-1}{D}}) \rrbracket$ defined on arbitrary $D$-dimensional integer lattices. 
Therefore, for any arbitrarily small \(\epsilon>0\), one can achieve 
\(k=\Omega(n^{1-\epsilon})\) with \(D=2/\epsilon\), and 
\(d=\Omega(n^{1-\epsilon})\) with \(D=1/\epsilon\). 
That is, the code parameters can be made arbitrarily close to linear in finite dimensions.

\subsubsection{Non-qLDPC codes}
Finally, for generic non-qLDPC code families, the connectivity graphs have unbounded degree: either some stabilizer generators involve $ \omega(1)$ qubits, or a sequence of physical qubits are involved in $\omega(1)$ checks. Hence, the ball of constant radius can contain an unbounded number of vertices. This violates the doubling condition (Proposition~\ref{doubling_condition}) at scale $R=\frac12$, and therefore the intrinsic dimension is infinite. This does not, however, mean that this case is uninteresting: for example, as will be explained in Section~\ref{sec:intrinsic dimension bridges code symmetry and indistinguishability}, the infinite intrinsic dimension feature of non-qLDPC codes opens the possibility of fault-tolerant logical gates at arbitrarily high levels of the Clifford hierarchy, which is indeed achieved recently~\cite{Golowich_2024,He_2025}.

\section{intrinsic dimension constrains code parameters}\label{sec:intrinsic dimension constrains code parameters}

In this section, we establish a general tradeoff bound for the basic rate and distance parameters dictated by the intrinsic dimension. For integer dimensions, our bound recovers the known parameter tradeoffs~\cite{Bravyi_2010}, while extending and strengthening them for more general code families.
 
 Previous work by Baspin et al.~\cite{Baspin_2022} shows that the \emph{separation profile}, a function of graph separators, constrains the code parameters. For a graph family $\mathcal{X}=\{G_{{\lambda}}\}_{{\lambda}\in \mathbb{N}}$ with Assouad dimension $\beta$, one can bound the separation profile as $s_n(r)=O(r^{1-\frac{1}{\beta}})$ by identifying the separator between a maximal ball of volume $r=\Theta(R^{\beta})$ and the remainder of $G_{\lambda}$ with a boundary sphere of volume $O(R^{\beta-1})$. Applying Corollary 24 in Ref.~\cite{Baspin_2022} then yields a bound $k \cdot d^{\frac{2}{\beta}} \leq O(n)$ on the code parameters. However, as noted by the authors, this bound is suboptimal without more specific geometric conditions. Leveraging our  strengthened regularity condition that controls the growth of the boundary of a ball, we obtain an improved distance bound, further optimizing the exponent of $d$ from $\frac{2}{\beta}$ to $\frac{2}{\beta-1}$. The resulting bound can be seen as a generalization of the Bravyi--Poulin--Terhal bound~\cite{Bravyi_2010} based on intrinsic dimension.
 
\begin{definition}[Strengthened Ahlfors $\beta$-regularity]
    A family of graphs $\mathcal{X}=\{G_{\lambda}\}_{\lambda\in I}$ exhibits \emph{strengthened Ahlfors $\beta$-regularity} up to scale $R_{\lambda}$ if the volume of any sphere in $G_{\lambda}$ with radius $R\leq R_{\lambda}$ is uniformly bounded by
    \begin{equation}
        v_{*} R^{\beta-1} \leq \mathrm{Vol}(\partial N_{q}(R)) \leq \max\{v^{*}R^{\beta-1}, 1\}, \quad \forall \lambda \in I,
    \end{equation}
    where $\partial N_{q}(R) \coloneqq \{p \mid d_{\lambda}(p,q)=R\}$ denotes the boundary sphere centered at $q$. In particular, summing over $R$ implies standard Ahlfors $\beta$-regularity, thus confirming $\dimasd(\mathcal{X})=\beta$ up to scale $R_{\lambda}$.
\end{definition}
 This condition can be understood as requiring the Assouad dimension to be $\beta$ along with second-order spatial homogeneity. It holds naturally under translation invariance or self-similarity at local scale, where the constant $v_{*}$ ($v^{*}$) is chosen as the minimum (maximum) over a finite set.
\begin{theorem}[Intrinsic code parameter bound]\label{Intrinsic BPT bound}

Let \(\{C_{\lambda}(\mathcal{S}_{\lambda})\}_{\lambda\in I}\) be a code family with specified stabilizer generator sets, and
let \(\mathcal{X}=\{G_{\lambda}\}_{\lambda\in I}\) be the associated family of connectivity graphs, with intrinsic dimension $\beta\geq2$.
If $G_{\lambda}$ satisfies strengthened Ahlfors $\beta$-regularity up to scale $d_{\lambda}$, then 
the following quantitative constraints on code parameters holds:
    \begin{equation}
        {k_{\lambda}d_{\lambda}^{\frac{2}{\beta-1}}}\leq \frac{27\beta\cdot48^{\beta}(v^*)^{\frac{2\beta}{\beta-1}}}{2v_{*}^2}\cdot n_{\lambda}, ~\forall \lambda\in I.
    \end{equation}
More simply put, there exists a constant $w$ depending only on $\beta,v_{*},v^{*}$ such that $w{k_{\lambda}d_{\lambda}^{\frac{2}{\beta-1}}}\leq n_{\lambda}$ universally holds.
\end{theorem}

This not only improves previous rate-distance tradeoff scalings and extends them to arbitrary dimensions, but also strengthens them beyond asymptotic scaling by setting explicit constants in the bounds.

Beyond the asymptotic exponent, the finite-size prefactors in geometric rate-distance tradeoffs are themselves of interest code design. In two dimensions, searches over BB codes and related generalized toric and planar constructions have increased the reported values of $kd^{2}/n$ by optimizing stabilizer geometry and boundary conditions~\cite{Bravyi_2024,Liang_2025,Liang_2025_planar,Liang_2026_selfdual}. A representative $n=360$ BB instance was reported with $kd^{2}/n=19.2$~\cite{Liang_2025}. Our Theorem~\ref{Intrinsic BPT bound} establishes general limitations also for the prefactors: for any code family with fixed intrinsic dimension $\beta\geq2$ and fixed regularity constants $(v_{*},v^{*})$, the generalized index $k_{\lambda}d_{\lambda}^{\frac{2}{\beta-1}}/n_{\lambda}$ remains uniformly bounded. Consequently, under mild regularity assumptions, there is a fundamental limit to how far this index can be improved. We expect that the bound can be further tightened by sharper combinatorial divisions and more refined estimates in our proof. Determining the optimal constant in physically relevant dimensions, and thereby closing the gap between constructive lower benchmarks and the general analytical upper bound, is an interesting direction for future study.

Moreover, as a direct implication:

\begin{remark}
    Theorem~\ref{Intrinsic BPT bound} directly indicates that under the regularity assumptions, no code families with finite intrinsic dimension can simultaneously have positive asymptotic rate and growing distance.
\end{remark}

This gives a precise mathematical sense in which expansion is not merely a pathology of many high-parameter qLDPC codes: nongeometric connectivity growth beyond any finite intrinsic dimension is necessary to support a good $k$ together with any growing distance.


Also interestingly:
\begin{remark}
    It is shown in Ref.~\cite{baspin2025stabilizercodesdimensionsconstant} that when $\beta<2$, the code distance is bounded by a constant independent of the number of physical qubits~. This result completes the theoretical picture of dimensional constraints, confirming that $\beta = 2$ is the critical dimension for asymptotic distance growth.
\end{remark}

We now proceed to the proof. We first state a lemma which is needed for the proof. This lemma constructs a cover of $G_\lambda$ with geodesic balls whose boundaries have controlled pairwise intersections. This enables us to build a tripartite partition $\{\mathcal{A},\mathcal{B},\mathcal{C}\}$ and in turn to prove an upper bound on the volume of region $\mathcal{C}$. 
\begin{widetext}
    
    \begin{lemma}\label{Division lemma}
    Let $G_{\lambda}$ be a connectivity graph with strengthened Ahlfors $\beta$-regularity up to scale $R_{\lambda}$. For any scale $R\in \left(0,\frac{R_{\lambda}}{6}\right]$, there always exists a set of centers $\mathcal{E}_{\lambda}=\{x_{1},x_{2},\cdots,x_{l(\lambda)}\}\subseteq G_{\lambda}$ and a sequence of integer radii $(t_1,t_2,\cdots,t_{l(\lambda)})\in \left[R+1,2R\right]^l$ such that
    \begin{itemize}
     \item The collection of balls $N_{x_i}(t_i),i=1,2,\cdots,l(\lambda)$ constitutes a cover of $G_{\lambda}$, $l(\lambda)\leq w_0 \cdot \frac{n}{R^{\beta}}$, where $w_0$ is a constant that only depends on $\beta,v_{*}$.
     \item For any $i\in \left[l\right]$, there are at most $M$ of $j\in \left[l\right]$ such that  $N_{x_i}(2R+1)\cap N_{x_j}(t_j)\neq \emptyset$. The volume of the intersection of the boundaries satisfies 
     \begin{equation}
         \sum\limits_{i,j=1,i\neq j}^{l}\sum\limits_{\delta_i=0,\pm 1}\mathrm{Vol}(\partial N_{x_i}(t_i+\delta_i)\cap \partial N_{x_j}(t_j))\leq w_2\cdot \frac{n}{R^2},
     \end{equation}
     where $w_2,M>0$ are constants that only depend on $\beta,v_{*},v^{*}$.
    \end{itemize}
    The factor $\delta_i$ is included here to control the behavior of the outer and inner boundaries of the domains $\partial_{+}N_{x_j}(t_j),\partial_{-}N_{x_j}(t_j), \forall j\in \left[l\right]$.
\end{lemma}
\begin{proof}
    Consider a maximal disjoint collection of balls with radius $\frac{R}{2}$ on $G_{\lambda}$. The set of centers $\mathcal{E}_{\lambda}=\{x_1,x_2,\cdots,x_{l(\lambda)}\}$ is a $R$-net satisfying the $R$-separation condition: $d(x_i,x_j)> R$ for all $i\neq j$. The number $l(\lambda)$ is controlled by 
    \begin{equation}
        l=l(\lambda)\leq \frac{n}{\mathrm{Vol}(N_{x_i}(\frac{R}{2}))}= \frac{n}{\sum\limits_{s=0}^{\frac{R}{2}}\mathrm{Vol}(\partial N_{x_i}(s))}\leq \frac{2^{\beta}\beta}{v_{*}}\frac{n}{R^{\beta}}\eqqcolon w_0\frac{n}{R^{\beta}}.
    \end{equation}
    Note that for any integer $t_j\in\left[R+1,2R\right]$, $N_{x_i}(2R+1)$ intersecting with $N_{x_j}(t_j)$ implies $d(x_{i},x_j)\leq 2R+1+t_j\leq4R+1\leq 5R$. However, by the $R$-separation condition $d(x_{j^{\prime}},x_j)> R$, the disjoint union of balls centered at $x_j$ with radius $\frac{R}{2}$ is contained in $N_{x_i}(6R)$. Hence, the number of such $x_j$ is bounded by the following volume estimation:
    \begin{equation}
        \sharp|\{j|N_{x_i}(2R+1)\cap N_{x_j}(t_j)\}\neq \emptyset|\leq \frac{\mathrm{Vol}(N_{x_i}(6R))}{\mathrm{Vol}(N_{x_j}(\frac{R}{2}))}\leq \frac{v^{*}}{v_{*}}12^{\beta}\eqqcolon M,\quad \forall t_j \in \left[R+1,2R\right].
    \end{equation}
    For the next property, we make use of the pigeonhole principle and  double counting. We consider the following auxiliary quantity:
    \begin{equation}
        \zeta(\mathbf{t})=\zeta(t_1,t_2,\cdots,t_{l(\lambda)})\coloneqq\sum_{i,j=1,i\neq j}^{l} \sum_{\delta_i=0,\pm 1}\mathrm{Vol}(\partial N_{x_i}(t_{i}+\delta_i)\cap \partial N_{x_j}(t_j)).
    \end{equation}
Summing over integer vector $\mathbf{t}\in \left[R+1,2R\right]^{l(\lambda)}$, we derive an upper bound:
\begingroup
\allowdisplaybreaks[4]
\begin{align}
    \sum_{\mathbf{t}\in \left[R+1,2R\right]^{l}} \zeta(\mathbf{t})=&\sum_{\mathbf{t}\in \left[R+1,2R\right]^{l}}\sum_{i,j=1,i\neq j}^{l} \sum_{\delta_i=0,\pm 1}\mathrm{Vol}(\partial N_{x_i}(t_{i}+\delta_i)\cap \partial N_{x_j}(t_j))\\
    =&\sum_{i=1}^l\sum_{\mathbf{t}\setminus \{t_i\}\in \left[R+1,2R\right]^{l-1}}\left[\sum_{t_{i}=R+1}^{2R}\sum_{j=1,j\neq i}^{l}\sum_{\delta_i=0,\pm1} \mathrm{Vol}(\partial N_{x_i}(t_{i}+\delta_i)\cap \partial N_{x_j}(t_j))\right]\\
    =&\sum_{i=1}^l\sum_{\mathbf{t}\setminus \{t_i\}\in \left[R+1,2R\right]^{l-1}}\left[\sum_{j=1,j\neq i}^{l}\sum_{\delta_i=0,\pm1}\sum\limits_{t_{i}=R+1}^{2R}\mathrm{Vol}(\partial N_{x_i}(t_{i}+\delta_i)\cap \partial N_{x_j}(t_j))\right]\\
    <&\sum_{i=1}^l\sum_{\mathbf{t}\setminus \{t_i\}\in \left[R+1,2R\right]^{l-1}}\left[\sum_{\{j|j\neq i, N_{x_i}(2R+1)\cap N_{x_j}(t_j)\neq \emptyset\}} \sum_{\delta_i=0,\pm1}\mathrm{Vol}(N_{x_i}(2R+\delta_i)\cap \partial N_{x_{j}}(t_{j}))\right]\\
    \leq&\sum_{i=1}^l\sum_{\mathbf{t}\setminus \{t_i\}\in \left[R+1,2R\right]^{l-1}} \left(M\cdot 3\max\limits_{1\leq j\leq l}\{\mathrm{Vol}(\partial N_{x_j}(t_j))\}\right)\\
    \leq&\sum_{i=1}^l\sum_{\mathbf{t}\setminus \{t_i\}\in \left[R+1,2R\right]^{l-1}}\left(3Mv^{*}(2R)^{\beta-1}\right) =l\cdot 3Mv^{*}2^{\beta-1}R^{l+\beta-2}\\
    \leq&(3w_0Mv^{*}2^{\beta-1})R^{l-2}n=\frac{3\beta\cdot48^{\beta}(v^*)^2}{2v_{*}^2}R^{l-2}n.
\end{align}
\endgroup
By the pigeonhole principle for $\mathbf{t}\in \left[R+1,2R\right]^l$, there exists a sequence of integer radii $\mathbf{t}_0=(t_1,t_2,\cdots,t_l)$ such that $\zeta(\mathbf{t}_0)\leq (\frac{3\beta\cdot48^{\beta}(v^*)^2}{2v_{*}^2})\frac{n}{R^2}\eqqcolon w_2\cdot\frac{n}{R^2}$. After the radii are determined, we conclude that $\sum\limits_{i,j=1,j\neq i}^{l}\sum\limits_{\delta_i=0,\pm 1}\mathrm{Vol}(\partial N_{x_i}(t_{i}+\delta_i)\cap \partial N_{x_j}(t_{j}))\leq w_2 \frac{n}{R^2}$. Finally, we discard any possible redundant ball when some $N_{x_i}(t_i)\subseteq \bigcup\limits_{j\neq i}N_{x_j}(t_j).$ 
\end{proof}
Having established the volume estimation for the boundaries of the domain, we are ready to prove the intrinsic code parameter bound.
\begin{proof}[Proof of Theorem~\ref{Intrinsic BPT bound}]
Given a stabilizer code $C_{\lambda}$ and its connectivity graph $G_{\lambda}$, we first construct a tripartite partition for $G_{\lambda}=\mathcal{A}\bigsqcup \mathcal{B}\bigsqcup \mathcal{C}$, where regions $\mathcal{A}$ and $\mathcal{B}$ are correctable. Once we establish such a partition, the standard entropy argument then gives $k \leq|\mathcal{C}|$, from which we can generalize the Bravyi--Poulin--Terhal argument to obtain our result.  

For the scale $R=\frac13(\frac{d_{\lambda}}{v^{*}})^{\frac{1}{\beta-1}}< \frac{d_{\lambda}}{6}~(\text{when } \beta=2,\text{ we require } v^{*}> 2)$, there exists a cover $\{N_{x_i}(t_i),i=1,\cdots ,l\leq w_0\frac{n_{\lambda}}{R^{\beta}}\}$ for $G_{\lambda}, R+1\leq t_i\leq2R$ such that it satisfies the properties in Lemma~\ref{Division lemma}. We define the sign function by $\operatorname{sgn}(x)\vdef 
\begin{cases}
-1 & \text{if } x < 0, \\
0  & \text{if } x = 0, \\
1  & \text{if } x > 0.
\end{cases}$  Since $\{N_{x_i}(t_i),i=1,\cdots,l\}$ is a cover for $G_{\lambda}$, each $x\in G_{\lambda}$ has a unique configuration determined by sign functions
\begin{equation}
    \mathrm{config}(x)=(\mathrm{sgn}(t_1-|x-x_1|),\mathrm{sgn}(t_2-|x-x_2|),\cdots,\mathrm{sgn}(t_l-|x-x_l|))\in \{-1,0,+1\}^{l},
\end{equation} where $|x-x_i|$ stands for the length of shortest path between $x$ and $x_i$ on $G_{\lambda}$. Observe that $G_{\lambda}$ is covered by the union of balls in Lemma~\ref{Division lemma}, so there exists no vertex $x$ with $\mathrm{config}(x)=(-1,-1,\cdots,-1).$

The following partition is determined by this configuration (See Fig.~\ref{fig:demo_BPT_division} for an illustration):

\begin{align}
    &\mathcal{C}=\mathop{\bigcup}_{i,j=1,i\neq j}^{l}\mathop{\bigcup}_{\delta_i\in \{0,\pm1\}} \{\partial N_{x_j}(t_j)\cap \partial N_{x_i}(t_i+\delta_i)\};\\
    &\mathop{\bigsqcup}_{\text{At least two 0s in config } \vec{f}} \bigg\{x\in G_{\lambda}| \mathrm{config}(x)=\vec{f}\bigg\}=\mathop{\bigcup}_{i,j=1,i\neq j}^{l} \{\partial N_{x_j}(t_j)\cap \partial N_{x_i}(t_i)\}\subseteq \mathcal{C}\nonumber,\\
    &\mathcal{A}=\mathop{\bigsqcup}_{\text{odd +1 in config }\vec{f}\in \{-1,1\}^l} \left(\bigg\{x\in G_{\lambda}| \mathrm{config}(x)=\vec{f}  \bigg\}\backslash\mathcal{C}\right)=\mathop{\bigsqcup}_{\text{odd +1 in config }\vec{f}\in \{-1,1\}^l} A_{\vec{f}},\\
    &\mathcal{B}=\mathop{\bigsqcup}_{\text{even +1 in config }\vec{f}\in \{-1,1\}^l} \left(\bigg\{x\in G_{\lambda}|\text{ config}(x) \text{ differs from }  \vec{f} \text{ by at most one 0} \bigg\}\backslash\mathcal{C}\right)\nonumber\\
    &\quad =\mathop{\bigsqcup}_{\text{even +1 in config }\vec{f}\in \{-1,1\}^l}B_{\vec{f}}.
\end{align}

\begin{figure}[t]
    \centering

    \begin{minipage}[b]{0.31\columnwidth}
        \centering
        \includegraphics[width=\linewidth]{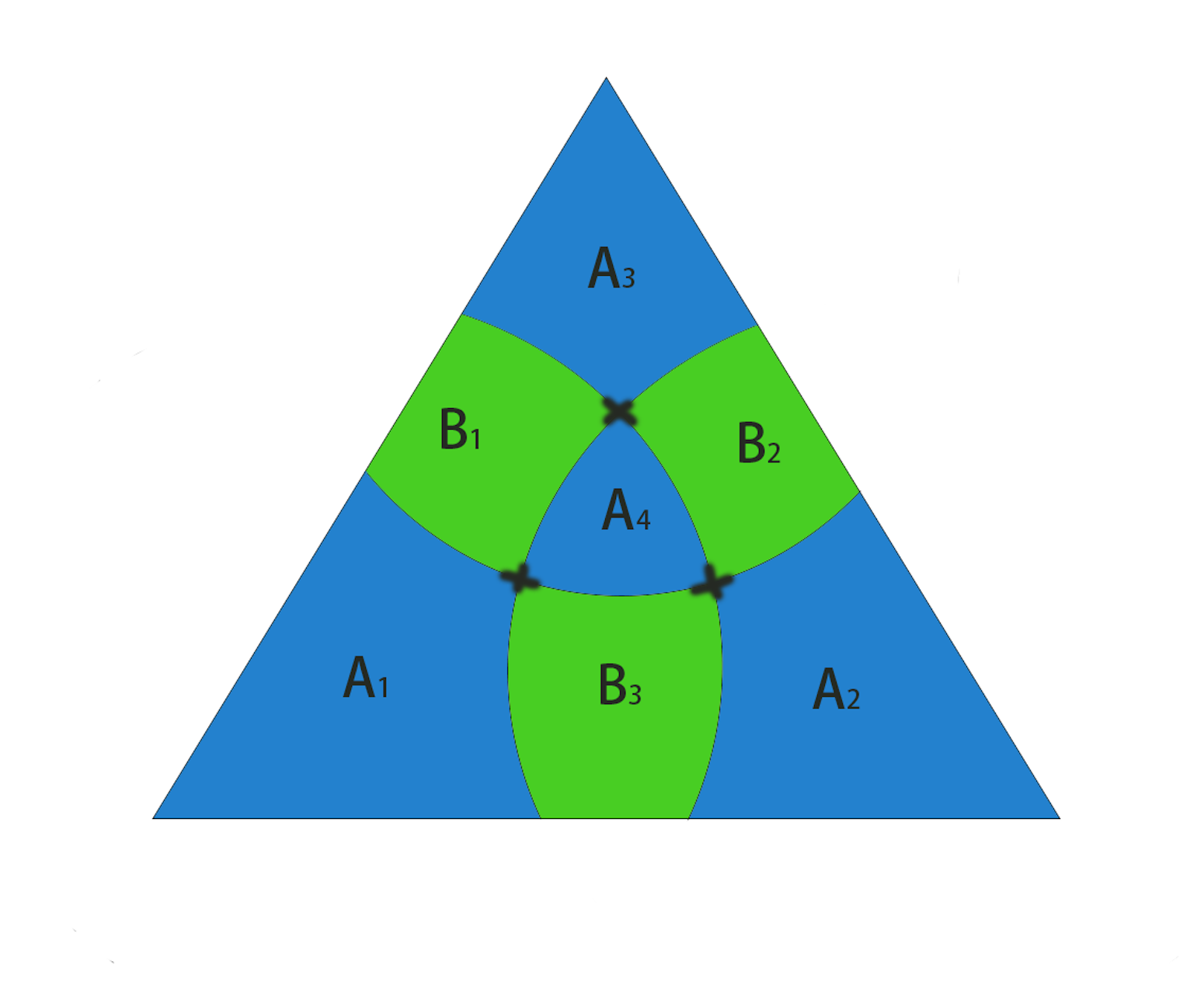}\\[-0.2em]
        {\scriptsize\bfseries (a)}
    \end{minipage}
    \hfill
    \begin{minipage}[b]{0.31\columnwidth}
        \centering
        \includegraphics[width=\linewidth]{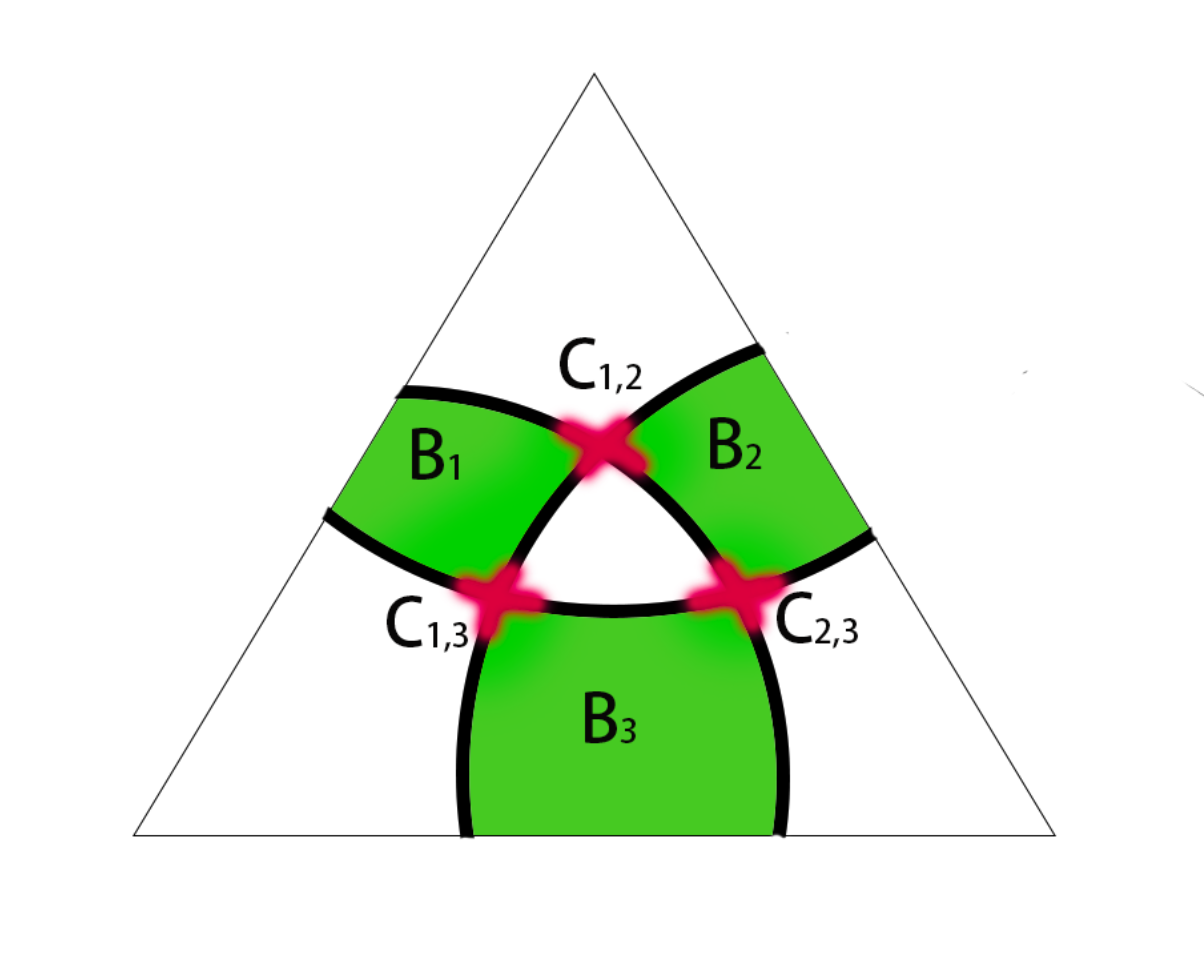}\\[-0.2em]
        {\scriptsize\bfseries (b)}
    \end{minipage}
    \hfill
    \begin{minipage}[b]{0.31\columnwidth}
        \centering
        \includegraphics[width=\linewidth]{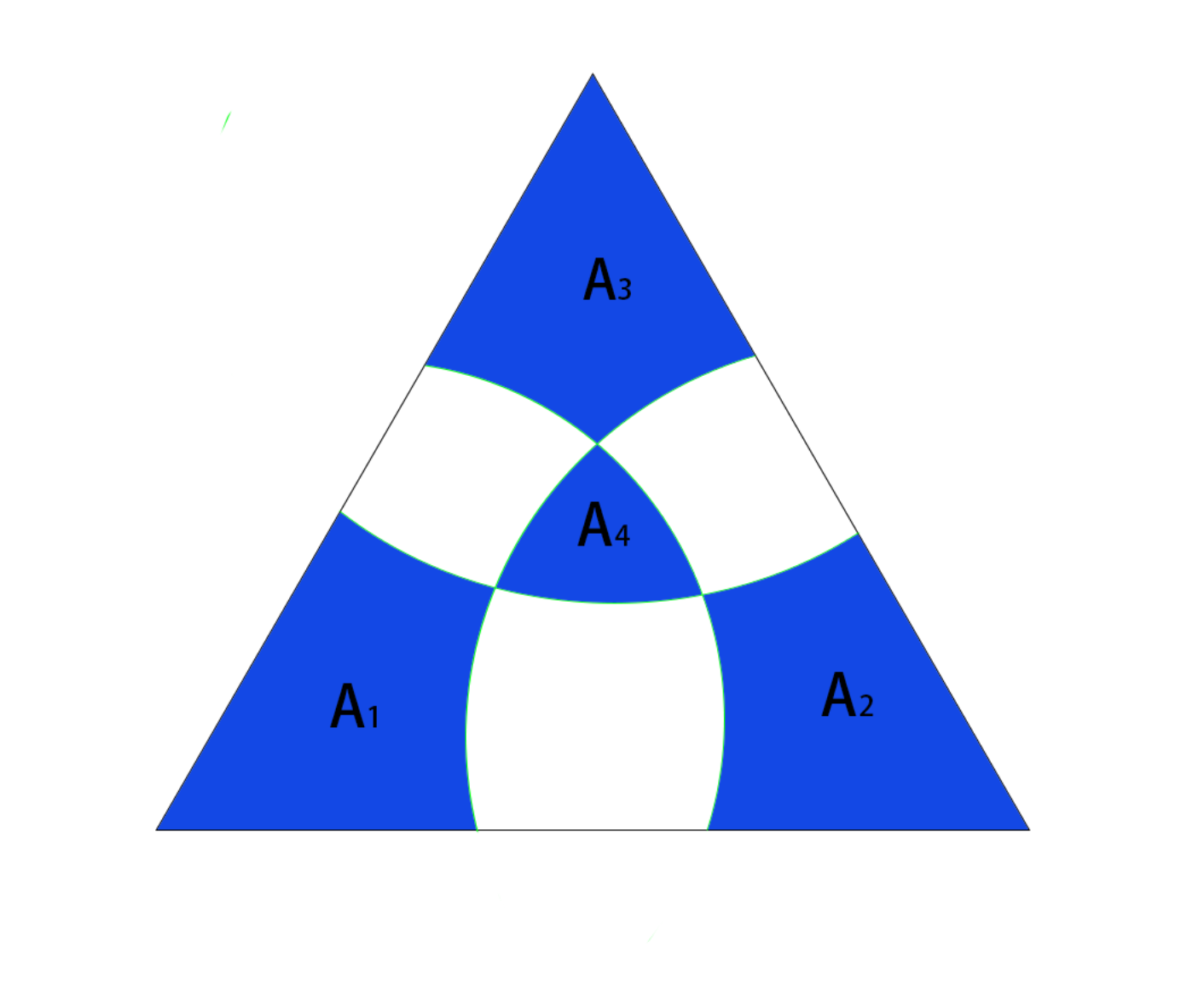}\\[-0.2em]
        {\scriptsize\bfseries (c)}
    \end{minipage}

    \caption{
    Illustration of the generalized BPT partition.
    (a) A tri-partition for a triangle.
    (b) $\mathcal{C}$ is the union of the thickened red crosses. The black interfaces are absorbed into the green regions
    $\mathcal{B}=\bigcup_{\vec f} B_{\vec f}$.
    (c) $\mathcal{A}=\bigcup_{\vec f'} A_{\vec f'}$ is the union of the blue regions without the interfaces.
    }
    \label{fig:demo_BPT_division}
\end{figure}

We now show that every component region making up $\mathcal{A}$ and $\mathcal{B}$ is correctable.
Notice that any ball $N_{x}(t_i)$ can be viewed as a consecutive expansion along the radial direction $r=1,2,\cdots,t_i$. The volume of boundary satisfies $\mathrm{Vol}(\partial N_x(r))\leq v^{*}(2R)^{\beta-1}=\left(\frac{2}{3}\right)^{\beta-1}d_{\lambda}< d_{\lambda}$. By Expansion Lemma~\ref{Expansion Lemma}, $N_x(t_i)$ is always a correctable set. As a subset of certain $N_{x_i}(t_i)$, $A_{\vec{f}}$ or $B_{\vec{f}}$ is still a correctable set, where $\vec{f}$ is required not to be $(-1,-1,\cdots,-1)$ for $B_{\vec{f}}$. Now we consider the case of $\vec f_{-}=(-1,\ldots,-1)$. Since the balls form a cover,
the all-minus configuration does not occur. Hence every point in
$B_{\vec f_{-}}$ has a unique zero coordinate. Define $B_{\vec f_{-},i}
=\{x\in B_{\vec f_{-}}:
\operatorname{config}_i(x)=0\}$, then 
$B_{\vec f_{-}}=\bigsqcup_i B_{\vec f_{-},i}, B_{\vec f_{-},i} \subseteq \partial N_{x_i}(t_i),$ so each $B_{\vec f_{-},i}$ is correctable. Suppose that $x\in B_{\vec f_{-},i}$ and
$y\in B_{\vec f_{-},j}$ are adjacent, where $i\neq j$.
Since $\operatorname{config}_i(y)=-1$, adjacency implies
$d(y,x_i)=t_i+1$. Therefore,$
y\in\partial N_{x_i}(t_i+1)\cap\partial N_{x_j}(t_j)\subseteq\mathcal C$, 
which is a contradiction. Thus the sets
$B_{\vec f_{-},i}$ are pairwise $2$-separated, and the Union
Lemma implies that $B_{\vec f_{-}}$ is correctable.

It remains to use the Union Lemma to show that $\mathcal{A}$ and $\mathcal{B}$ are both correctable. Observe that the outer boundary of $A_{\vec{f}},B_{\vec{f}}$ is contained piecewise within $\mathop{\bigcup}\limits_i\partial N_{x_i}(t_i+\delta_i)$. Any two correctable sets $A_{\vec{f}},A_{\vec{f}^{'}}\in\mathcal{A}$ are separated by some elements in $\mathcal{C}$ or $\mathcal{B}$. Indeed, the configuration of the outer boundary satisfies $\partial_{+}A_{\vec{f}}\subseteq \mathop{\bigcup}\limits_i\partial N_{x_i}(t_i) \subseteq\mathcal{B}\bigsqcup\mathcal{C}, \partial_{+}A_{\vec{f}}\cap A_{\vec{f}^{'}}\subseteq (\mathcal{B}\bigsqcup\mathcal{C})\cap\mathcal{A}=\emptyset$. The Union Lemma~\ref{Union Lemma} shows that $\mathcal{A}=\bigsqcup A_{\vec{f}}$ is a correctable set. Similarly for the family of sets $\mathcal{B}$, we prove its elements are mutually separated by contradiction. Given two disjoint correctable sets $B_{\vec{f}}\cap B_{\vec{f}^{'}}=\emptyset, \vec{f}\neq \vec{f}^{\prime}$, we suppose $x\in B_{\vec{f}}$ and $y\in B_{\vec{f}^{\prime}}$ are adjacent. Since $x,y\notin\mathcal{C}$, each of their configurations contains at most one zero. Notice that the distance to each center changes by at most one along the extension of the boundary. Since  $\vec{f}$ and $\vec{f}^{\prime}$ both have even parity and are distinct, their Hamming distance is at least two. It follows that the configurations of $x$ and $y$ must have zeros at two distinct coordinates $i\neq j$ with centers $x_i,x_j$. Consequently, for some $\delta_i\in\{-1,0,1\}$,
\begin{equation}
    y\in \partial N_{x_i}(t_i+\delta_i)\cap\partial N_{x_j}(t_j)\subseteq\mathcal C,
\end{equation}
 since $x$ and $y$ are adjacent vertices. 
 This contradicts $y\in B_{\vec{f}^{\prime}}\subseteq G_\lambda\setminus\mathcal C$, hence
\begin{equation}
\partial _{+}B_{\vec{f}}\bigcap B_{\vec{f}^{'}}\subseteq\mathop{\bigcup}\limits_{i\neq j,\delta_i}(\partial N_{x_i}(t_i+\delta_i)\cap \partial N_{x_j}(t_j))\backslash \mathcal{C}=\emptyset.
\end{equation}
 By the Union Lemma~\ref{Union Lemma}, $\mathcal{B}$ is also a correctable set.

 For any state $\rho\in \mathcal{H}_{\lambda}\simeq(\mathbb{C}^2)^{\otimes n_{\lambda}}$ and region $A\subseteq G_{\lambda}$, $\mathcal{S}(A)\vdef -\mathrm{Tr}(\rho_{A}\log{\rho_{A}})$ is the von Neumann entropy of the reduced density matrix $\rho_{A}$ on region $A$. $\mathcal{S}(A|A^{\prime})\vdef \mathcal{S}(AA^{\prime})-\mathcal{S}(A^{\prime})$ is the conditional entropy. Recall that $\mathcal{S}(A|\overline{A})=-\mathcal{S}(A)$ when $A$ is a correctable region (Fact 1 in Ref.~\cite{Bravyi_2010}). Consider the maximally mixed state $\rho_{\mathrm{max}}$ in the code space, we may estimate the logical dimension of code space by calculating its von Neumann entropy on the entire space:
\begin{align}
    &k_{\lambda}=\mathcal{S}(\mathcal{A}\mathcal{B}\mathcal{C})=\mathcal{S}(\mathcal{B}\mathcal{C})+\mathcal{S}(\mathcal{A}|\mathcal{B}\mathcal{C})\leq \mathcal{S}(\mathcal{B})+ \mathcal{S}(\mathcal{C})- \mathcal{S}(\mathcal{A}),\\
     &k_{\lambda}=\mathcal{S}(\mathcal{A}\mathcal{B}\mathcal{C})=\mathcal{S}(\mathcal{A}\mathcal{C})+\mathcal{S}(\mathcal{B}|\mathcal{A}\mathcal{C})\leq \mathcal{S}(\mathcal{A})+ \mathcal{S}(\mathcal{C})- \mathcal{S}(\mathcal{B}).
\end{align}
Adding the above two inequalities and dividing by 2 yields
\begin{equation}
     k_{\lambda}\leq \mathcal{S}(\mathcal{C})\leq |\mathcal{C}|\leq \sum_{i,j=1,j\neq i}^{l}\sum_{\delta_i=0,\pm1}\mathrm{Vol} (\partial N_{x_j}(t_j)\cap \partial N_{x_i}(t_i+\delta_i)).
\end{equation}
Recall from Lemma~\ref{Division lemma} that the total volume of the intersection of the boundaries is bounded by 
\begin{equation}
    k_{\lambda}\leq w_2 \frac{n_{\lambda}}{R^2}=(\frac{3\beta\cdot48^{\beta}(v^*)^2}{2v_{*}^2}) \frac{n_{\lambda}}{(\frac13(\frac{d_{\lambda}}{v^{*}})^{\frac{1}{\beta-1}})^2}=\frac{27\beta\cdot48^{\beta}(v^*)^{\frac{2\beta}{\beta-1}}}{2v_{*}^2}\cdot\frac{n_{\lambda}}{d_{\lambda}^{\frac{2}{\beta-1}}}.
\end{equation}
 Then we conclude that $k_{\lambda}\cdot d_{\lambda}^{\frac{2}{\beta-1}}\leq \frac{27\beta\cdot48^{\beta}(v^*)^{\frac{2\beta}{\beta-1}}}{2v_{*}^2} \cdot n_{\lambda}$.
\end{proof}
\end{widetext}

The asymptotic growth rate in Theorem~\ref{Intrinsic BPT bound} is known to be tight when $\beta\in \mathbb{Z}$. Indeed, a series of recent results~\cite{portnoy2023localquantumcodessubdivided,lin2023geometrically,Li_2024_GeometricallyLocal,Williamson_2024,liAlmostOptimalGeometrically2026a} demonstrate constructions of qLDPC codes embedded in $\mathbb{Z}^{D}$ with optimal code parameters $\llbracket n, \Omega(n^{\frac{D-2}{D}}), \Omega(n^{\frac{D-1}{D}}) \rrbracket$. Nonetheless, the tightness of the bound for noninteger dimension $\beta\notin \mathbb{Z}$ remains open.

We now illustrate the consequences of the intrinsic code parameter bound Theorem~5 through representative code families.

In Refs.~\cite{Zhu_2022,Dua_2023}, the iterated family of fractional 3D toric codes $\{\mathrm{FC}(a,b,l),l\geq 1\}$ has intrinsic dimension $\beta=\log_a{(a^3-b^3)}>2,b\leq a-1$ and code parameters $\llbracket n_{l},k_{l}, d_{l}\rrbracket=\llbracket (a^3-b^3)^l=a^{\beta l},3, a^{l}\rrbracket$. These parameters can be shown to satisfy the constraint in Theorem~\ref{Intrinsic BPT bound}:
\begin{align}
    &\frac{k_{l}\cdot d_{l}^{\frac{2}{\beta-1}}}{n_{l}}=3a^{(\frac{2}{\beta-1}-\beta)l}\leq 3a^{(\frac{2}{\beta-1}-\beta)},\forall l\in \mathbb{Z},\\
    &\text{ where } \frac{2}{\beta-1}<2<\beta =\log_a{(a^3-b^3)}. \nonumber
\end{align}
For MM codes with constant range, we can also apply Theorem~\ref{Intrinsic BPT bound} to give a general set of bounds on their code parameters.
Let the code family $\{C_{\lambda}\}_{\lambda\in \mathbb{N}}$ be defined by a series of polynomials with bounded degree on a quotient ring with $D$ variables, with the system size $\lambda\to +\infty$. $\mathrm{Vol}(N_{O}(R))$ equals to the volume of the ball spanned by product of monomials in $\Gamma$ on $\mathbb{Z}^{D}$. Existing literature shows that the growth series of finitely generated abelian group is rational (see e.g.~Theorem 1.2 in Ref.~\cite{benson1983growth} and Chapter 6 of  Ref.~\cite{de2000topics}), and can be expanded in the following form:
\begin{align}
    \Sigma(z)&=\sum_{R=0}^{+\infty}\mathrm{Vol}(\partial N_{O}(R))z^{R}=\frac{P(z)}{(1-z)^{m}}\\
    &=P(z)\cdot \sum_{i=0}^{+\infty}\binom{i+m-1}{m-1}z^i\\
    &= \sum_{R=0}^{+\infty} \left(\sum_{j=0}^k p_j \binom{R-j+m-1}{m-1}\right)z^R,
\end{align}
where $P(z)=\sum_{j=0}^k p_j z^j$ is a polynomial of constant degree $k$ in $z$, and $p_0=1$. Comparing the coefficient before $z^R$, we deduce that $\mathrm{Vol}(\partial N_{O}(R))$ is a polynomial of certain integer degree $d$ (which is consistent with Theorem 1 in Ref.~\cite{khovanskii1992newton}). Combined with the previous estimation $\mathrm{Vol}(N_{O}(R))=\Omega(R^D)$ from Eq.~($\ref{lower_bound_MM_code}$), we conclude that $\mathrm{Vol}(\partial N_{O}(R))$ is a polynomial of degree $D-1$, satisfying the strengthened Ahlfors $D$-regularity condition. Therefore, our Theorem~\ref{Intrinsic BPT bound} applies, so $\frac{{k_{\lambda}d_{\lambda}^{\frac{2}{D-1}}}}{n_{\lambda}}$ is uniformly upper bounded by the constant $\frac{27D\cdot48^{D}(v^*)^{\frac{2D}{D-1}}}{2v_{*}^2}$ depending on $D,v_{*},v^{*}$. 

\section{Intrinsic dimension bridges code symmetry and indistinguishability}\label{sec:intrinsic dimension bridges code symmetry and indistinguishability}

In this section, we discuss the connection between intrinsic dimension and logical gates.
For regular topological codes on integer-dimensional lattices, Bravyi and K\"onig showed that the achievable level in the Clifford hierarchy for constant-depth logical gates is constrained by the spatial dimension~\cite{Bravyi_2013}. 
Here, we show that the underlying idea of constructing suitable division patterns, mathematically related to the so-called Nagata dimension, formally generalizes the result and yields a highly versatile framework for understanding the possible fault-tolerant logical operations on general stabilizer codes. 
For instance, our theorems naturally accommodate tuning of the notion of locality and fault tolerance, and apply to extensive code families including fractal codes, codes enhanced by long-range connectivity, etc. 

To prove our results, a key technical ingredient is the \emph{Assouad--Nagata reduction}, discussed in Subsection~\ref{subsection: Nagata}, which gives a combinatorial division pattern at various length scales. 

\subsection{Clifford hierarchy}
For self-containedness, we first define the notion of Clifford hierarchy, originally introduced in works related to gate teleportation and gate injection~\cite{Gottesman_1999}, which is a particularly important language for studying fault-tolerant logical gates. 
It stratifies logical gates according to their action by conjugation on Pauli operators, and therefore naturally captures the algebraic structure underlying many fault-tolerant gate constructions and no-go theorems. 

The Clifford hierarchy is a nested tower of sets of unitary operators. 
The levels of the Clifford hierarchy can be defined recursively as follows.

\begin{definition}[Clifford hierarchy]
    For an $n$-qubit system, the Clifford hierarchy is a sequence of sets of unitaries
    $\{\mathcal P_\ell\}_{\ell\geq 1}$ in $\mathbb U(2^n)$ defined recursively as follows.
    The first level $\mathcal P_1$ is the $n$-qubit Pauli group. For $\ell\geq 2$, the $\ell$-th level is defined by
    \begin{align}
        \mathcal P_\ell
        =
        \left\{
        U\in \mathbb U(2^n)
        \,\middle|\,
        U P U^\dagger \in \mathcal P_{\ell-1}
        \ \text{for all } P\in \mathcal P_1
        \right\}.
    \end{align}
\end{definition}

 The Clifford group $\mathcal{P}_2$ can be generated by the Hadamard gate, Pauli gates, the phase gate $S=\sqrt{Z}$ and the CNOT gate. For $\ell>2$, $\mathcal{P}_{\ell}$ is not closed under multiplication and therefore is not a group. Two prominent examples of non-Clifford gates in the third level are the $T=\sqrt{S}$  gate and the $\ccz$ gate. More generally, $
Z_\ell := \operatorname{diag}\left(1,e^{2\pi i/2^\ell}\right)
$ and $\mathrm{C}^{(\ell-1)}Z$ are typical examples of gates in $\mathcal{P}_{\ell}$ for $\ell>2$.

\subsection{Graph freedom and graph-local gates}
In this subsection, we introduce the notion of graph-local gates. As a generalization of geometrically local gates, graph-local gates are defined with respect to a connectivity graph, with the requirement that errors propagate locally on the graph. 

Before turning to details, we first highlight the practical significance of the flexibility of the graph.
For a given physical architecture or experimental setup, one may choose an admissible set of stabilizer generators that are compatible with the available connectivity. In this way, we can adapt the fault-tolerant gate implementations and the constraints imposed on them to the specific physical scenario.
This graph adaptivity naturally encompasses broader gate implementation schemes such as fold-transversal or multi-block gates by augmenting the locality graph to reflect additional physically allowed connectivity.  More
generally, changing the locality graph changes both the operations regarded as local and the degree of geometric obstruction to fault-tolerant gates that we will show below. 
In particular, there is a potential tradeoff: increased connectivity may increase the intrinsic dimension and consequently allow a larger set of gates, but it also introduces additional pathways for error propagation. 
Our results below provide a mathematical understanding of the interplay between geometry and fault tolerance.


We now formally define the notion of graph-local gate set. The graph-local gate set naturally supplies fault-tolerant gates due to restricted error propagation, analogous to how geometrically-local gates are fault-tolerant in the setting of regular topological codes. Specifically, graph-local gates can only propagate errors within a constant radius on the connectivity graph, that is, the errors can be detected by at most a constant number of stabilizers that form a connected path within this ball of constant radius on the graph. 

 \begin{definition}[Graph-local gate set]

Let \(\{C_{\lambda}(\mathcal{S}_{\lambda})\}_{\lambda\in I}\) be a code family with specified stabilizer generator sets, and
let \(\mathcal{X}=\{G_{\lambda}\}_{\lambda\in I}\) be the associated family of connectivity graphs.
A gate set $\mathcal{G}$ is said to be \emph{graph-local} with respect to 
$\mathcal{X}$ if any two vertices within the support of a graph-local gate can be connected by a path with bounded length $\eta=O(1)$, independent of 
$\lambda$. In other words, for every $g\in \mathcal{G}$ and every operator $P$,
\[
\supp(U_g P U_g^\dagger)
\subseteq N_{G_\lambda}(\supp(P),2\eta).
\]
Here $N_{G_\lambda}(\mathcal{V},r)$ denotes the $r$-neighborhood of a set of vertices 
$\mathcal{V}$ in the graph $G_\lambda$.
\end{definition}

The concept of graph-local gates is naturally dependent on the choice of stabilizer generator set. For a given stabilizer code, it is possible to find two different sets of stabilizer generators, leading to two different connectivity graphs with different intrinsic dimension. As such, it is possible that their respective graph-local gates can implement different logical gates.

In practice, verifying that the graph-local condition holds can be reduced to checking whether a generating set of gates are all graph-local. For regular topological codes, this means checking that each physical gate is geometrically local.

\subsection{Nagata dimension and Assouad--Nagata reduction}~\label{subsection: Nagata}

We now proceed to show that the intrinsic dimension restricts logical action with graph-local structures. The first step requires the notion of Nagata dimension and applying the Assouad--Nagata reduction technique.

\begin{definition}[Nagata dimension, cf.~\cite{Assouad1982,Lang2005,Buyalo2007,le2015assouad}]\label{Nagata dimension}
    Let $\mathcal{X}=\{(G_{\lambda},d_\lambda)\}_{\lambda \in I}$ be a family of graphs indexed by $I$, where $d_{\lambda}$ is the path metric on $G_{\lambda}$. We say a graph $G_{\lambda}=(V_{\lambda},d_{\lambda})$ satisfies the $D$-division condition  if for any $r>0$, there always exists $D+1$ families of subsets of $V_{\lambda}$,  $\mathcal{B}_1=\mathop{\bigcup}\limits_{j_1} B_1^{(j_1)},\mathcal{B}_2=\mathop{\bigcup}\limits_{j_2} B_2^{(j_2)},\cdots,\mathcal{B}_{D+1}=\mathop{\bigcup}\limits_{j_{D+1}} B_{D+1}^{(j_{D+1})}$, which satisfy the conditions:
    \begin{itemize}
        \item $\mathop{\bigcup}\limits_{i=1}^{D+1}{ \mathcal{B}_i}=V_{\lambda}$,
        \item Each $\mathcal{B}_i$ is $r$-separated, i.e., any pair of sets $B_i^{(1)},B_i^{(2)}$ in $\mathcal{B}_i$ has distance $d_{\lambda}(B_{i}^{(1)},B_i^{(2)})\geq r$,
        \item $B\in \mathcal{B}_i$ is $c\cdot r$-bounded, i.e., each $B$ is contained in a ball with radius $c\cdot r$. 
    \end{itemize}
     The \emph{Nagata dimension} of $\mathcal{X}$ is the minimal integer $D\in \mathbb{N}$ such that the $D$-division conditions are satisfied for all $G_{\lambda},\lambda\in I$. Furthermore, we say $\mathcal{X}=\{(G_{\lambda},d_\lambda)\}_{\lambda \in I}$ has Nagata dimension $D$ up to scale $\overline{r}$, if $D$ is the minimal integer such that $D$-division conditions are satisfied for any $0<r<\overline{r}$. 
\end{definition}

Intuitively, the Nagata dimension represents the minimum number of colors required such that there exists a cover consisting of bounded balls and neighboring balls are assigned different colors at all scales (like the $D$-colex for  color codes~\cite{bombin2015gauge,kubica2015universal}).
The thickening of a $D$-skeleton formed from the triangulation of a $D$-dimensional Euclidean space naturally provides a $D$-division for the Nagata dimension.

We then formulate the following lemma, which converts an Assouad-type covering growth bound into a Nagata-type decomposition at comparable scales. 
In this sense, the lemma bridges intrinsic dimensional growth and the coarse decompositions needed to control locality.

\begin{lemma}[Assouad--Nagata reduction with scale]\label{Assouad--Nagata Reduction at scale $R$} Let $(X,d)$ be a metric space with finite Assouad dimension $\dim_{\mathrm{Assouad},R}(X) = \beta$ up to scale $R$, satisfying the covering condition with constant $C_R$. For any $\beta < \alpha < \lfloor\beta\rfloor + 1$, there exists a constant $C_\alpha$, depending only on $\alpha$ and $C_R$, such that $X$ has Nagata dimension at most $\lfloor\alpha\rfloor$ up to scale $r_{\mathrm{Nagata}} = C_\alpha R$.
\end{lemma}

We relegate more detailed information and full proof to Appendix~\ref{app:A-N reduction}, and provide key intuitions here. The goal is to construct a decomposition satisfying the conditions in Definition~\ref{Nagata dimension}. The main idea is to first choose a cover of the space by large balls. Inside each large ball, we select an annular shell of thickness $r$, lying between radii $R/2$ and $R$. One can choose this annulus so that it has effectively one lower dimension. Removing these annuli separates the remaining inner regions from one another. These inner regions then form a family of subsets: they are uniformly bounded and mutually $r$-separated. The removed annuli form a lower-dimensional residue region. We then repeat the same construction on this residue region. Each iteration lowers the effective dimension by one, and the process terminates once the remaining region becomes zero-dimensional. For a $\beta$-dimensional space, the iteration ends at $k = \lfloor \beta+1\rfloor$, proving that the Nagata dimension is at most $\lfloor\beta \rfloor$.

\begin{corollary} \label{Assouad-Nagate bound} 
    Let $\mathcal{X}=\{(G_{\lambda},d_{\lambda})\}_{\lambda\in I}$ be a family of graphs with shortest-path metric $d_{\lambda}$. If the Assouad dimension of $\mathcal{X}$ at  all scales is $\beta$, then the Nagata dimension of $\{G_{\lambda}\}$ is bounded above by $\lfloor\beta\rfloor$ at all scales:
    \begin{align}\mathrm{dim}_{\mathrm{Nagata}}\mathcal{X}\leq \lfloor\beta\rfloor\leq \dimasd\mathcal{X}=\beta.
    \end{align}
\end{corollary}

 Below we give the all-scale intrinsic fault-tolerant gate bounds. In principle, to refine the gate bounds, we may choose any growing scale $R_{\lambda}$ such that $r\sim h\eta\ll R_{\lambda}, R_{\lambda}^\beta\ll d_{\lambda}$, where $d$ is the code distance and $r$ is the (constant) thickness of the annular shell.

\subsection{Limitations on constant-depth logical gates from intrinsic dimension}

In this subsection, we prove the intrinsic fault-tolerant gate bound. 
It is worth noting again that the bound may in general depend on the choice of stabilizer generators, which can give different connectivity graphs and $\beta$. This is natural and useful, as different connectivity can allow different gates. 

\begin{theorem}[Intrinsic fault-tolerant gate bound] \label{Intrinsic gate bound}
.Let \(\{C_{\lambda}(\mathcal{S}_{\lambda})\}_{\lambda\in I}\) be a code family with specified stabilizer generator sets admitting macroscopic code distance, and
let \(\mathcal{X}=\{G_{\lambda}\}_{\lambda\in I}\) be the associated family of connectivity graphs, with intrinsic dimension $\dim(C_I(\mathcal{S}_{I})) = \beta\in\left[2,+\infty\right)$. 
 Let $U$ be a unitary operator realized by a sequence of constant depth quantum circuits over a graph-local gate set $\mathcal{G}$. Suppose that $U$ preserves the codespace and hence implements an encoded logical gate $\bar{L}$. Then the logical gate $L$ corresponding to $\bar{L}$ belongs to the $\lfloor \beta \rfloor$-th level of the Clifford hierarchy: $L \in \mathcal{P}_{\lfloor\beta\rfloor}$.
\end{theorem}
\begin{proof}
    By Theorem~\ref{Assouad-Nagate bound}, the Nagata dimension of connectivity graphs $\{G_{\lambda}\}$ is bounded by $D\stackrel{\text{def}}{=}\lfloor\beta\rfloor$. According  to Definition~\ref{Nagata dimension}, there exists a $(D+1)$-partition for each connective graph $G_{\lambda}\coloneqq\bigcup_{i=1}^{D+1}\mathcal{B}_i$ with macroscopic scale $r\coloneqq r_{\mathrm{Nagata}}= \log d$, and constant $c_{\mathrm{Nagata}}$, where $h\eta\ll r\ll d^{\frac{1}{\beta}}$ and $h$ is the upper bound for the depth of the quantum circuits.
    
    The $(c_{\mathrm{Nagata}}r)$-boundedness of $B\in\mathcal{B}_i$ means that the volume of $B$ on $G_{\lambda}$ is at most $O(r^{\beta})= o(d)$. Thus, a single $B$ is correctable. By the graph-local condition, any enlargement of $B\in \mathcal{B}_i$ on $G_{\lambda}$ by $i$ layers of gates can be reached by additional $\Delta_i \cdot \eta$ stabilizers on connectivity graph, where $\Delta_i$ is some constant that depends on $i,h$ to be determined later. 
    Therefore, the enlarged support of $B$ remains $(r-2\Delta_i \eta)$-separated from each other. In particular, $N_{\Delta_i\eta}(B)$ on $G_{\lambda}$ is still 2-separated. By the Union Lemma~\ref{Union Lemma}, we conclude that $N_{\Delta_i\eta}(\mathcal{B}_i)$ still forms a correctable set. 
    
    For an arbitrary sequence of logical Pauli operators $P_1,\cdots, P_{D}$, where each is a product of single-qubit Paulis, we first apply the Cleaning Lemma~\ref{lem:cleaning} to change the logical representative of each operator in the sequence such that the support of $P_i$ avoids $N_{\Delta_i \eta }(\mathcal{B}_i)$, that is the correctable set $N_{\Delta_i \eta }(\mathcal{B}_i)\cap \supp(P_i)=\emptyset$ by choosing appropriate representatives.
    
    We define $K_1\vdef UP_1 U^{\dagger}$ and then recursively, $K_j\vdef P_jK_{j-1}P_j^{\dagger}K_{j-1}^{\dagger}$. Since the unitary only spreads the support by a radius of $(2h+1)\eta$, $K_1$ does not intersect $\mathcal{B}_1$. 
    \begin{equation}
        \mathrm{supp}(K_1)\subseteq \mathcal{B}_2\cup\mathcal{B}_3\cup\cdots\cup\mathcal{B}_{D+1}.
    \end{equation}
    For $K_2=P_2K_1P_2^{\dagger}K_1^{\dagger}$, we can show that the support is contained in $\mathrm{supp}(K_1)\cap\mathrm{supp}(K_1P_2^{\dagger} K_1^{\dagger})\subseteq\mathcal{B}_3\cup\cdots\cup\mathcal{B}_{D+1}$, while its circuit depth is of $\Delta_2=2(2h+1+1)$. The iteration proceeds. A similar argument shows that $K_j$ can be implemented within circuit depths $\Delta_j\stackrel{\text{def}}{=}2^jh+3\cdot 2^{j-1} -2$, and that $\mathrm{supp}(K_j)\subset \mathcal{B}_{j+1}\cup\cdots\cup \mathcal{B}_{D+1}$. Iterating through, we obtain that $\supp(K_D) \subset \mathcal{B}_{D+1}$ is a correctable set. Therefore, we have $K_D \Pi=c\Pi$, or equivalently, $K_{D-1}P_D^{\dagger}K_{D-1}^{\dagger}\Pi=cP_D\Pi$. 
    
    Now using the following argument, we can determine that $c=\pm1$.
    Notice that $(K_{D-1}P_D^{\dagger}K_{D-1}^{\dagger})^2=K_{D-1}P_D^{\dagger}K_{D-1}^{\dagger}\cdot K_{D-1}P_D^{\dagger}K_{D-1}^{\dagger}=\mathrm{I}$, we have
    \begin{equation}
        \begin{aligned}
            \Pi=&(K_{D-1}P_D^{\dagger}K_{D-1}^{\dagger})^2\Pi=(K_{D-1}P_D^{\dagger}K_{D-1}^{\dagger})cP_D\Pi\\
            =&c(K_{D-1}P_D^{\dagger}K_{D-1}^{\dagger}P_D)\Pi=cK_D\Pi=c^2\Pi\\ \Rightarrow& c=\pm 1 \Rightarrow (K_{D-1}P_D^{\dagger}K_{D-1}^{\dagger}P_D)\Pi=\pm \Pi.
        \end{aligned}
    \end{equation}
    Thus the commutator of logical operators $\left[K_{D-1},P_D\right]=\pm\overline{\mathrm{I}}$.
    Recall that any unitary operator which commutes or anti-commutes with an arbitrary logical Pauli operator has to be a logical Pauli operator.  This implies that $\overline{K_{D-1}}$ is a logical Pauli operator. Repeating the argument backwards gives $K_{D-j}\in \mathcal{P}_j$. Therefore $K_1 
    \in \mathcal{P}_{D-1}$ and $U \in \mathcal{P}_D$.
\end{proof}
 Theorem~\ref{Intrinsic gate bound} implies that, starting from an arbitrary family of connectivity graphs, the fault-tolerant logical gate set is restricted to $\mathcal{P}_{\lfloor\beta\rfloor}$. 
 
 
 On the other hand, given a set of physical gates that are allowed to be implemented on a certain physical architecture, we have the flexibility of choosing stabilizer generators to increase the intrinsic dimension of its connectivity graph, which can potentially enrich the graph-local gate set $\mathcal{G}$. Such generator and graph engineering may potentially enable more fault-tolerant logical gates compiled by enlarged $\mathcal{G}$ and thus of practical interest.

\subsection{Intrinsic fault-tolerant gate bound across multiple code blocks}
 We now move on to the multi-block case, where we consider the tensor product of $K=O(1)$ copies of codes with additional cross-block interactions. Denote the new gate set by $\widetilde{\mathcal{G}}=\mathcal{G}\bigcup \mathcal{G}^{\prime}$, where $\mathcal{G}$ is the set of gates within a single code block and $\mathcal{\mathcal{G}^{\prime}}$ is the set of cross-block physical gates that are translation invariant. The modified connectivity graph $\{\widetilde{G_{\lambda}}\}_{\lambda\in I}$ is updated by adding a set of \emph{virtual edges}. The vertices of any virtual edge lie in the support of a cross-block action $g^{\prime}\in \mathcal{G}^{\prime}$ on $\bigsqcup_{i =1}^{K}G_{\lambda}^{(i)}$. 
 \begin{equation}
     \widetilde{G_{\lambda}}=\mathop{\bigcup}_{i=1}^{K}G_{\lambda}^{(i)}\bigcup \{e\mid e=(q_i,q_j),q_i,q_j\in \mathrm{supp}(g^{\prime})\}.
 \end{equation}
  \begin{figure}[htbp]
    \vspace{-1em}
    \centering
        \includegraphics[width=0.8\linewidth]{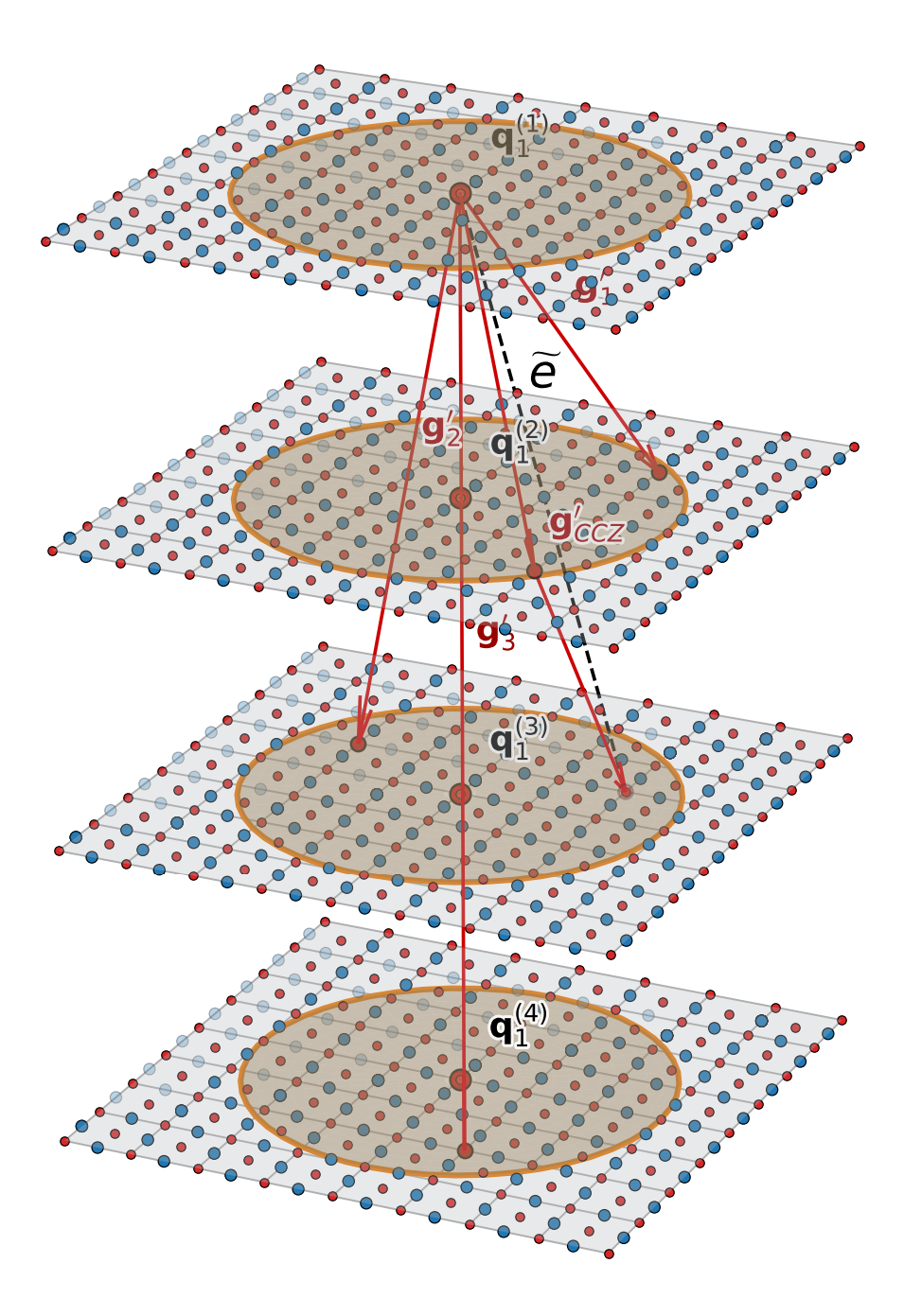}
    \caption{{Illustration of multi-code-block gates}. The red dot $q_1^{(i)}$ denotes the unique physical qubit parallel to $q_1$ on $i$-th code block. The arrows $g_1^{\prime},g_2^{\prime},g_3^{\prime},g_{\ccz}^{\prime}$ are virtually graph-local cross-block multi-qubit gates, whose support is contained in a constant-radius ball centered at $q_1^{(i)}$ on different blocks. The dashed black line $\widetilde{e}$ stands for a new virtual edge on $\widetilde{G}_\lambda$. By our definition, its weight $w_{\widetilde{e}}=d_{G_{\lambda}^{(3)}}(q_1^{(3)},q_3)+1=4$, as we ignore the blue ancilla qubits on connectivity graph. }
    \label{cross-block_gates}
\end{figure}
 Since each code block is identical, we can identify the unique corresponding qubits on each code block. For a qubit $q_i\in G_{\lambda}^{(i)}$, we denote the corresponding $q_i^{(j)}\in G_{\lambda}^{(j)}$. 
 For each newly added inter-block edge, $e=(q_i,q_j),q_i\in G^{(i)}_{\lambda}, q_j\in G_{\lambda}^{(j)}, i\neq j$, if  $q_i^{(j)}\in G_{\lambda}^{(j)}$ corresponding to $q_i$, then the edge $e\in \widetilde{G_{\lambda}}$ is assigned  weight $w_e$ as follows,
 \begin{align}\label{new_weight}
     w_e&\vdef d_{G_{\lambda}^{(j)}}(q_i^{(j)},q_j)+1\nonumber\\
     &=|\text{Shortest path between }q_i^{(j)},q_j \text{ on } G_{\lambda}^{(j)}|+1.
 \end{align} 
 
 $\widetilde{G_{\lambda}}$ is now a connected graph with a well-defined weighted path metric $\widetilde{d_{\lambda}}$ between all vertices. The original metric within each block remains unchanged. 
 \begin{equation}\label{Loop-free condition}
    \forall p,q \in G_{\lambda}^{(i)}, \widetilde{d_{\lambda}}|_{G_{\lambda}^{(i)}}(p,q)=d_{G_{\lambda}}(p,q), i\in\left[K\right].
 \end{equation}
 For two qubits on the same code block, any path involving cross-block interaction has an additional length penalty due to the ``+1'' in Eq.~(\ref{new_weight}). Consequently, the only modification is the distances between qubits residing in distinct blocks.
 
 We say the cross-block gate set $\mathcal{G}^{\prime}$ is \emph{virtually graph-local}, if every newly added edge is of weight $w_e$ is bounded by a constant $\eta=O(1)$ independent of $\lambda$ (See Fig.~\ref{cross-block_gates}). This is a sufficient condition that ensures the locality of the cross-block gates with respect to the new connectivity graph. Thus, the entire set of gates $\widetilde{\mathcal{G}}=\mathcal{G}\bigcup\mathcal{G}^{\prime}$ remains graph-local.  The virtually graph-local condition justifies the locality of the cross-block gate set on connectivity graphs. More precisely, given a graph-local gate set $\mathcal{G}$ and a virtually graph-local cross-block gate set $\mathcal{G}^{\prime}$, $\widetilde{\mathcal{G}}=\mathcal{G}\bigcup\mathcal{G}^{\prime}$ satisfies the graph-local condition for $\{\widetilde{G}_\lambda\}_{\lambda\in I}$. For any $\widetilde{g}\in \widetilde{\mathcal{G}}$ and operator $P$, 
 \begin{equation}
     \supp(U_{\widetilde{g}} P U_{\widetilde{g}}^\dagger)
\subseteq N_{\widetilde{G}_{\lambda}}(\supp(P),2\eta).
 \end{equation}
 The \emph{range} of cross-block interaction $\rho$ on $\widetilde{G_{\lambda}}$ is defined by 
 \begin{equation}
     \rho=\sup_{\lambda\in I,q\in \widetilde{G}_{\lambda}}\left|\bigcup_{g^{\prime}\in \mathcal{G}^{\prime},q\in\mathrm{supp}(g^{\prime})}\mathrm{supp}(g^{\prime})\right|\in \mathbb{R}\cup\{+\infty\}.
 \end{equation}
 The range $\rho$ measures the maximal number of virtual edges connected to a single qubit within the entire graph family $\widetilde{\mathcal{X}}$. Note that the range $\rho$ is always finite if the cross-block gate set $\mathcal{G}^{\prime}$ is a both virtually graph-local and translationally invariant. In this case, the range $\rho$ is upper bounded by the sum of volume of balls with constant radius $\eta$ in each block.
 \begin{equation}\label{finite_range}
     \rho\leq \mathrm{Vol}\left(\bigcup_{i=1}^{K}N_{q_i\in G_{\lambda}^{(i)}}(\eta)\right)\leq K\cdot C_{\beta} \eta^{\beta}<+\infty.
 \end{equation}

 \begin{theorem}[Intrinsic multi-block fault-tolerant gate bound]\label{Intrinsic gate bound Multi}
Let \(\{C_{\lambda}(\mathcal{S}_{\lambda})\}_{\lambda\in I}\) be a code family with specified stabilizer generator sets admitting macroscopic code distance, and
let \(\mathcal{X}=\{G_{\lambda}\}_{\lambda\in I}\) be the associated family of connectivity graphs, with intrinsic dimension $\dim(C_I(\mathcal{S}_{I})) = \beta\in\left[2,+\infty\right)$.
     For a $K$-fold code family $\{\widetilde{C_{\lambda}}=C_{\lambda}^{\otimes K}\}_{\lambda\in I}$, let $ \widetilde{\mathcal{X}}$ be the family of connectivity graphs $\{\widetilde{G_{\lambda}}\}_{\lambda\in I}$ updated by the virtually graph-local cross-block gate set $\mathcal{G}^{\prime}$, then the intrinsic dimension of $\widetilde{C_I}$ remains unchanged:
     \begin{equation}
         \mathrm{dim}(\widetilde{C_I})=\dimasd(\widetilde{\mathcal{X}})= \beta.
     \end{equation}
     Any constant-depth unitary operator $U$ implements encoded logical action restricted to the $\lfloor \beta\rfloor$-th level of the Clifford hierarchy $\mathcal{P}_{\lfloor\beta\rfloor}$.
 \end{theorem}
\begin{proof}
    We will show that the Assouad dimension of $\{\widetilde{(G_{\lambda}},\widetilde{d_{\lambda}})\}_{\lambda\in I}$ is $\beta$. Consider the ball $N_{q}(R)$ centered at $q$ of radius $R\in \mathbb{Z}^{\geq 0}$ on $\widetilde{G_{\lambda}}$, denote its restriction on the subgraph of its own code block by $N_{q}^{\circ}(R)$. Let \(A_q(R,r)\) and \(A_q^\circ(R,r)\) denote the minimum numbers of balls of radius \(r\) needed to cover \(N_q(R)\) and \(N_q^\circ(R)\), respectively.
     
     To derive a decomposition of $N_q(R)$, we classify its qubits in $N_{q}(R)$ by path configuration to $q$. For any $p\in N_{q}(R)$, the shortest path connecting $q=q_0$ to $p=q_{\text{final}}$ can be expressed by:
    \begin{equation}
        q_0\sim q_1\xrightarrow{g^{\prime}_1}q_1^{\prime}\sim q_2\xrightarrow{g^{\prime}_2}q_2^{\prime}\sim\cdots \xrightarrow{g^{\prime}_{f}}q_f^{\prime}\sim q_{\text{final}}=p,
    \end{equation}
    where $\sim$ is the shortest path within the same code block, $g_i^{\prime}$ belongs to the cross-block gate set $\mathcal{G}^{\prime}_i$. Since the metrics within each $G^{(i)}_{\lambda}$ are identical, the path configuration can be rearranged by translation of $g^{\prime}_{i}$:
    \begin{equation}
        q_0\xrightarrow{g^{\prime}_1}q_1\xrightarrow{g^{\prime}_2}q_2\cdots \xrightarrow{g^{\prime}_{f}}q_f\sim q_{\text{final}}=p.
    \end{equation}
    where each $q_i$ belongs to different code block. Otherwise, assume $q_r,q_s\in G_{\lambda}^{(i)},r\neq s$. Then, by Eq.~(\ref{Loop-free condition}), there exists a shorter path connecting $q_r,q_s$ within $G_{\lambda}^{(i)}$, which contradicts  the shortest path assumption. Therefore, $f\leq K$. We define $W_K$ as the set of sequences of cross-block gates in $\mathcal{G}^{\prime}$ with  length less than or equal to $K$. The sequence $(g_1^{\prime},g_2^{\prime},\cdots,g_f^{\prime})\in W_k$ is the cross-block configuration of $p\in N_q(R)$. This leads to the estimation of the covering number of $N_q(R)$ at scale $r$:
\begin{align}
    N_q(R)
    &\subseteq
    \bigcup_{(g'_1,g'_2,\ldots,g'_f)\in\mathcal G'}
    N_{q_{\mathrm{final}}}^{\circ}
    \left(
        R-\sum_{i=1}^{f}|g'_i|
    \right),
    \\
    A_q(R,r)
    &\leq
    \sum_{(g'_1,g'_2,\ldots,g'_f)\in\mathcal G'}
    A_{q_{\mathrm{final}}}^{\circ}
    \left(
        R-\sum_{i=1}^{f}|g'_i|,r
    \right)
    \\
    &<
    |W_K|\, A_{q_{\mathrm{final}}}^{\circ}(R,r)\leq
    \left(\sum_{i=1}^{K}\rho^i\right)
    \cdot c_{\beta}
    \left(\frac{R}{r}\right)^{\beta}
    \\
    &\leq
    K\rho^K c_{\beta}
    \left(\frac{R}{r}\right)^{\beta}.\label{multicodeblock_IGB}
\end{align}
    Since $\rho$ is finite, by Eq.~(\ref{multicodeblock_IGB})  $\dimasd(\widetilde{X}) \leq \beta$. On the other hand, since $G_{\lambda} \subseteq \widetilde{G_{\lambda}}$, it follows that $\dimasd(\widetilde{X}) \geq \dimasd(\{G_{\lambda}\}_{\lambda\in I}) = \beta$. Hence $\dimasd(\widetilde{X}) = \beta$. Finally, applying Theorem~\ref{Intrinsic gate bound} for $\widetilde{\mathcal{X}}$, the set of fault-tolerant logical gate is restricted to the $\lfloor\beta\rfloor$-th level of the Clifford hierarchy.
\end{proof}
  As a typical setting, substituting each qubit in a given code family $\mathcal{X}$ by a site with $N=O(1)$ new qubits with full connectivity among themselves will not change the intrinsic dimension. 
  
   An immediate example that saturates this bound for integer dimension is the stack of $D$-dimensional toric codes, which supports transversal $\overline{\mathrm{C}^{D-1}Z} \in \mathcal{P}_D$~\cite{Wang_2024,Jochym_O_Connor_2021}. 
   
   For codes that are non-qLDPC, i.e.~have diverging weight or degree, our intrinsic gate bounds allow the possibility of achieving arbitrarily high level of the Clifford hierarchy. This bound is saturated in recent work~\cite{Golowich_2024,He_2025} using asymptotically good algebraic geometry (AG) codes with linear weight, where transversal $\ckz$ gates for arbitrarily large $k$ can be realized on this model by leveraging \(k\)-orthogonality. Such codes illustrate how code properties can benefit from infinite intrinsic dimension: by sacrificing qLDPC-ness and geometric locality, one can simultaneously achieve good parameters and high-level gates.
   Whether these desirable properties can be achieved with lower weight or in the single-block case remains as interesting open problems.
   
   For generic qLDPC codes, it remains unknown whether fault-tolerant logical gates can systematically saturate the hierarchy levels corresponding to their Nagata dimensions. Recent studies~\cite{Breuckmann_2026,Li_2025,li2026theorycohomologicalinvariantsquantum} suggest that (co)homological  invariants may offer the algebraic machinery to relate the support of transversal operations to the underlying Nagata division.

  Refs.~\cite{Zhu_2022,Dua_2023} indicate that the fractional 3D toric code with intrinsic dimension $\beta<3$ can still implement cross-block logical $\ccz$ gate using transversal global $\ccz$ followed by lattice surgery. Here we briefly review the construction of $\mathrm{FC}(a, b)$ in 3D lattice and show that this does not contradict our result. A standard cube is first divided equally as $a \times a \times a$ smaller cubes. The $b \times b \times b$ small $m$-type cubes in the center are then punched out. This process is repeated $l$ times for fractional 3D code at level $l$. The induced code family is denoted by $\mathrm{FC}(a, b) =\{\mathrm{FC}(a, b, 1)\subseteq \mathrm{FC}(a, b, 2)\subseteq \cdots\subseteq \mathrm{FC}(a, b, \infty)\}$. A direct calculation indicates that the Assouad dimension of the fractal set $\beta=\dimasd\mathrm{FC}(a,b)=\log_{a}{(a^3-b^3)} \to 2+\log_{a}(3(a-b))$. 
  The procedure for executing a logical $\ccz$ gate is divided into three steps. First, tile a
stack of three blocks $\mathrm{FC}(a, b)$ in a predetermined direction with the $(m_a,m_b,e_c)$ boundary on the outer surface and the $(m_1,m_2,m_3)$ boundaries on the inner holes. After applying a global transversal gate $\widetilde{\ccz}=\otimes_j \ccz_{j;1.2.3}$, or equivalently a cross-block physical $\ccz$ sweeper $s^{(3)}_{1,2,3}$ along the tiled system, the domain wall condenses on the outer boundary but turns the original hole boundaries into exotic $(m_1s^{(2)}_{2,3},m_2s^{(2)}_{1,3},m_3s^{(2)}_{1,2})$ boundaries. Finally, apply lattice merge and split to copy the logical states to ancilla code blocks. Such operations involve adaptive logical action and connectivity modification, which implicitly violate the original circuit setting and gate local assumption of the intrinsic fault-tolerant gate bound.  
Thus, it does not contradict the result from Corollary~\ref{Intrinsic gate bound Multi} for a stack of three fractional 3D toric code, which shows that only the constant-depth logical $\overline{\mathrm{C}Z}$ is allowed. Our result matches the convenient and neat expression for transversal $\overline{\mathrm{C}Z}$ by avoiding $(m_1,m_2,m_3)$hole boundaries in Fig.~37 of Ref.~\cite{Zhu_2022}.

\section{Geometric obstructions to self-correcting quantum memories}
\label{sec:towards the nonexistence}

Beyond the above static properties, the possibility of ``self-correction'' without the need for active error-correction operations is a particularly desirable feature of quantum codes that has attracted longstanding interest.
More specifically, a self-correcting quantum memory (SCQM) is a family of quantum many-body systems that can passively preserve encoded quantum information for a memory time that diverges rapidly with system size when coupled to a sufficiently cold thermal bath at a fixed temperature. For stabilizer Hamiltonians, self-correction is closely tied to the logical energy barrier: a logical operator that can be implemented through a sequence of local errors with only a constant maximal energy cost provides a direct obstruction to self-correction. 

Understanding the dimensional requirements for SCQM has been a central problem, attracting sustained attention in quantum information and many-body physics. 
Commuting stabilizer Hamiltonians in 1D and 2D have a constant energy barrier and thus cannot be self-correcting~\cite{Bravyi_2009}.
The 4D toric code provides the canonical SCQM~\cite{alicki2010thermal}. Fractalized 4D toric-code models have also been shown to retain topological order and macroscopic code distance on supports of Hausdorff dimension $3+\epsilon$, providing candidate SCQMs just above three dimensions~\cite{Zhu_2022}. After a series of candidate models and partial self-correction results in 3D~\cite{williamson2025partialselfcorrectionlayercodes,gu2025layercodespartiallyselfcorrecting,lin2024proposals3dselfcorrectingquantum}, very recent work by Balasubramanian, Davydova and Lin constructed a 3D local Pauli stabilizer Hamiltonian with an exponentially growing memory time at sufficiently low nonzero temperature~\cite{balasubramanian2026passiveselfcorrectingquantummemory}. In the commuting stabilizer setting, $3$ is the minimum Euclidean dimension known to support a self-correcting logical qubit.

Our framework naturally raises the question of whether  SCQMs can be realized in dimension below three. 
Here, we advance this inquiry by using the Assouad--Nagata reduction (see Section~\ref{sec:intrinsic dimension bridges code symmetry and indistinguishability} and Appendix~\ref{app:A-N reduction}) to identify a natural geometric obstruction to $\beta<3$ SCQMs. Specifically, we prove the following conditional no-go theorem, which essentially states that if the stabilizer geometry admits a transverse projection of Assouad dimension strictly below $2$, then it has a constant energy barrier and cannot be self-correcting. 

 \begin{theorem}\label{thm:conditional no-go on scqm}
    Given a family of stabilizer codes $\{\mathcal{C}_{\lambda}\}_{\lambda\in I}$
    defined on lattice $\mathbb{Z}^3$ with macroscopic code distance and local stabilizer generators, suppose that there exists a projection direction $p\in \mathbb{S}^{2}$ on the 2D sphere such that the projection image along $p$ has dimension $\dimasd \pi_p \mathcal{C}_I<2$, then the quantum code family realized by the degenerate ground spaces of the corresponding stabilizer Hamiltonians cannot be self-correcting. Furthermore, the existence of a dimension-deficit projection among the stabilizer generators imposes an extended constraint, forbidding self-correction in any subsystem code family $\{\mathcal{C}^{\prime}_{\lambda}(\mathcal{G}_{\lambda})\}_{\lambda\in I}$ realized by the ground space of gauged local stabilizer Hamiltonians.
  \end{theorem}
  \begin{proof}
      Assume such a projection $\pi_p \mathcal{C}_I$ exists. By the Assouad--Nagata reduction (\ref{Assouad-Nagate bound}), the Nagata dimension of $\pi_{p}\mathcal{C}_{I}$ is 1. There exists a Nagata division for $\pi_{p}\mathcal{C}_I\coloneqq\mathcal{B}_1\cup\mathcal{B}_2$ at a sufficiently large constant scale $r$ that is greater than the diameter of the local stabilizer generators. By definition of the Nagata division, any set in either of those two families is $cr$-bounded whose area is bounded above by some constant $A$. As a result, taking the preimages of these sets along projection $p$ provides a partition of the 3D lattice into a union of strips whose bases are elements in $\mathcal{B}_1$ or $\mathcal{B}_2$. 
      
      Next, we claim that there exists a logical operator living in one of the vertical strips by contradiction. Assume the contrary holds, we deduce that any strip with a base set in $\mathcal{B}_1$ contains no logical operator and is hence correctable by the cleaning lemma (see Appendix~\ref{app: preliminary}). We also notice that a single stabilizer generator merely intersects with at most one of the strips in $\pi_p^{-1}(\mathcal{B}_1)$ since $\mathcal{B}_1$ is $r$-separated. Therefore, the disjoint union of strips $\pi_p^{-1}(\mathcal{B}_1)$ is still correctable by the union lemma (see Appendix~\ref{app: preliminary}). The cleaning lemma further allows us to modify the support of the logical operator $P$ to avoid $\pi_p^{-1}(\mathcal{B}_1)$ by multiplying some stabilizers $S\in\mathcal{S}$. Afterwards, the support of the modified logical operator $\mathcal{L}=PS$ lies in $\pi_p^{-1}(\mathcal{B}_2)$, whose domain can be further decomposed into subregions within the disjoint strips in $\pi_p^{-1}(\mathcal{B}_2)$,
      \begin{equation}
          \mathcal{L}=\mathcal{L}_1\times\cdots\times \mathcal{L}_{|\mathcal{B}_2|},\supp(\mathcal{L}_i)\subseteq \pi_p^{-1}(B_2), B_2\in \mathcal{B}_2 .
      \end{equation}
     Since $\mathcal{B}_2$ is $r$-separated and any stabilizer generator interacts with at most one $\mathcal{L}_{i}$. The commutativity of $\mathcal{L}$ with stabilizer group $\mathcal{S}$ can be expressed locally, 
     \begin{equation}
         \mathcal{L}\in \mathcal{C}(\mathcal{S})\iff \mathcal{L}_{i}\in \mathcal{C}(\mathcal{S}) \text{ for every }i.
     \end{equation} Any fixed $\mathcal{L}_{i}$ commutes with the generators in $\mathcal{S}$ on its own support in $\pi_p^{-1}(B_2)$ or simply has no interaction. For both cases, this proves $\mathcal{L}_i$ to be a local logical operator whose support lies in a single strip. Since $\mathcal{L}\neq \overline{\mathrm{I}}$ is a nontrivial logical operator, there always exists a logical operator $\mathcal{L}_i\neq \overline{\mathrm{I}}$ located within a single strip belonging to $\pi_p^{-1}(\mathcal{B}_2)$. One can construct an energy walk $\mathcal{W}(\mathrm{I}, \mathcal{L}_i)$ to implement $\mathcal{L}_i$ in a layer-by-layer manner along the projection $p$ within this strip. The emerged violated check of this energy walk always lives on the top and the bottom layer of the column within this strip, which can be bounded above by the constant $2A$. 

     This argument essentially demonstrates the existence of the logical quasi-string and thus proves an $O(1)$ energy barrier for this code family $\{\mathcal{C}_{\lambda}\}_{\lambda\in I}$.
     
     Furthermore, the above argument can be naturally extended to the case of the subsystem code $\{\mathcal{C}^{\prime}_{\lambda}(\mathcal{G}_{\lambda})=\mathcal{C}^{\prime}_{\lambda}\left(\langle\mathcal{S}_{\lambda}\rangle;\overline{X}_1,\overline{Z}_1,\cdots,\overline{X}_{g_{\lambda}},\overline{Z}_{g_{\lambda}}\right)\}_{\lambda\in I}$ for any choice of the gauge Pauli operators $\overline{X}_1,\overline{Z}_1,\cdots,\overline{X}_{g_{\lambda}},\overline{Z}_{g_{\lambda}}$ with only the generating set $\{\langle\mathcal{S}_{\lambda}\rangle\}_{\lambda\in I}$ being local. We can go through a similar argument by firstly assuming no dressed logical operator $\mathcal{L}_{\mathrm{dress}}\in \mathcal{C}(\mathcal{S}_{\lambda})\backslash\mathcal{G}$ lives in any single strip of $\pi_p^{-1}(\mathcal{B}_1)$. And we apply the cleaning lemma on a certain logical operator in $\mathcal{C}(\mathcal{G})$ for subsystem code (Lemma 2 in Ref.~\cite{Bravyi_2009}) by multiplying some $S\in \mathcal{S}_{\lambda}$, and then decompose its residual support to derive a local quasi-string in $\mathcal{C}(\mathcal{S}_{\lambda})\backslash\mathcal{G}$. This contradiction proves that there always exists a logical quasi-string beyond the gauge Pauli operator $\overline{X}_i,\overline{Z}_i$ having a constant energy cost. Such a string precludes the possibility of designing a quantum memory in a multi-logical qubit system by simply gauging out the constant-scale logical qubits.
  \end{proof}

This result establishes a general local geometric constraint that any \(\beta<3\) SCQM must evade. Specifically, if, within some region, projecting the stabilizer geometry transverse to a given direction produces a set of Assouad dimension below \(2\), then a quasi-string-like logical operator is guaranteed to exist within a rigid constant-radius tube oriented along that direction, thereby ruling out self-correction. Intuitively, this projection condition captures a dimensional deficit feature: although the code is embedded in 3D, it does not appear uniformly 2D filling when viewed along some direction. 
The hierarchical geometry  of the recent 3D SCQM construction~\cite{balasubramanian2026passiveselfcorrectingquantummemory}  evades this projection-deficit condition.
Note that besides the single-qubit main construction, Ref.~\cite{balasubramanian2026passiveselfcorrectingquantummemory} also discusses in Appendix C a multiqubit extension based on tensor products and a subsystem encoding without ``small relations.'' Our result
 also addresses the subsystem code case to clarify that a naive approach, by simply treating thermally unstable logical strings as gauge operators, is insufficient. Realizing self-correction demands a substantial modification to the geometric structure of stabilizer generators at a macroscopic scale. 


Theorem~\ref{thm:conditional no-go on scqm} thus provides a foundation for comprehensively understanding the  obstructions to SCQMs in below 3D. 
  An important route toward advances would be to successively apply this conditional no-go argument across multiple local domains. This program has  intriguing connections to Kakeya-type problems, a central class of questions in geometric measure theory and harmonic analysis. Consider the 2D slice of constant thickness perpendicular to an arbitrary direction $p$ with a fixed base point. By Theorem~\ref{thm:conditional no-go on scqm}, the dimension deficit of stabilizer support within such a slice will imply the existence of a local quasi-string along $p$, which can be seen as an energy walk commuting with those stabilizer generators. The failure of admitting a dimension deficit for any direction would force a Kakeya-type $\delta$-thickening disc obstruction, where the cubic connectivity growth should be impossible in Assouad dimension $<3$~\cite{Wang_2025}. By imposing a spatial similarity condition weaker than the translation invariance, we may patch those local strings at constant energy cost. By ``patching'' these local string-like segments at their ending points, one could construct a global logical operator confined entirely within a flexible, piecewise-straight tubular neighborhood of constant radius. The existence of such a bendable quasi-1D support would rigorously imply a constant energy barrier for the entire system and hence preclude self-correction. 

We close this section with a discussion of fractal constructions and the role of geometric regularity in self-correction below 3D.
 The ferromagnetic Ising model on the Sierpiński carpet exhibits finite-temperature order~\cite{vezzani2003spontaneous,gefen1984phase}, while generalized Sierpiński carpets have Hausdorff dimensions dense in $(1,2)$. Inspired by this behavior, Brell suggested that a homological product of Sierpiński carpet codes supported on a fractal geometry of Hausdorff dimension $2+\epsilon$ may be SCQM~\cite{Brell_2016}. Nonetheless, a growing memory lifetime was not established and the proposal has since been argued not to be self-correcting~\cite{lin2024proposals3dselfcorrectingquantum,balasubramanian2026passiveselfcorrectingquantummemory}.
The recent 3D SCQM construction~\cite{balasubramanian2026passiveselfcorrectingquantummemory},  viewed through our framework, can have intrinsic dimension $\beta=2+\epsilon$. 
The construction, however, is hierarchical and breaks translation invariance. This suggests the more physical question of whether a geometrically homogeneous 3D SCQM can be realized, and whether a broader geometric obstruction can rule out all sufficiently regular SCQMs below 3D. We expect progress toward answering these questions to be highly profound both mathematically and physically.

\section{Discussion and outlook}\label{sec:Discussion and outlook}

 In this work, we proposed and developed a mathematical framework for characterizing the intrinsic locality of stabilizer codes. This concept should be viewed as a fundamental parameter associated with the stabilizer generators independent of any prescribed background geometry, extending qLDPC codes beyond the usual integer-dimensional setting in a natural and rigorous manner. 
 Defined through the maximal growth rate of connectivity over the entire space and across all scales, this dimension classifies quantum codes by their intrinsic degree of locality rather than some ambient space.

We show, across fundamental aspects of quantum codes including code parameters, fault-tolerant logical gates, and thermal stability, that intrinsic dimension serves as a fundamental organizing parameter that encapsulates the essential nature of quantum codes.
These results lift the conceptual  role of geometric dimension from a passive background feature to an active structural parameter that governs the capabilities and limitations of quantum codes in a unified fashion. 
 Importantly, the versatility of this perspective opens up rich avenues for both the mathematical theory and practical  implementation of quantum codes.

 Here we have laid the foundation for the theory of intrinsic dimension.  {Many intriguing questions and promising future directions emerge from this theory, a few of which we now outline.  }

\begin{itemize}

\item \textbf{Algorithmic optimization of stabilizer generators and minimal intrinsic dimension.} For a given encoding, an important open question is whether there exists an efficient algorithm for finding a set of stabilizer generators that realizes the minimal intrinsic dimension $\beta_0$. This is a mathematically profound and challenging problem, with some related discussion in the main text. This would pave the way for optimal hardware efficiency by systematically minimizing the overhead required for scalable physical implementations.
    
\item \textbf{Optimal bi-Lipschitz geometric embeddings.}  Further, it would be ideal to understand whether a code family can always be strictly embedded into a geodesic space exhibiting tight Ahlfors $\beta_0$-regularity via a uniformly bi-Lipschitz map. 
This is expected to be feasible, with intuition that the minimal intrinsic dimension should exhibit an underlying spatial homogeneity that facilitates such embeddings. 
Such embeddings would provide a concrete geometric foundation for the mathematical dimension notions and guide the optimal design of experimental code architectures.

    \item \textbf{Saturating intrinsic code parameter bounds.} Another important question is to identify well-behaved codes and stabilizers that saturate the intrinsic code parameter bounds imposed by Theorem~\ref{Intrinsic BPT bound}. As discussed above, this has been  achieved for integer dimensions, whereas the noninteger dimensional regime remains largely open and constitutes an interesting direction for future work.
    
    \item \textbf{Dimension unification on non-compact manifolds.} Extending our framework to noncompact Riemannian manifolds also raises an appealing and deep mathematical question.  A primary difficulty is the intrinsic divergence of the global Assouad dimension in spaces with negative curvature, such as hyperbolic manifolds. In particular, this is consistent with the polylogarithmic breaking of the BPT bound achieved by hyperbolic codes. Developing a more refined theory and better bounds for hyperbolic codes would likely require modification of the hyperbolic metric. Furthermore, as the code block size scales to infinity, a subsequent geometric challenge is to construct a well-behaved triangulation bi-Lipschitz to the original manifold and then control the scale of Nagata division. For fault-tolerant logical gates in this non-compact setting, a detailed analysis on the Assouad dimension up to scale $o(d_{\lambda}^{\frac{1}{\beta}})$ may improve the results.

    \item \textbf{Tuning intrinsic dimension via Tanner graph engineering.} As discussed above, the 
    intrinsic dimension can be manipulated by introducing or pruning long-range connectivity. However, this is subject to the  fundamental physical interplay between connectivity costs and code capabilities. Notably, understanding how code distance and logical gate properties behave under such tuning requires a more careful analysis, and is of evident practical value.

    \item \textbf{Higher-level logical gates via non-local cross-block interactions.} In the regime of unrestricted interactions across multiple code blocks, the homological product on the ambient geometric space $\mathbb{R}^k$ may enable the transversal implementation of logical $\overline{\ckz}$ gates within the $k$-th level of the Clifford hierarchy $\mathcal{P}_{k}$~\cite{Breuckmann_2026,Li_2025}. By permitting non-local cross-block interactions, our intrinsic fault-tolerant gate bounds can be circumvented. Such block padding effectively amplifies both the connectivity and the intrinsic locality dimension of the code family. This dimensional inflation potentially unlocks the capability to implement fault-tolerant logical gates at higher levels of the Clifford hierarchy (see also Appendix B of Ref.~\cite{Jochym_O_Connor_2018}).

\end{itemize}

Beyond these specific mathematical questions, our framework also suggests insights and connections to broader physical and experimental contexts.

 As mentioned, noninteger dimensions have long played a fundamental
role in physics from the $\epsilon$-expansion and dimensional
regularization to many-body physics and field theories on fractal geometry
~\cite{WilsonFisher1972,BolliniGiambiagi1972,tHooftVeltman1972,
ElShowk2014,DelmastroLiRathore2026,gefen1984phase,
Eyink1989,KarRajeev2012}.
On the other hand, quantum error correction has revealed deep physical connections to our understanding of e.g.~quantum many-body order~\cite{Kitaev_2003,dennis2002topological,BravyiHastingsMichalakis2010,YiYeGottesmanLiu2024} and holographic quantum gravity~\cite{AlmheiriDongHarlow2015,PastawskiYoshidaHarlowPreskill2015}. 
As locality and geometry underlie the essential physics in these settings, our theory of intrinsic dimension creates a new bridge between quantum error correction and quantum phases of matter, field theory, and holography, providing a promising framework for understanding and characterizing their physics across integer and noninteger dimensions. 
Relatedly, an interesting goal is to characterize the intrinsic interaction geometry of many-body quantum systems directly from physical diagnostics such as entanglement and correlation structures.

On the experimental side, our framework may serve as a geometric design principle for practical implementations of quantum error correction and fault-tolerant quantum computation, especially on platforms with flexible or nonuniform connectivity. Incorporating the available connectivity and its physical costs into the locality graph may provide a systematic route to jointly optimizing stabilizer generators, code layouts, and fault-tolerant gate implementations under realistic hardware constraints.

\begin{acknowledgments} 
We thank Margarita Davydova, Shankar Balasubramanian, and Ting-Chun Lin for valuable correspondence regarding their recent self-correcting quantum memory construction in 3D.
YL and ZWL are supported in part by NSFC under Grant No.~12475023, Dushi Program, and a startup funding from YMSC. XF acknowledges the international visiting student program of Tsinghua University for hospitality.
\end{acknowledgments}

\bibliography{references}


\onecolumngrid

\appendix
\counterwithin{theorem}{section}

\section{Preliminaries} \label{app: preliminary}
\subsection{Quantum error-correcting codes} \label{app:pre_quantum_code}
This subsection includes notations and definitions from quantum error correction that will be used throughout the paper. Readers already familiar with these notions may skip and refer back to it as needed. We first review stabilizer codes, code distance, and qLDPC code families. We then recall logical gates and fault-tolerant implementations, including transversal gates, constant-depth circuits, and the Clifford hierarchy. Finally, we introduce the connectivity graph and Tanner graph associated with a stabilizer code, and restate several standard theorems related to no-go theorems in a form adapted to a graph-theoretic language.
  \subsubsection{Stabilizer codes}
  The single qubit Pauli group is given by $\hat{\mathcal{P}}=\langle X,Y,Z\rangle $, and the $n$-qubit Pauli group is given by $\mathcal{P}_1=\hat{\mathcal{P}}^{\otimes n}$. We denote $\langle \mathcal{S}\rangle$ as the group generated by $\mathcal{S}$, where $S\subset \mathcal{P}_1$. $\mathcal{N}(\mathcal{S})$ is the normalizer of $\langle \mathcal{S}\rangle$ defined as $\mathcal{N}(\mathcal{S})=\{p\in \mathcal{P}_1\mid S p=pS,\forall S\in \mathcal{S}\}$.
 \begin{definition}[Quantum code and stabilizer code] A quantum code is a subspace of a Hilbert space $\mathcal{H}$. A stabilizer code $C=C\langle\mathcal{S}\rangle$ with parameters $\llbracket n,k,d\rrbracket$ is a quantum code on $n$ physical qubits, given by the common $+1$ eigenspace of an abelian stabilizer group $\langle \mathcal{S}\rangle$ generated by $n-k$ independent Pauli operators (a $2^k$-dimensional subspace of the $2^n$-dimensional physical Hilbert space, encoding $k$ logical qubits into $n$ physical qubits), with code distance $d$.
\end{definition}
\begin{definition}[Code distance]
    The nontrivial logical operators of a stabilizer code $C$ are the elements of
    $\mathcal{N}(\mathcal{S})\setminus \langle \mathcal{S}\rangle$.
    The distance $d$ of $C$ is the minimum weight of any Pauli operator in
    $\mathcal{N}(\mathcal{S})\setminus \langle \mathcal{S}\rangle$, namely,
    the smallest weight of a nontrivial logical operator.
\end{definition}
Quantum low-density parity-check (qLDPC) codes form a central and widely studied class of stabilizer codes, with both fundamental theoretical significance and practical promise.

\begin{definition}[Quantum low-density parity-check (qLDPC) code]\label{def:qLDPC}
  A family of stabilizer codes is called a quantum low-density parity-check (qLDPC) code family if it admits a choice of stabilizer generators such that each generator has constant weight, i.e., acts nontrivially on at most a constant number of qubits, and each qubit has constant degree, i.e., participates in at most a constant number of generators, where both constants are independent of the code length.
\end{definition}

\subsubsection{Connectivity and Tanner graphs} \label{app:pre_b_logical_gate}
Here we formally define the closely related notions of connectivity graph and Tanner graph of stabilizer codes.
\begin{definition}[Connectivity graph] \label{Connectivity graph}
 Let $\mathcal C = \langle\mathcal{S}\rangle$ be a stabilizer code with generating set $\mathcal S$. 
    The \emph{connectivity graph} associated with $\mathcal S$ is the graph
    $G(\mathcal S)=(V,E)$ defined as follows. 
    The vertex set $V$ consists of the physical qubits. Two distinct vertices $v_i,v_j\in V$ are connected by an edge $e=(v_i,v_j)\in E$ if and only if they are both contained in the support of some stabilizer generator $S_i\in \mathcal{S}$.
    The graph distance between any two distinct vertices $v_1,v_2\in V$, denoted by $d(v_1,v_2)$, is defined as the length of the shortest path between them on the graph.
\end{definition}

\begin{definition}[Tanner graph]\label{tanner graph}
Let $C=\langle \mathcal{S}\rangle$ be a stabilizer code with generating set
$\mathcal{S}=\{S_i\}$. 
The \emph{Tanner graph} associated with $\mathcal{S}$ is the bipartite graph
$G_T=(V_0,V_1,E)$, where the two vertex sets $V_0$ and $V_1$ correspond to
stabilizer generators and physical qubits, respectively. There is
an edge $(S_i,q_j)\in E$ if and only if
$q_j\in \operatorname{supp}(S_i).$
\end{definition}

These graphs depend on the choice of generators and are naturally assumed to be connected. 
For an undirected graph $G$, the ball of radius $r$ on $G$ is denoted by $N_{v}(r)\vdef\{u\in G\mid d_{G}(u,v)\leq r\}$. The volume of $N_{v}(r)$ is the number of vertices it contains $\mathrm{Vol}(N_{v}(r))\vdef\sharp\{x\in N_{v}(r)\subseteq G\}$.

\subsubsection{Logical gates} \label{app:pre_c_logical_gate}
 Here we recall various key definitions and notation concerning logical gates, which are used in the paper.
 \begin{definition}[Support]
    The \emph{support} of an operator $P$, denoted supp$[P]$, is the set of physical qubits on which $P$ acts nontrivially.
    The \emph{size} of the support, denoted $|P|$, is the weight of the non-identity elements of the operator $P$.
 \end{definition}
 \begin{definition}[Code projector]
     Let $\Pi$ denote the projector from the Hilbert space $\mathcal{H}$ onto the code space of $C$. For any encoded state $\rho$, we have $\Pi\rho\Pi=\rho$.
 \end{definition}
 \begin{definition}[Logical gate]
    A unitary operator $U$ is a logical gate for the code $C$  if it preserves the code space of $C$, i.e., $U\Pi_{C}U^{\dagger}=\Pi_{C}$.
 \end{definition}
 Loosely speaking, for fault tolerance, we would like the physical implementation of logical gates to be sufficiently ``local'' so that small errors do not spread into uncorrectable ones. The simplest and most  studied such gates are transversal gates.

 \begin{definition}[Transversal gate]
  A unitary operator $U$ is said to be \emph{transversal} if it can be decomposed into a tensor product of unitaries with pairwise disjoint supports, each acting on at most a constant number of qubits. Namely,
\begin{align}
    U = \bigotimes_{i=1}^{\ell} u_i,
    \qquad
    u_i \in \mathbb{U}(d_i),
    \qquad
    \operatorname{supp}(u_i)\cap \operatorname{supp}(u_j)=\emptyset
    \quad \text{for } i\neq j .
\end{align}
 \end{definition}
 Transversal gates are necessarily fault-tolerant, but the converse need not hold. The relaxation to constant-depth circuit implementations guarantees locality preservation, namely that a local error can propagate only to bounded number of other qubits, and hence the circuit remains fault-tolerant.

\begin{definition}[Constant-depth quantum circuit]\label{def:constant-depth-circuit}
    A unitary operator $U$ is said to be a \emph{constant-depth circuit} if it can be implemented with $\mathcal O(1)$ layers of a local gate set $\mathcal G$, with disjoint supports within each layer.
\end{definition}

\subsubsection{Useful code lemmas} \label{app:pre_d_useful_lemma}

Several fundamental lemmas concerning logical operators and correctable regions
are central to previous studies of the limitations of stabilizer codes~\cite{Bravyi_2009,Bravyi_2010,Bravyi_2013} and to
our derivations. 
We collect them here, reformulated in the language of connectivity graphs for
later use.

\begin{lemma}[Cleaning Lemma]\label{lem:cleaning}
    Let $\mathcal C$ be a stabilizer code on the set of qubits $\Lambda$, with
    stabilizer group $\mathcal S$, and let $A\subseteq \Lambda$ be a region.
    Then, for any logical Pauli operator $L$, exactly one of the following
     holds:
    \begin{enumerate}
        \item $L$ can be cleaned from $A$, i.e., there exists a stabilizer
        $S\in \mathcal S$ such that
            $\operatorname{supp}(LS)\subseteq \Lambda\setminus A$.
        \item There exists a nontrivial logical Pauli operator $M$ supported in
        $A$ such that $M$ anticommutes with $L$.
    \end{enumerate}

    In particular, if $A$ is correctable, equivalently if $A$ supports no
    nontrivial logical Pauli operator, then every logical Pauli operator can
    be cleaned from $A$.
\end{lemma}

Then the following lemmas allow us to build larger correctable regions.

\begin{definition}[Outer boundary]
    Let $\mathcal C$ be a stabilizer code with generating set
    $\mathcal S=\{S_i\}$ on the set of physical qubits $V$, and let
    $U\subseteq V$. The \emph{outer boundary} $\partial_{+}U$ is the set of all physical qubits $v\in V\setminus U$ such that there exists some $u\in U$ and some stabilizer generator $S_i\in\mathcal{S}$ with $u,v\in \mathrm{supp}(S_i)$.
\end{definition}

\begin{lemma}[Union Lemma]\label{Union Lemma}
      Let $U_1,U_2,\ldots,U_t$ be disjoint correctable regions. Suppose that
$\partial_+ U_i\cap U_j=\emptyset$ for all
        $i\neq j$.
    Then the union $\bigcup_{i=1}^t U_i$ is also correctable.
\end{lemma}
In other words, if a collection of correctable sets is pairwise 2-separated in the connectivity graph, then their union remains correctable. 

Finally, we introduce the following lemma which allows one to enlarge a correctable region along correctable boundaries.
\begin{lemma}[Expansion Lemma] \label{Expansion Lemma}
    Suppose $U$ and $T \supset \partial _{+}U$ are both correctable sets. Then $T\cup U$ is correctable. 
\end{lemma}

\subsection{Dimensions in fractal geometry} \label{Appendix A.2}

 In this subsection, we collect various fundamental properties of the Assouad dimension and their proofs. These results provide the mathematical foundation for bridging the gap between abstract geometry and explicit quantum code constructions. We begin by first defining the Assouad dimension. We then prove two key results: its relation to the doubling condition and the sufficient condition given by Ahlfors $\beta$-regularity. Finally, we show that, for a fixed code, the Assouad dimension of its Tanner graph agrees with that of its connectivity graph.

\begin{definition}
A metric space $(X,d)$ consists of a set $X$ and a metric $d:X\times X \to \mathbb{R}_{\geq 0}$ satisfying for all $x, y, z \in X$:
    \begin{enumerate}
        \item (Positive definiteness) $d(x,y) \geq 0$, with $d(x,y)=0 \iff x=y$;
        \item (Symmetry) $d(x,y) = d(y,x)$;
        \item (Triangle inequality) $d(x,z) \leq d(x,y) + d(y,z)$.
    \end{enumerate}
\end{definition}

For the following definition and also in the rest of the paper, $N_x(R)$ is defined as a ball centered at $x$ of radius $R$ in a graph $X_\lambda$ with $x \in X_\lambda$. For a set $\mathcal{F}\subseteq X_\lambda$ , $N_\mathcal{F}(R)$ is defined as $\bigcup_{x\in \mathcal{F}}N_x(R)$.

 \begin{definition}[Assouad dimension in relative version, cf.~\cite{assouad1977espaces,fraser2020}] \label{Assouad dimension}
    The Assouad dimension of a family $\mathcal{X}$ of metric spaces $\{(X_{\lambda},d_{\lambda})\}_{\lambda \in I}$, is defined to be the infimum of $\beta$ satisfying the following \emph{covering condition}: There exists a uniform constant $C$ such that for any $0<r<R$, for all $\lambda \in I,\forall x\in X_\lambda$, a ball centered at $x$ of radius $R$ in $X_{\lambda}$ can always be covered by at most $C\cdot(\frac{R}{r})^{\beta}$ balls with radius $r$ within $X_{\lambda}$.
    
    We say that $\mathcal{X}$ has an Assouad dimension up to scale $\overline{R}$, if the covering condition is satisfied for any $0<r<R\leq \overline{R}.$
    \begin{align}
        \mathrm{dim}_{\mathrm{Assouad},\overline{R}}(\mathcal{X})=\inf_{\beta\in \mathbb{R}\cup \infty}\{\beta&\mid \exists C>0,\forall X_{\lambda}\in\mathcal{X}, 0<r<R\leq \overline{R}, x\in X_{\lambda},\exists \text{ set }\mathcal{F}\subseteq X_{\lambda},|\mathcal{F}|\leq C\cdot \left(\frac{R}{r}\right)^{\beta}\nonumber\\
        &\text{ such that } N_{x}(R)\subseteq\bigcup N_{\mathcal{F}}(r)\}.
    \end{align} 
    \begin{equation}
        \dimasd(\mathcal{X})=\sup_{\overline{R}\in \mathbb{R^{+}}} \mathrm{dim}_{\mathrm{Assouad},\overline{R}}(\mathcal{X})\in \mathbb{R}_{\geq 0}\cup\{+\infty\}.
    \end{equation}
    If the connected space $X_\lambda$ has unbounded cardinality as $\lambda$ varies over $ I$, then we always have $\dimasd(\mathcal{X})\geq 1$.
\end{definition}

The Assouad dimension is invariant under rescaling of metric, as proven below. The same argument also can be directly extended to proving that the Assouad dimension is stable under more general bi-Lipschitz maps.   
\begin{proposition} \label{factor multiplication}~\cite{fraser2020}
    Let $\mathcal{X}=\{(G_{\lambda}, d_{\lambda})\}_{\lambda\in I}$ be a family of metric spaces with $\dimasd(\mathcal{X})=\beta$. Let $\delta: I \to \mathbb{R}^+$ be a positive real-valued function, and define the rescaled family as $\widetilde{\mathcal{X}}=\{(G_{\lambda}, \delta(\lambda)d_{\lambda})\}_{\lambda\in I}$. Then, $\dimasd(\widetilde{\mathcal{X}}) = \beta$
\end{proposition}
 \begin{proof}
    Rescaling the metric by $\delta(\lambda)$ from $(r, R) \to (\delta(\lambda)r, \delta(\lambda)R)$, leaves the underlying covering sets invariant, so one can always find a finite $\delta(\lambda)r$ cover for any $N_x(\delta(\lambda)R)$. Because this is true for all $(r,R)$ pairs, the Assouad dimension of the rescaled family $\widetilde{\mathcal{X}}$ remains as $\beta$.
 \end{proof}
An undirected graph $G=(V,E)$ can also be viewed as a metric space, where the distance $d$ between two vertices is defined as:
 \[d(v_{i},v_{j})=|\text{the length of the shortest path between } v_{i},v_{j}|\]  The Assouad dimension of a family of undirected graphs is defined by the covering condition with respect to the balls of path distance. For a family of connectivity graphs, the following proposition~\ref{doubling condition} shows that having finite Assouad dimension means that $G(\mathcal{S})$ has bounded degree for each vertex.
\begin{proposition}\label{doubling condition}
    The family $\mathcal{X}$ has finite Assouad dimension if and only if it is \emph{doubling}, , that is there exists some  constant $C$ such that for any $R>0$, every ball of radius $2R$ in $\mathcal{X}$ can be covered by at most $C$ balls of radius $R$.
\end{proposition}
\begin{proof}
    The forward direction is obvious by taking $r = \frac{R}{2}$ for a given $R$ in the covering condition of Assouad dimension. For the converse, any ball of radius $R$ can be covered by at most $C^{\lceil\log_2{(\frac{R}{r})}\rceil}$ balls of radius $r$, after iteratively applying the doubling condition. Then $C^{\lceil\log_2{(\frac{R}{r})}\rceil}< C^{1+\log_2{(\frac{R}{r})}}=C\cdot (\frac{R}{r})^{\log_2C}$ implies that  $\dimasd\mathcal{X}\leq \log_2C<+\infty.$
\end{proof}
\begin{remark}
    The Assouad dimension captures both the local and macroscopic properties of any family of metric spaces. As a direct consequence of its definition, it provides an upper bound for the other fractal dimensions: \[\mathrm{dim}_{\mathrm{Hausdorff}}\mathcal{X}\leq \mathrm{dim}_{\mathrm{Minkowski}}\mathcal{X}\leq \dimasd \mathcal{X}\]
\end{remark}
 \begin{definition}[Ahlfors $\beta$-regularity on graphs]
    Given a family of undirected graphs $\mathcal{X}=\{G_{\lambda}\}_{\lambda\in I}$, we say that $\mathcal{X}$ has Ahlfors $\beta$-regularity if for any $q\in G_\lambda$, $R\leq \mathrm{diam}(G_{\lambda})$, the volume of $N_{q}(R)\subseteq G_{\lambda}$ is uniformly bounded by
     \begin{equation}
         v_{*}\cdot R^{\beta}\leq \mathrm{Vol}(N_{q}(R))\leq \max \{v^{*}R^{\beta},1\}
     \end{equation}
 \end{definition}
 \begin{proposition}\label{Ahlfors to Assouad}
     Let $\mathcal{X}$ be a family of graphs satisfying uniform Ahlfors $\beta$-regularity, then the Assouad dimension of $\mathcal{X}$ at all scales is $\beta$.
 \end{proposition}
\begin{proof}
     We first show that $\dim_{\mathrm{Assouad}}\mathcal{X}\geq \beta$. A ball of radius $R$ requires at least $v_{*}R^{\beta}=\frac{v_{*}}{2^{\beta}}\cdot (\frac{R}{\frac12})^{\beta}$ balls of radius $\frac12$ in any cover. Thus, we have $\dimasd\mathcal{X}\geq \beta$. For the converse, consider the two cases. First, when $v^{*}R^{\beta}<1$, $R<1,N_q(R)$ contains only a single vertex, so the covering condition is trivially satisfied. When $ v^{*}R^{\beta}\geq1,\mathrm{Vol}(N_{q}(R))\leq v^{*}R^{\beta}$, we consider the maximal disjoint packing of $N_{q}(R)$ with balls of radius $\frac{r}{2}$, $N_{q_i}(\frac{r}{2}):$
     $\bigcup\limits_{i=1}^{k>0}N_{q_i}(\frac{r}{2})\subseteq N_{q}(R), N_{q_i}(\frac{r}{2})\cap N_{q_j}(\frac{r}{2})=\emptyset;k \text{ maximal}$.
     By maximality, enlarging the radius from $\frac{r}{2}$ to $r$ forms a cover for $N_q(R)$: $N_q{(R)}\subseteq\bigcup\limits_{i=1}^{k>0}N_{q_i}(r)$. Taking volume, we obtain
     \begin{align}k v_{*}\left(\frac{r}{2}\right)^{\beta}\leq \mathrm{Vol}\left(N_{q}(R+\frac{r}{2})\right)\leq v^{*}\left(R+\frac{r}{2}\right)^{\beta}.\end{align}
     This implies $k\leq \frac{v^{*}}{v_{*}}(\frac{3R}{r})^{\beta}$, so the covering condition is satisfied with index $\beta$. We conclude that $\dimasd\mathcal{X}= \beta$.
\end{proof}

 \begin{proposition} \label{bi-Lipschitz reduction}
     Consider a family of continuous bijections $\{f_{\lambda}\}_{\lambda\in I}$ between $\mathcal{X}=\{(X_{\lambda},d_{\lambda})\}_{\lambda\in I}$ and $\mathcal{Y}=\{(Y_{\lambda},d^{\prime}_{\lambda})\}_{\lambda\in I}$, where $f_{\lambda}:X_{\lambda}\to Y_{\lambda}$ is uniformly bi-Lipschitz with respect to the corresponding metric,
     \begin{align} K_1 \cdot d_{\lambda}(x,x^{\prime})\leq d^{\prime}_{\lambda}(f_{\lambda}(x),f_{\lambda}(x^{\prime}))\leq K_2\cdot d_{\lambda}(x,x^{\prime}),\forall \lambda\in I, \forall x,x^{\prime}\in X_{\lambda}.\end{align}
     Then, the Assouad dimension is preserved by $\{f_{\lambda}\}$, i.e. $\dimasd\mathcal{X}=\dimasd\mathcal{Y}$.
 \end{proposition}
 \begin{corollary} \label{Tanner/connectivity unified}
     For any family of stabilizer codes, the Assouad dimension of their Tanner graphs equals the Assouad dimension of their connectivity graphs.
 \end{corollary}
 \begin{proof}
     Denote the family of Tanner graphs by $\mathcal{Y}$ and the family of connectivity graphs by $\mathcal{X}$. Consider the natural bijection between the vertices corresponding to physical qubits on the Tanner graph and that on the connectivity graph. Their distance differs by a factor of 2. By the bi-Lipschitz mapping in proposition~\ref{bi-Lipschitz reduction} to the subspace of $\mathcal{Y}$, $\mathcal{X}$ has smaller Assouad dimension than $\mathcal{Y}$: $\beta_{\mathcal{X}}\leq \beta_{\mathcal{Y}}$. When $\beta_{\mathcal{Y}}<+\infty$, we take $R=\frac12$. The doubling condition implies that the vertices in $\mathcal{Y}$ have bounded degree, i.e., the stabilizer code is qLDPC. The number of associated stabilizer generators (check nodes) is at most a constant multiple of the number of physical qubits, so the Assouad dimension remains unchanged by adding check nodes. When $\beta_{\mathcal{Y}}=+\infty$, we prove that $\beta_{\mathcal{X}}=+\infty$ by contradiction. Assume $\beta_{\mathcal{X}}<+\infty$, the vertices in $\mathcal{X}$ have bounded degree. If there exists a sequence of check nodes in $\mathcal{Y}$ that has unbounded degree, then the physical qubits associated with this sequence have unbounded degree in $\mathcal{X}$ as well (See definition~\ref{Connectivity graph}). Similarly, for any $x\in \mathcal{X}$, the associated check nodes (stabilizer generators) in $\mathcal{Y}$ are finite, because $x$ has bounded degree $\Delta$. More concretely, any adjacent check node (stabilizer generator) can only connect to the neighbors of $x$ in $\mathcal{X}$, so there are at most $2^{\Delta}-1$ possible connection configurations, and the number of associated check nodes is bounded. In other words, the original code family is a qLDPC code family. Again, adding only a constant number  of ancilla qubits to each physical qubit does not change the Assouad dimension, thus $\beta_{\mathcal{X}}=\beta_{\mathcal{Y}}=+\infty$, which contradicts our assumption $\beta_{\mathcal{X}}<+\infty$. In conclusion, we have $\beta_{\text{Tanner graph}}=\beta_{\text{connectivity graph}}$.
 \end{proof}

\subsection{Riemannian manifolds and fractal dimensions} \label{Appendix A.3}

In this subsection, we explain how Assouad dimension applies to compact Riemannian manifolds. The Riemannian metric induces a geodesic distance, so covering numbers and Assouad dimension can be defined using geodesic balls. We show that their Assouad dimensions agree with the manifolds' dimensions.

Let $(\mathcal{M}^{n},g)$ be a compact $n$-dimensional Riemannian manifold, where the Riemannian metric $g|_{p}: T_{p}\mathcal{M} \times T_{p}\mathcal{M} \to \mathbb{R}$ defines a global inner product on the tangent bundle. The distance between any two points $p, q \in \mathcal{M}$ is defined by the infimum of the lengths of all admissible curves connecting them:
\begin{equation}
    d_{g}(p,q) = \inf_{\gamma\in \mathcal{A}_{p,q}} \int_{0}^{1} \sqrt{g_{\gamma(t)}(\gamma'(t),\gamma'(t))} \, dt,
\end{equation}
where the family of admissible curves is given by $\mathcal{A}_{p,q} = \{\gamma \mid \gamma \text{ is a piecewise smooth curve with } \gamma'(t) \neq 0, \gamma(0)=p, \gamma(1)=q\}$. In our framework, this shortest path is exactly realized by the minimal geodesic with respect to the Levi-Civita connection. A closed geodesic ball of radius $r$ centered at $p \in \mathcal{M}$ is denoted by $N_{p}(r) = \{q \in \mathcal{M} \mid d_{g}(p,q) \leq r\}$.

We introduce the Ricci curvature at $p \in \mathcal{M}$:
\begin{align}
\begin{split}
    \mathrm{Ric}_{p}: T_{p}\mathcal{M} \to \mathbb{R}, \quad X \longmapsto \mathrm{Tr}_{Z\in T_p{\mathcal{M}}} (Z \to R(X,Z)X), \\
    \text{where } R(X,Y)Z = \nabla_{X}\nabla_{Y}Z - \nabla_{Y}\nabla_{X}Z - \nabla_{[X,Y]}Z.
\end{split}
\end{align}
The associated volume form, strictly compatible with $g$, is locally expressed in a chart $(U, x)$ as $d\mathrm{Vol}|_{U} = \sqrt{\det(g_{ij})} \, dx_1 \cdots dx_n$. 

We first show that any compact manifold satisfies Ahlfors regularity.

\begin{lemma}[Ahlfors regularity] \label{Volume on manifold}
    For a compact Riemannian manifold $(\mathcal{M},g)$, there exist two global constants $c_1, c_2 > 0$ such that for any geodesic ball $N_{p}(r)$ of radius $r\leq\mathrm{diam}(\mathcal{M})$ centered at an arbitrary point $p \in \mathcal{M}$, the volume satisfies:
    \begin{equation}
        c_1 \cdot r^{n} \leq \mathrm{Vol}(N_{p}(r)) \leq c_2 \cdot r^n.
    \end{equation}
\end{lemma}

While the local bi-Lipschitz equivalence between $\mathcal{M}$ and $\mathbb{R}^n$ intuitively suggests $r^n$ scaling for sufficiently small radius, estimating the volume of geodesic balls is in general nontrivial for non-compact manifolds. To proceed with the formal proof, we make use of the following Bishop--Günther volume comparison theorem.

\begin{theorem}[Bishop–Günther Comparison Theorem~\cite{Bishop1963,gunther1960einige}] \label{B-G Theorem}
    Let $(\mathcal{M}^n,g)$ be a Riemannian manifold with Ricci curvature bounded by $\mathrm{Ric}(\mathcal{M})\geq(n-1)ag$, and sectional curvature $
\mathrm{Sec }_{\mathcal{M}}\leq b$ for some constants $a,b \in \mathbb{R}$. Let $M_{k}^n$ denote the complete, simply connected $n$-dimensional space form with constant sectional curvature $k$. Then, the volume ratio function
    \begin{equation}
        \phi(r) = \frac{\mathrm{Vol}(N_{p}(\mathcal{M},r))}{\mathrm{Vol}(N_{p'}(M^{n}_a,r))} \quad (p\in \mathcal{M}, p'\in M^{n}_{a})
    \end{equation}
    is monotonically non-increasing on $(0,\mathrm{Inj}(\mathcal{M}))$, where $\mathrm{Inj}(\mathcal{M})$ is the injectivity radius. Conversely, the function
    \begin{equation}
        \varphi(r) = \frac{\mathrm{Vol}(N_{p}(\mathcal{M},r))}{\mathrm{Vol}(N_{p'}(M^{n}_{b},r))} \quad (p\in \mathcal{M}, p'\in M^{n}_{b})
    \end{equation}
    is monotonically non-decreasing on $(0,\mathrm{Inj}(\mathcal{M}))$.
\end{theorem}

\begin{proof}[Proof of Lemma~\ref{Volume on manifold}]
    Since $\mathcal{M}$ is compact, its injectivity radius is strictly positive, $\widetilde{r_0} \equiv \mathrm{Inj}(\mathcal{M}) > 0$. Applying Theorem~\ref{B-G Theorem}, and noting that $\phi(r), \varphi(r) \to 1$ as $r \to 0$, we obtain the bounds $\mathrm{Vol}(N_{p'}(M^{n}_{b},r)) \leq \mathrm{Vol}(N_{p}(\mathcal{M},r)) \leq \mathrm{Vol}(N_{p'}(M^{n}_{a},r))$ for all $0 < r < \widetilde{r_0}$. 
    
    According to the standard volume expansion for space forms~\cite{gray1974volume}, 
    \begin{equation}
        \mathrm{Vol}(N_p(M_{k}^n,r)) = \frac{\alpha_n \cdot r^n}{n} \left(1 - \frac{k}{6(n+2)}r^2 + o(r^2)\right), \quad \alpha_{n} = \frac{2\pi^{n/2}}{\Gamma(n/2)}.
    \end{equation}
    There exists a small scale $\epsilon > 0$ such that for $r < \epsilon$, the higher-order residual term is strictly bounded by $1/6$. By defining a critical radius $r_0 = \min\{\epsilon, \widetilde{r_0}, \sqrt{\frac{3(n+2)}{|a|}}, \sqrt{\frac{3(n+2)}{|b|}}\}$, we ensure that for all $r < r_0$:
    \begin{equation}
        \frac{\alpha_n}{2n}r^n < \mathrm{Vol}(N_p(M_{k}^n,r)) < \frac{5\alpha_n}{3n}r^n.
    \end{equation}
    For the regime $r \geq r_0$, the bounded diameter $d$ and total volume $\mathrm{Vol}(\mathcal{M})$ of the manifold allow us to control the growth:
    \[ 
        \left(\frac{\alpha_n r_0^n}{2n d^n}\right) r^n \leq \mathrm{Vol}(N_p(M_{k}^n,r_0)) \leq \mathrm{Vol}(\mathcal{M}) = \frac{\mathrm{Vol}(\mathcal{M})}{r_0^n}r_0^n \leq \frac{\mathrm{Vol}(\mathcal{M})}{r_0^n}r^n.
    \]
    Combining both regimes, the lemma holds globally by choosing the constants $c_1 = \frac{\alpha_n r_0^n}{2n d^n}$ and $c_2 = \max\left\{\frac{5\alpha_n}{3n}, \frac{\mathrm{Vol}(\mathcal{M})}{r_0^n}\right\}$. (In the following context, these bounding constants are denoted as $v_{n*}$ and $v_n^*$, respectively).
\end{proof}

\begin{remark}
    We note that establishing such bounds for non-compact complete manifolds with an appropriate metric can be highly nontrivial. While the Bishop--Gromov theorem provides precise upper bounds when $\mathrm{Ric}(\mathcal{M}) \geq 0$, ensuring a linear lower bound requires more delicate estimates, as demonstrated by Yau's result~\cite{Completemanifold73}.
\end{remark} 

Using the volume bounds in Lemma~\ref{Volume on manifold}, we now determine the Assouad dimension of $\mathcal{M}$.

\begin{definition}
    The Assouad dimension of a metric space $(\mathcal{M},g)$, denoted $\dim_{\mathrm{A}}\mathcal{M}$, is the infimum of all $\beta \in \mathbb{R} \cup \{\infty\}$ satisfying the covering condition for geodesic balls at all scales on $\mathcal{M}$.
\end{definition}

\begin{theorem} \label{Dimension unification}
    For any compact Riemannian manifold $\mathcal{M}$, its topological dimension equals its Assouad dimension at all scales, i.e., $\dim_{\mathrm{A}}\mathcal{M} = n$.
\end{theorem}

\begin{proof}
    Because $\mathcal{M}$ is locally bi-Lipschitz equivalent to $\mathbb{R}^{n}$, the standard metric definition immediately implies $\dim_{\mathrm{A}}\mathcal{M} \geq n$.
    
    To prove the upper bound, consider a target ball $N_{p}(R) \subset \mathcal{M}$. We construct a maximal disjoint packing of this region using smaller geodesic balls of radius $r/2$ centered within $N_p(R)$. Specifically, we select a maximal finite sequence $\{x_i\}_{i=1}^k$ in $N_p(R)$ such that $N_{x_i}(r/2) \cap N_{x_j}(r/2) = \emptyset$ for $i \neq j$, and $\bigsqcup_{i=1}^{k} N_{x_i}(r/2) \subset N_p(R + r/2)$.
    
    By the maximality of this packing, the enlarged balls of radius $r$ form a cover: $N_p(R) \subset \bigcup_{i=1}^{k} N_{x_i}(r)$. If this is not a cover, then there will be a point $q \in N_p(R)$ outside the union, allowing us to add $N_q(r/2)$ to our collection without intersecting the existing balls, contradicting the maximality assumption.
    
    Since the balls of radius $r/2$ are disjoint, their volumes add and is bounded by the volume of $N_p(R + r/2)$. Invoking the Ahlfors regularity from Lemma~\ref{Volume on manifold}, we have:
    \begin{equation}
        k \cdot c_1 \left(\frac{r}{2}\right)^n \leq \mathrm{Vol}\left(\bigsqcup_{i=1}^{k} N_{x_i}\left(\frac{r}{2}\right)\right) \leq \mathrm{Vol}\left(N_p\left(R + \frac{r}{2}\right)\right) \leq c_2 \left(R + \frac{r}{2}\right)^n.
    \end{equation}
    Rearranging this inequality yields an upper bound on the covering number $k$:
    \begin{equation}
        k \leq \frac{c_2 2^n}{c_1} \left(\frac{R}{r} + \frac{1}{2}\right)^n < \frac{c_2 2^n}{c_1} \left(\frac{3R}{2r}\right)^n = \frac{c_2 3^n}{c_1} \left(\frac{R}{r}\right)^n.
    \end{equation}
    Since $N_p(R)$ can always be covered by $C(R/r)^n$ balls of radius $r$ (with $C = c_2 3^n / c_1$), the  Assouad dimension of the manifold is bounded from above by $n$. Thus, $\dim_{\mathrm{A}}\mathcal{M} = n$.
\end{proof}
    
\section{Dimensions of topological codes}\label{Appendix Dimension unification}

In this appendix, we define general topological stabilizer codes and show that for such codes, the Assouad dimension exactly coincides with the topological dimension of the underlying triangulated manifold. Furthermore, we show how to define codes on a quasi-convex self-similar set $K$, such that they have the same dimension as the Assouad dimension of $K$.

\subsection{Construction of general topological stabilizer codes}
Given any $D$-dimensional Riemannian manifold $(\mathcal{M}^D,g)$, the strong Whitney embedding theorem guarantees its smooth embedding within an ambient Euclidean space $\mathbb{R}^{p}$ where $p=2D$~\cite{whitney1944self}. Consequently, classic results by Whitehead~\cite{whitehead1940c1} ensure the existence of a well-defined and combinatorially equivalent triangulation for $\mathcal{M}$ in the ambient space $\mathbb{R}^{p}$. This extrinsic triangulation does not strictly respect the original metric. To better preserve the intrinsic geometric data of the manifold, several improved triangulation algorithms have been proposed~\cite{liu2015efficient,boissonnat:hal-01509888,sharp2019navigating}. 

We now introduce two key ingredients for a good triangulation, which measure the ``worst'' behavior of simplexes in a triangulation. Let the length of the shortest edge in the triangulation be $\delta$, and suppose each simplex $\sigma$ in the triangulation can be bounded by a circumscribed ball of radius $R(\sigma)$. The radius-edge ratio of a triangulation $\mathcal{C}^{\bullet}$ is defined as:
\begin{equation}
    \rho(\mathcal{C}^{\bullet}) = \sup_{\sigma\in\mathcal{C}^{\bullet}} \frac{R(\sigma)}{\delta}.
\end{equation}
This ratio enforces the triangulation to be uniform at scale $\delta$, whose longest edge of simplex has to be $O(\delta)$. The other ingredient is the concept of the thickness. 
\begin{definition}[Uniform thickness]
    For any $k$-simplex $\sigma\in \mathcal{C}^k_{\mathcal{M},\delta}$, let  $\Delta(\sigma)\coloneqq \max\limits_{p,q\in \sigma} d(p,q)$ denote the length of its longest edge. For each $p\in\mathrm{Vert}(\sigma)$, let $
\sigma_p\coloneqq \mathrm{conv}\bigl(\mathrm{Vert}(\sigma)\setminus\{p\}\bigr)$
be the facet opposite $p$, and define the altitude of $p$ in
$\sigma$ by $h(p,\sigma)
\coloneqq
\mathrm{dist}\bigl(p,\mathrm{aff}(\sigma_p)\bigr)$. Then the thickness of $\sigma$ is defined by 
    \[
    \Upsilon(\sigma)
        \coloneqq
        \begin{cases}
            1, & k=0, \\[2mm]
            \displaystyle
            \min_{p\in\operatorname{Vertex}(\sigma)}
                \frac{h(p,\sigma)}
     {k\,\Delta(\sigma)},
& k\geq 1,
\end{cases}
\]
 where the $d$ is the standard Euclidean distance associated with the simplex. We say the triangulation is uniformly thick if $\Upsilon(\sigma)\geq \Upsilon_0>0$ for any $D$-simplex $\sigma\in \mathcal{C}^{\bullet}_{\mathcal{M},\delta}$. Any $k$-simplex $(k\leq D)$ in the triangulation has thickness $\geq \Upsilon_0$ if the condition is satisfied for $D$-simplex.
\end{definition}
These two quantities bound the smallest solid angle of a triangulation and is of fundamental importance in numerical analysis, particularly in the finite-element method, serving as a primary benchmark for the quality of triangulation.

\begin{definition}[Good triangulation]\label{Good triangulation}
    For a Riemannian manifold $(\mathcal{M}^{D},g)$, a family of triangulations $\{\mathcal{C}^{\bullet}_{\mathcal{M},\delta} \mid \delta \leq \delta_0\}$ at sufficiently small scales is a \emph{good triangulation} if it satisfies the following conditions:
    \begin{itemize}
        \item The radius-edge ratio $\rho(\mathcal{C}^{\bullet}_{\mathcal{M},\delta})$ is uniformly bounded from above by a constant $\rho_0$.

        \item The triangulation is uniformly thick $\Upsilon(\sigma)\geq\Upsilon_0>0$
        for every $D$-simplex $\sigma\in\mathcal C^\bullet_{\mathcal M,\delta}$.

        \item There exists a family of bi-Lipschitz homeomorphisms $\{H_{\delta}\}$ mapping $\mathcal{C}^{\bullet}_{\mathcal{M},\delta}$ to $\mathcal{M}$ at scale $\delta$, such that for all $x,y \in \mathcal{C}^{\bullet}_{\mathcal{M},\delta}$:
        \begin{equation}
            K_1 \cdot d_{PL}(x,y) \leq d_{g}(H_{\delta}(x),H_{\delta}(y)) \leq K_2 \cdot d_{PL}(x,y),
        \end{equation}
        where $K_1, K_2 > 0$ are constants independent of $\delta$, and $d_{PL}$ is the piecewise-linear Euclidean metric restricted to the skeleton of $\mathcal{C}^{\bullet}_{\mathcal{M},\delta}$.
    \end{itemize}
\end{definition}

For any compact $C^2$-differentiable manifold, there always exists a $(\delta,\epsilon)$-sampling set. Based on this sampling, Boissonnat, Dyer, and Ghosh~\cite{boissonnat2010manifold} demonstrated that a triangulation with a globally bounded radius-edge ratio is guaranteed for any sufficiently small $\delta$. Furthermore, \cite{boissonnat2010manifold} proved the existence of a family of restricted Delaunay triangulations whose associated homeomorphisms exhibit strong bi-Lipschitz properties between $d_{g}$ and $d_{PL}$~\cite{boissonnat:hal-01509888}.

\begin{theorem}[Riemannian Delaunay triangulation~\cite{boissonnat:hal-01509888}] \label{Delaunay triangulation}
    Let $\mathcal{M}^{m}$ be a Riemannian manifold, and $P \subset \mathcal{M}^{m}$ be a $(\mu_0,\epsilon)$-net with respect to $d_{g}$, where 
    \[ \epsilon \leq \min\left\{\frac{\mathrm{Inj}(\mathcal{M})}{4}, \frac{1}{2^6\sqrt{\Lambda}} \left(\frac{\rho_0}{\widetilde{C}}\right)^{2m+1}\right\}. \]
    Here, $\Lambda$ bounds the absolute value of the sectional curvatures, $\mathrm{Inj}(\mathcal{M})$ is the injectivity radius, and $\widetilde{C} = m^{3/2}(2/\mu_0)^{5m^2+5m+21}$. If $\rho_0 \leq \mu_0/5$, the output $\mathrm{Del}(P')$ of the extended algorithm is a Delaunay triangulation equipped with a natural homeomorphism $H: |\mathrm{Del}(P')| \to \mathcal{M}$ satisfying:
    \begin{equation}
        |d_{g}(H(x),H(y)) - d_{PL}(x,y)| \leq \left(2^8\Lambda \left(\frac{\widetilde{C}}{\rho_0}\right)^{2m} \epsilon^2\right) d_{PL}(x,y).
    \end{equation}
\end{theorem}

Because the protected manifold Delaunay triangulations constructed by
the above theorem are stable under controlled perturbations
and enjoy both uniform thickness and a uniformly bounded global
radius-edge ratio~\cite{boissonnat:hal-01509888}, they serve as a
canonical model for our \emph{good triangulations}.

 To construct a homological CSS code, we utilize the simplicial complex induced by a triangulation. The boundary operator $\partial_k$ acts on an labeled $k$-simplex by taking the alternating sum of its $(k-1)$-dimensional faces, explicitly given by the chain complex: 
\begin{equation}
    \dots \xrightarrow{\partial_{k+1}} [t_0, t_1, \dots, t_k] \xrightarrow{\partial_k} \sum_{i=0}^k (-1)^i [t_0, \dots, \hat{t}_i, \dots, t_k] \xrightarrow{\partial_{k-1}} \dots,
\end{equation}
where $\hat{t}_i$ denotes the omission of the $i$-th vertex. The composition of two consecutive boundary maps vanishes $\partial_{k-1} \circ \partial_k = 0$.

Assuming the base field $\mathbb{F}$ has characteristic 2, we can take any three consecutive terms from this chain complex to define a CSS code, where the parity-check matrices are given by their boundary maps:
\begin{equation}
   \mathcal{C}^{k+1} \xrightarrow{H_{Z}^{T}} \mathcal{C}^{k} \xrightarrow{H_X} \mathcal{C}^{k-1}.
\end{equation}
The logical code space is isomorphic to the $k$-th homology group. Any sequence of three consecutive terms in a $D$-dimensional triangulation defines a valid $D$-dimensional topological code (e.g., the 3D cubic toric code and its dual). In other words, the geometric dimension $D$ of the triangulated manifold naturally defines the topological dimension of the code, coinciding with the highest nontrivial homological degree.

\subsection{Dimension unification for regular topological codes}\label{Dimension_unification_for_classical_code}
 In this subsection, we will show that, for a regular topological code derived by a family of good triangulations on a compact connected manifold $\mathcal{M}$, the Assouad dimension of its connectivity graph coincides with the topological dimension of the ambient manifold.
To prove the uniform bi-Lipschitz property between the ambient manifold and the connectivity graph, we establish the following proposition to capture the microscopic behaviour of good triangulation.

\begin{proposition}\label{local behavior of D-Triangulation}
    Given a family of good triangulations $\mathcal{C}^{\bullet}_{\mathcal{M},\delta}$ with bounded radius-edge ratio $\rho_0$ and bi-Lipschitz homeomorphisms $H_{\delta}$, there exists a local scale $r_0 = 4K_2\rho_0\delta$ such that for sufficiently small $\delta$:
    \begin{itemize}
        \item Any geodesic ball of radius $r_0$ on $\mathcal{M}$ intersects at most $N_0$ $k$-simplexes in $\mathcal{C}^{k}_{\mathcal{M},\delta}$ under the image of $H_{\delta}$, where $N_0$ is independent of $\delta$ and of the specific center of the geodesic ball.
        \item For any fixed $\sigma_i \in \mathcal{C}^k_{\mathcal{M},\delta}$, all of its adjacent $k$-simplexes on the connectivity graph are strictly contained within some pulled-back neighborhood $H_{\delta}^{-1}(N_{x}(r_0))$, where $x \in H_{\delta}(\sigma_i)$.
    \end{itemize}
\end{proposition}

\begin{proof}
    Without loss of generality, we define the local scale as $r_0 = c_1 \delta$ (with $c_1$ to be determined). The inverse image of a geodesic ball $N_{x}(r)$ on $\mathcal{M}$ is entirely covered by a ball of radius $r / K_1$ in $\mathcal{C}^{\bullet}_{\mathcal{M},\delta}$. Therefore,  the volume of the covering ball on the triangulation is at most $v^{*}_D (r/K_1)^D$. Meanwhile, the volume of any single $D$-dimensional simplex $\sigma$ is bounded from below by an iterated reduction to the volume of its sub-simplex:
    \begin{equation} \mathrm{Vol}_{D}(\sigma)=\frac{h(p,\sigma)}{D}\mathrm{Vol}_{D-1}(\sigma^{\prime})\geq\frac{h(p,\sigma)}{D\Delta(\sigma)}\mathrm{Vol}_{D-1}(\sigma^{\prime})\delta\geq \Upsilon_0 \delta\cdot\mathrm{Vol}_{D-1}(\sigma^{\prime})\geq \Upsilon_0^{D-1}\delta^D\eqqcolon V_0. 
    \end{equation}
    For any simplex $\sigma\bigcap H_{\delta}^{-1}\left(N_{x}(r_0)\right)\neq \emptyset$, the image $H_{\delta}(\sigma)\subseteq N_{x}((c_1+2K_2\rho_0)\delta)$.
    Because the interiors of distinct simplexes are mutually disjoint, the volume comparison limits the maximum number of simplexes within the ball to $\frac{v^{*}_D (c_1+2K_2\rho_0)^D} {\Upsilon_0^{D-1}K_1^{D}}$. Since each $D$-simplex is associated with $\binom{D+1}{k+1}$ $k$-simplexes, the first claim holds for $N_0=\binom{D+1}{\lfloor\frac{D+1}{2}\rfloor}\frac{v^{*}_D (c_1+2K_2\rho_0)^D} {\Upsilon_0^{D-1}K_1^{D}}$.

    For the second claim, two adjacent $k$-simplexes on the connectivity graph must share either an $X$-check ($(k-1)$-simplex) or a $Z$-check ($(k+1)$-simplex). Thus, they are either glued by a common $(k+1)$-simplex or share a common vertex (bounded by $2\Delta(\sigma) \leq 4\rho_0\delta,\sigma\in \mathcal{C}^k_{\mathcal{M},\delta}$). Therefore, setting $c_1 = 4K_2\rho_0$ ensures containment, completing the proof.
\end{proof}

\begin{lemma}[Uniform Assouad dimension under net approximation]
\label{Uniform net approximation lemma}
Let $(X,d)$ be a compact metric space, and let
$I\subset(0,\delta_0]$ be a set with $0$ as an
accumulation point. Fix constants $c_1,c_2>0$. For every
$\delta\in I$, let $X_{\delta}\subseteq X$ be a finite $\delta$-net
satisfying the separetion
\begin{equation}\label{eq:net-separation}
    d(u,v)\geq c_1\delta,
    u,v\in X_{\delta},\, u\neq v,
\end{equation}
and the density condition
\begin{equation}\label{eq:net-density}
    d(x,X_{\delta})\leq c_2 \delta,\,
    x\in X.
\end{equation}
Equip $X_{\delta}$ with the metric inherited from $X$, and define $\mathcal{X}
    \coloneqq
    \{(X_{\delta},d|_{X_{\delta}})\}_{\delta\in I}$.
Then
\begin{equation}\label{eq:net-Assouad-equality}
    \dimasd\mathcal{X}=\dimasd X.
\end{equation}
\end{lemma}

\begin{proof}
 As $\{X_{\delta}\}_{\delta\in I}$ are subspaces of $X$ with induced metrics, any admissible Assouad exponent for $X$ is also admissible for $\mathcal{X}$. Hence, we have $\dimasd \mathcal{X} \leq \dimasd X$.

For the reverse inequality, suppose $\alpha<\dim_{\mathrm{Assouad}}X$ is an admissible
Assouad exponent for $\mathcal{X}$, with uniform covering
constant $C_{\alpha}$. Fix $x\in X$ and $0<r<R$. Since $0$ is an
accumulation point of $I$, choose $\delta\in I$ such
that $c_2\delta<\frac{r}{8}$. By density condition Eq.~\eqref{eq:net-density}, one can choose $x_\delta\in X_\delta$
with $d(x,x_\delta)\leq c_2\delta$. Similarly, for every $z\in N_{x}(R)$, we choose $z_\delta\in X_\delta$ satisfying $d(z,z_\delta)\leq c_2\delta$. Then by the triangle inequality, we have 
\begin{equation}
    d(z_\delta,x_\delta)\leq d(z_\delta,z)+d(z,x)+d(x,x_\delta)\leq R+2c_2\delta<2R
\end{equation}
and $z_\delta\in N_{x_\delta}(2R)$. The latter ball can be covered by at most $C_{\alpha}\left(\frac{2R}{r/2}\right)^\alpha=4^\alpha C_{\alpha}\left(\frac{R}{r}\right)^\alpha$ balls $N_{u_i}(r/2)$, where $u_i\in X_{\delta}$. If
$z_\delta\in N_{u_i}(r/2)$ on $X_{\delta}$, then $d(z,u_i)\leq d(z,z_\delta)+d(z_\delta,u_i)\leq c_2\delta+\frac{r}{2}<r$ on $X$.
It follows that the corresponding balls $N_{u_i}(r)$ cover
$N_{x}(R)$ on $X$. Hence $\alpha<\dim_{\mathrm{Assouad}}X$ is also an admissible exponent for
$X$, and therefore
 $\dim_{\mathrm{Assouad}}X\leq\dim_{\mathrm{Assouad}}\mathcal{X}$.
This completes the proof of Eq.~\eqref{eq:net-Assouad-equality}.
\end{proof}

With these preparations, we present our core unification theorem.

\begin{theorem} \label{Dimension unification for code}
    Given a family of good triangulations of $(\mathcal{M}^{D},g)$ as $\delta \to 0$, the Assouad dimension of the connectivity graph of the resulting topological code $\{C_{\delta}\}$ exactly equals the Assouad dimension of the underlying manifold $\mathcal{M}$, i.e. \[\dim(C_{\delta}) = \dim(\mathcal{M}) = D.\]
\end{theorem}

\begin{proof}
    By Corollary~\ref{factor multiplication}, scaling the metric on $G_{\delta}$ does not alter its Assouad dimension. We define the discrete graph distance as $d_{G_{\delta}}(q_i,q_j) = \delta \cdot |\text{the shortest  path  between }q_i,q_j|$. We trace the sequence of mappings:
    \[ (\mathcal{M},d_{g}) \xrightarrow{H_{\delta}^{-1}} (\mathcal{C}^{\bullet}_{\mathcal{M},\delta},d_{PL}) \xrightarrow{\mathcal{S}_{\delta}} (G_{\delta},d_{\delta}). \]
    Let $\mathcal{M}_{\delta}(\mathcal{C}^{k}_{\mathcal{M},\delta})\coloneqq \{H_{\delta}(b(\sigma))\mid \sigma\in \mathcal{C}^{k}_{\mathcal{M},\delta}\}$ denote the discrete net determined by the images of barycenters $b(\sigma)$ of $k$-simplex $\sigma$ under $H_{\delta}$, where the distance is induced by the metric $d_{g}$ on $\mathcal{M}$.
    One can show that for two distinct $k$-simplexes $\sigma_i, \sigma_j$, the distance between their barycenters is lower bounded by $d_{PL}(b(\sigma_i),b(\sigma_j))\geq \frac{D\Upsilon_0}{k+1}\delta$, which implies the separation condition Eq. (\ref{eq:net-separation}) on $\mathcal{M}_{\delta}\colon d_{g}(H_{\delta}(b(\sigma_i)),H_{\delta}(b(\sigma_j))) \geq \frac{K_1D\Upsilon_0}{k+1}\delta$.
    Furthermore, the density condition Eq. (\ref{eq:net-density}) is satisfied because for any point $x\in \mathcal{M}$, there exists a $k$-simplex $\sigma$ such that $d_{g}(x,H_{\delta}(b(\sigma)))\leq r_0=4K_2\rho_0\delta$.
    So $M_{\delta}$ is indeed a discrete $\delta$-net of $\mathcal{M}$. By Lemma~\ref{Uniform net approximation lemma}, the Assouad dimension of $\mathcal{M}_{\delta}(\mathcal{C}^{k}_{\mathcal{M},\delta})$ is equal to that of $\mathcal{M},\, \dimasd(\mathcal{M}_{\delta}(\mathcal{C}^{k}_{\mathcal{M},\delta})) = \dimasd(\mathcal{M})=D$.
    By Proposition~\ref{bi-Lipschitz reduction}, it suffices to show that $\mathcal{S}_{\delta} \circ H_{\delta}^{-1}|_{\mathcal{M}_{\delta}}$ is uniformly bi-Lipschitz between $(\mathcal{M}_{\delta},d_g)$ and $(G_{\delta},d_{\delta})$ for all sufficiently small $\delta$.

    Consider two physical qubits corresponding to $k$-simplexes $\sigma_i, \sigma_j$, connected by a graph path of length $L$. By the second property in proposition~\ref{local behavior of D-Triangulation}, each $H_{\delta}(b(\sigma))$ can be connected to the next adjacent barycenter $H_{\delta}(b(\sigma'))$ by a geodesic of length at most $r_0$. Therefore, the corresponding geodesic distance on $\mathcal{M}_{\delta}$ is tightly upper bounded by: 
    \begin{equation}\label{ineq:manifold biLip_1}
        d_{g}(H_{\delta}(b(\sigma_i)),H_{\delta}(b(\sigma_j))) \leq r_0L = 4K_2\rho_0\cdot d_{\delta}(\sigma_i,\sigma_j).
    \end{equation}
    
    Conversely, for any two points $x=H_{\delta}(b(\sigma_i)),y = H_{\delta}(b(\sigma_j)) \in \mathcal{M}_{\delta}\,(i \neq j)$, we may find a geodesic connecting them with length $d_{g}(x,y) = L\geq \frac{K_1D\Upsilon_0}{k+1}\delta$. After an arbitrarily small general-position perturbation, the length of geodesic differs from $L$ by at most a uniform factor $c_0$, whose pullback image under $H_\delta^{-1}$ crosses a sequence of $D$-simplexes $T_0,T_1,\ldots,T_m$, where $\sigma_i\preceq T_0$, $\sigma_j\preceq T_m$, and $T_{\ell-1}\cap T_\ell$ is a common $(D-1)$-face. Then we apply Proposition~\ref{local behavior of D-Triangulation} with $k=D$, and cover the geodesic by balls of radius $r_0$. It gives $m+1\leq N_0\left(\frac{c_0 L}{r_0}+2\right)$. For each $1\leq\ell\leq m$, choose a $k$-face$\eta_\ell\preceq T_{\ell-1}\cap T_\ell$. The $k$-faces contained in any fixed $D$-simplex form the connected Johnson graph $J(D+1,k+1)$ under connectivity-graph adjacency. Define its diameter by $J_{D,k}\coloneqq\operatorname{diam}J(D+1,k+1)=\min\{k+1,D-k\}$, which depends only on $D$ and $k$. Hence $\sigma_i$ can be connected to $\sigma_j$ by a path in $G_\delta$ of length at most $(m+1)J_{D,k}$.
    
    This observation yields a uniform upper bound on the graph distance:
    \begin{align}
        d_{\delta}(\sigma_i,\sigma_j) &\leq \delta \cdot J_{D,k}N_0 \left(\frac{c_0L}{r_0} + 2\right) = J_{D,k}N_0 \left(\frac{c_0L}{4K_2\rho_0} + 2\delta\right)\\
        &\leq J_{D,k}N_0\left[\frac{c_0L}{4K_2\rho_0}+\frac{2(k+1)}{K_1D\Upsilon_0}L\right]=\left[\frac{c_0}{4K_2\rho_0}+\frac{2(k+1)}{K_1D\Upsilon_0}\right]J_{D,k}N_0\cdot d_g(x,y).\label{ineq:manifold biLip_2}
    \end{align}

    Combining \eqref{ineq:manifold biLip_1} and~\eqref{ineq:manifold biLip_2} together demonstrate that $\mathcal{S}_{\delta} \circ H_{\delta}^{-1}$ provides a uniform bi-Lipschitz map between $(G_{\delta},d_{\delta})$ and $(\mathcal{M}_{\delta},d_g)$. Hence, the induced Assouad dimension of $\{G_{\delta}\}_{\delta\in I}$ agrees with the Assouad dimension of $\{\mathcal{M}_{\delta}\}_{\delta\in I}$, which in turn agrees with the topological dimension $D$.
\end{proof}

\subsection{Dimension unification for fractal stabilizer codes} \label{Appendix Dim unification fractal case}
   We start by assuming $K$ is a quasi-convex self-similar set on manifold $\mathcal{M}$, with induced metric on a rectifiable path connecting any two points on $K$. For readers who may be interested, Proposition 2.10 in Ref.~\cite{tyson2002conformal} and 2.15 in Ref.~\cite{heinonen1998quasiconformal} provide sufficient conditions for self-similar $K$ to be quasi-convex. Recall that the Assouad dimension of $K$ is defined by the covering condition for a geodesic ball of radius $r$ restricted on $K$ $N_{K}(r)\stackrel{\text{def}}{=}N_{\mathcal{M},g}(r)\bigcap K$. 
\begin{lemma}[Assouad dimension under controlled approximation]
\label{Controlled approximation lemma}
Let $(\mathcal{M},d)$ be a metric space and let $K\subseteq \mathcal{M}$ be compact.
Let $I\subset(0,\delta_0]$ have $0$ as an accumulation
point. For every $\delta\in I$, let
$K_{\delta}\subseteq \mathcal{M}$ be a finite set. Suppose that:

\begin{enumerate}
    \item There exists $A>0$, independent of $\delta$, such that
    \begin{equation}\label{eq:controlled-Hausdorff-distance}
        \sup_{w\in K_{\delta}}d(w,K) \leq A\delta,\quad  \sup_{x\in K}d(x,K_{\delta}) \leq A\delta.
    \end{equation}

    \item The family has uniformly bounded local multiplicity
    at scale $\delta$: for every fixed $t>0$, there exists
    $C(t)<+\infty$, independent of $\delta$ and $z\in \mathcal{M}$, such
    that
    \begin{equation}\label{eq:uniform-local-multiplicity}
        \#\bigl(K_\delta\cap N_{z}(t\delta)\bigr)\leq C(t).
    \end{equation}
\end{enumerate}

Then, with each $K_\delta$ equipped with the metric inherited
from $\mathcal{M}$,
\begin{equation}
    \dimasd \{(K_\delta,d|_{K_\delta})\}_{\delta\in I}=\dimasd K.
    \label{eq:controlled-Hausdorff-Assouad}
\end{equation}
\end{lemma}
\begin{proof}
 Let $\dimasd K=\beta$. We first show that $\dimasd\{K_\delta\}\leq \beta$. 
 
 Fix $\alpha>\beta$, and let
$C_\alpha$ be a covering constant for $K$ with Assouad exponent $\alpha$. We will show that there exists a uniform constant works for $K_{\delta}$ with exponent $\alpha$ as well. Fix $\delta\in I$, $v\in K_\delta$, and $0<r<R$.

\emph{Case 1: $r\geq 4A\delta$.}
Choose $x_v\in K$ such that $d(v,x_v)\leq A\delta$. For every
$w\in N_{v}(R)\bigcap K_{\delta}$, choose $x_w\in K$ satisfying
$d(w,x_w)\leq A\delta$. Then by $d(x_w,x_v)\leq R+2A\delta\leq R+\frac{r}{2}<\frac32 R$, we know the corresponding points $x_w$ lie in
$N_{x_v}(\frac{3R}{2})\bigcap K$. We cover $N_{x_v}(\frac{3R}{2})$ by balls of radius $r/4$, and the number
required is at most $C_\alpha\left(\frac{3R/2}{r/4}\right)^\alpha=6^\alpha C_\alpha\left(\frac{R}{r}\right)^\alpha$.
If $x_w$ belongs to one such ball centered at $y_i\in K$,
then $d(w,y_i)\leq d(w,x_w)+d(x_w,y_i)\leq A\delta+\frac{r}{4} \leq\frac{r}{2}$.
For every covering ball that contains at least one point of
$K_{\delta}$, choose such a point $w_i\in K_{\delta}$. Then all
points in $K_{\delta}$ assigned to that ball lie in $N_{w_i}(r)$.
This proves the required covering estimate in Case 1.

\emph{Case 2: $r<4A\delta$.} We cover $N_{x_v}(R+2A\delta)$ by balls of radius $\delta$. When $R<4A\delta$, the local multiplicity condition immediately gives
a uniform bound $C(6A)$. When $R\geq 4A\delta$, the number of $\delta$-balls
is at most $C_{\alpha}\left(\frac32\right)^{\alpha}\left(\frac R\delta\right)^{\alpha}$. Every point of $K_{\delta}\cap N_{v}(R)$ lies within the $A\delta$-neighborhood of one of these balls. By
Eq.~\eqref{eq:uniform-local-multiplicity}, each corresponding
ball of radius $(A+1)\delta$ contains at most $C(A+1)$ points
of $K_{\delta}$. Covering these points individually by
$r$-balls gives at most
\begin{equation}
    C(A+1)C_{\alpha}\left(\frac32\right)^{\alpha}
    \left(\frac R\delta\right)^\alpha<C(A+1)C_{\alpha}\left(6A\right)^{\alpha}
    \left(\frac Rr\right)^\alpha
\end{equation}
balls. Hence $\dimasd\{K_\delta\}\leq\dimasd K$. 

For the reverse inequality, suppose that $\alpha$ is an
admissible exponent for the family $\{K_\delta\}_{\delta\in I}$, with
uniform covering constant $C_{\alpha}$. Fix $x\in K$ and $0<r<R$, we choose $\delta\in I$ sufficiently small that $A\delta<\frac{r}{8}$.
By Eq.~\eqref{eq:controlled-Hausdorff-distance}, one may choose
$v_x\in K_{\delta}$ with $d(x,v_x)\leq A\delta$. For every
$y\in N_{x}(R)$, we choose $v_y\in K_{\delta}$ with
$d(y,v_y)\leq A\delta$. Then by $d(v_y,v_x)\leq R+2A\delta<2R$, we can cover $N_{v_x}(2R)$ by at most
$C_{\alpha}\left(\frac{2R}{r/2}\right)^{\alpha}=4^{\alpha} C_{\alpha}\left(\frac Rr\right)^\alpha$
balls $N_{w_i}(r/2),\, w_i\in K_{\delta}$. For each center $w_i$, choose $z_i\in K$ satisfying
$d(w_i,z_i)\leq A\delta$. If $v_y\in N_{w_i}(r/2)$, then $d(y,z_i)\leq d(y,v_y)+d(v_y,w_i)+d(w_i,z_i)\leq 2A\delta+\frac{r}{2}<r$.

Therefore at most $4^{\alpha} C_{\alpha}\left(\frac Rr\right)^\alpha$ balls $N_{z_i}(r)$ cover $N_{x}(R)$. Hence $\dimasd K\leq\dimasd\{K_\delta\}$.
These two inequalities prove the Eq.~\eqref{eq:controlled-Hausdorff-Assouad}.
\end{proof}

By leveraging the barycenters net of the simplexes in the triangulation $C^{k}_{\mathcal{M},\delta}(K)$, we can construct a controlled approximation $K_{\delta}\coloneqq \{H_{\delta}(b(\sigma))\mid \sigma\in C^{k}_{\mathcal{M},\delta}(K)\}$ of the quasi-convex set $K$ on the manifold $\mathcal{M}$ at sufficiently small scale $\delta$. By constructing a family of uniform bi-Lipschitz maps, the following argument establishes the unification between Assouad dimension of the connectivity graph of the resulting fractal code $\{G_{\delta}(K,k)\}_{\delta\in I}$ and the Assouad dimension of $\{K_{\delta}\}_{\delta\in I}$. Hence, we can conclude that
\begin{equation}
    \dimasd\{G_{\delta}(K,k)\}_{\delta\in I} = \dimasd\{K_{\delta}\}_{\delta\in I} = \dimasd K.
\end{equation}
    \begin{proof}
        We first show that the barycenters net $\{K_{\delta}\}_{\delta\in I}$ satisfies the conditions of Lemma~\ref{Controlled approximation lemma}. The first condition is satisfied because for any $w=H_{\delta}(b(\sigma))\in K_{\delta}$, we have $\sigma\preceq\tau\in \mathcal{C}^D_{\mathcal{M},\delta},\tau\bigcap H_{\delta}^{-1}(K)\neq \emptyset$, which implies 
        \begin{equation}
            d_g(w,K)\leq ||H_{\delta}|| \Delta(\tau) \leq (2K_2\rho_0) \delta.
        \end{equation}
        And for any $x\in K$, there exists a simplex $\tau\in\mathcal{C}_{\mathcal{M},\delta}^{D}$, such that $x\in H_{\delta}(\tau)$. Therefore, there exists some $k$-simplex $\sigma\preceq\tau$, which implies
        \begin{equation}
            d_g(x,K_{\delta})\leq d_{g}(x,H_{\delta}(b(\sigma)))\leq ||H_{\delta}||\Delta(\tau)\leq (2K_2\rho_0)\delta.
        \end{equation}
        Therefore, we may let $A=2K_2\rho_0$ in Lemma~\ref{Controlled approximation lemma}.
        
        The second condition is satisfied due to the separation property of berycenters net $d_g(H_{\delta}(b(\sigma_i)),H_{\delta}(b(\sigma_j)))\geq \frac{DK_1\Upsilon_0}{k+1}\delta$, which can be derived by the covering condition for some finite Assouad exponent $\alpha$ on $\mathcal{M}$.
        
        Then we apply a similar strategy in Theorem~\ref{Dimension unification for code} here, to trace the representative map from $K$ to its associated connectivity graph, which is uniformly bi-Lipschitz between $\{K_{\delta}\}_{\delta\in I}$  and $\{G_{\delta}(K,k)\}_{\delta\in I}$ at any sufficiently small scale $\delta$, 
        \[ K\subset\mathcal{M}\xrightarrow{H_{\delta}^{-1}|_{K}}\mathcal{C}^{\bullet}_{\mathcal{M},\delta}(K)\xrightarrow{\mathcal{S}_{\delta}}G_{\delta}(K,k).\]
        The graph distance on $G_{\delta}(K,k)$ is defined as $d_{\delta}(q_i,q_j) = \delta \cdot |\text{the shortest  path  between }q_i,q_j|$. 
        By Corollary~\ref{factor multiplication}, rescaling the metric by a constant factor does not change its Assouad dimension. 

        For $w=H_{\delta}(b(\sigma_i)),v=H_{\delta}(b(\sigma_j))\in K_{\delta}$ with representing qubits as $q_i=\mathcal{S}_{\delta}(\sigma_i), q_j=\mathcal{S}_{\delta}(\sigma_j)$, assume the distance between $q_i,q_j$ is $d_{\delta}(q_i,q_j)=L\cdot \delta$, i.e., they can be connected by $L$ consecutive simplexes on $\mathcal{C}^{k}_{\mathcal{M},\delta}(K)$. According to Proposition~\ref{local behavior of D-Triangulation}, each simplex can be connected to the next adjacent barycenter $H_{\delta}(b(\sigma^{\prime}))$ by a geodesic of length at most $r_0$. Therefore, the distance of $x,y$ on $\mathcal{M}$ is controlled by:
        \begin{equation}
            d_{\mathcal{M}}(w,v)\leq r_0 L\leq 4K_2\rho_0\, d_{\delta}(q_i,q_j).\label{ineq:fractal biLip_1}
        \end{equation} 
        For the converse direction, for any $w=H_{\delta}(b(\sigma_i)),v=H_{\delta}(b(\sigma_j))\in K_{\delta}$ with distance $d_{\mathcal{M}}(w,v)=L\geq \frac{DK_1\Upsilon_0}{k+1}\delta$, there exists two point $x_w\in H_{\delta}(\tau_i)\bigcap K, x_v\in H_{\delta}(\tau_j)\bigcap K$, where $\sigma_i\preceq \tau_i,\sigma_j\preceq\tau_j\in \mathcal{C}^{D}_{\mathcal{M},\delta}(K)$ and the distance between $x_w,x_v$ on $K$ is upper bounded by 
        \begin{equation}
            d_{K}(x_w,x_v)\leq \gamma d_{\mathcal{M}}(x_w,x_v)\leq \gamma(d_{\mathcal{M}}(w,v)+2\sup_{u\in K_{\delta}}\inf_{x_{u}\in K} d(u,x_{u}))\leq \gamma(L+2A\delta).
        \end{equation}
        Then we may cover the rectifiable shortest path connecting $x_w$ and $x_v$ on $K$ by a sequence of geodesic balls of radius $r_0$, each intersecting at most $N_0$ $D$-simplexes under $H_{\delta}^{-1}$. Differ by a constant factor $c_0\geq 1$, the number of such balls is at most $\frac{c_0\gamma(L+2A\delta)}{r_0} + 2$, so the total number of $k$-simplexes to connect $\sigma_i,\sigma_j$ is bounded by $J_{D,k}N_0 \left(\frac{c_0\gamma(L+2A\delta)}{r_0} + 2\right)$. This observation yields a uniform upper bound for the graph distance on $G_{\delta}(K,k)$:
        \begin{align}
            d_{\delta}(q_i,q_j) &\leq \delta \cdot J_{D,k}N_0 \left(\frac{c_0\gamma(L+2A\delta)}{r_0} + 2\right)= \delta\cdot J_{D,k}N_0 \left(\frac{c_0\gamma(L+4K_2\rho_0\delta)}{4K_2\rho_0\delta} + 2\right) \\
            &= J_{D,k}N_0\left[\frac{c_0\gamma}{4K_2\rho_0}L+(2+c_0\gamma)\delta\right]\leq J_{D,k}N_0\left[\frac{c_0\gamma}{4K_2\rho_0}+\frac{(c_0\gamma+2)(k+1)}{DK_1\Upsilon_0}\right]L\\
            &= J_{D,k}N_0\left[\frac{c_0\gamma}{4K_2\rho_0}+\frac{(c_0\gamma+2)(k+1)}{DK_1\Upsilon_0}\right]d_{\mathcal{M}}(w,v).\label{ineq:fractal biLip_2}
        \end{align}
        Combine inequalities~\eqref{ineq:fractal biLip_1} and~\eqref{ineq:fractal biLip_2} together, we conclude that $\mathcal{S}_{\delta} \circ H_{\delta}^{-1}|_{K_{\delta}}$ is uniformly bi-Lipschitz between $(K_{\delta},d_{\mathcal{M}})$ and $(G_{\delta}(K,k),d_{\delta})$. Hence, the induced Assouad dimension of $\{G_{\delta}(K,k)\}_{\delta\in I}$ agrees with the Assouad dimension of $\{K_{\delta}\}_{\delta\in I}$, which in turn agrees with the Assouad dimension of $K$.
    \end{proof}

\section{Assouad--Nagata reduction for connectivity graphs}\label{app:A-N reduction}

In this appendix, we carefully discuss the connection between the Assouad and Nagata dimensions.  Using the Assouad--Nagata reduction, we show that a graph family with Assouad dimension $\beta$ has Nagata dimension at most $\lfloor \beta \rfloor$. This provides the geometric input needed in Section~\ref{sec:intrinsic dimension bridges code symmetry and indistinguishability} to derive constraints on the code properties from their intrinsic dimensions.

\begin{definition}[Nagata dimension, cf.~\cite{Assouad1982,Lang2005,Buyalo2007,le2015assouad}]\label{Nagata dimension}
    Let $\mathcal{X}=\{(G_{\lambda},d_\lambda)\}_{\lambda \in I}$ be a family of graphs indexed by $I$, where $d_{\lambda}$ is the path metric on $G_{\lambda}$. We say a graph $G_{\lambda}=(V_{\lambda},d_{\lambda})$ satisfies the $D$-division condition  if for any $r>0$, there always exists $D+1$ families of subsets of $V_{\lambda}$,  $\mathcal{B}_1=\mathop{\bigcup}\limits_{i_1} B_1^{i_1},\mathcal{B}_2=\mathop{\bigcup}\limits_{i_2} B_2^{i_2},\cdots,\mathcal{B}_{D+1}=\mathop{\bigcup}\limits_{i_{D+1}} B_{D+1}^{i_{D+1}}$, which satisfy the conditions:
    \begin{itemize}
        \item $\mathop{\bigcup}\limits_{i=1}^{D+1}{ \mathcal{B}_i}=V_{\lambda}$,
        \item Each $\mathcal{B}_i$ is $r$-separated, i.e., any pair of sets $B_i^{(1)},B_i^{(2)}$ in $\mathcal{B}_i$ has distance $d_{\lambda}(B_{i}^{(1)},B_i^{(2)})\geq r$,
        \item $B\in \mathcal{B}_i$ is $c\cdot r$-bounded, i.e., each $B$ is contained in a ball with radius $c\cdot r$. 
    \end{itemize}
     The \emph{Nagata dimension} of $\mathcal{X}$ is the minimal integer $D\in \mathbb{N}$ such that the $D$-division conditions are satisfied for all $G_{\lambda},\lambda\in I$. Furthermore, we say $\mathcal{X}=\{(G_{\lambda},d_\lambda)\}_{\lambda \in I}$ has Nagata dimension $D$ up to scale $\overline{r}$, if $D$ is the minimal integer such that $D$-division conditions are satisfied for any $0<r<\overline{r}$. 
\end{definition}

  \begin{theorem}[\cite{le2015assouad}]
      For any metric space $(X,d)$, the Nagata dimension of $X$ is bounded above
    by
    the Assouad dimension of $X$:
 \begin{align}
\mathrm{dim}_{\mathrm{Nagata}}X\leq \dimasd X.
\end{align}
  \end{theorem}
  The original proof focuses on a single infinite metric space~\cite{le2015assouad}, whereas our framework requires a formulation that applies uniformly to a family of graphs. The following lemma explicitly proves that this reduction remains valid across all scales and is independent of the index $\lambda$. 

 \begin{lemma}[Assouad--Nagata reduction with scale]\label{Assouad--Nagata Reduction at scale $R$} Let $(X,d)$ be a metric space with finite Assouad dimension $\dim_{\mathrm{Assouad},R}(X) = \beta$ up to scale $R$, satisfying the covering condition with constant $C_R$. For any $\beta < \alpha < \lfloor\beta\rfloor + 1$, there exists a constant $C_\alpha$, depending only on $\alpha$ and $C_R$, such that $X$ has Nagata dimension at most $\lfloor\alpha\rfloor$ up to scale $r_{\mathrm{Nagata}} = C_\alpha R$.
\end{lemma}

\begin{proof}
    To establish this result, we follow the iterative annular decomposition argument of Ref.~\cite{le2015assouad}. We first fix a pair of radii $(r,R_0)$ satisfying $r < R_0<\frac{R}{2}$, where $R_0$ will be explicitly determined at the end of the iteration. At each step $m$ (for $1 \leq m \leq \lfloor\alpha\rfloor + 1$), we partition the metric space $X$ into a disjoint union:
    \begin{align}
        X = X_m \bigsqcup \bigcup_{k=1}^{m} \left( \bigsqcup_{n=1}^{N_k} Y_n^k \right),
    \end{align}
    where the sets $Y_n^k$ satisfy $\mathrm{diam}(Y_n^k) < 2R_0$ and $\mathrm{dist}(Y_n^k, Y_l^k) > r$ for any $n \neq l$. The construction ensures that for any $x \in X$, the local residual set $X_m \cap N_x(R_0)$ can be covered by at most $C_m (R_0/r)^{\alpha-m}$ balls of radius $r$. 
    
    Our goal is to show that at the final step $m = \lfloor\alpha\rfloor + 1$, we can choose $r$ sufficiently small such that the residual space $X_m$ becomes empty. Specifically, to guarantee $X_m = \emptyset$, we require $C_m (R_0/r)^{\alpha-m} < 1$, where $C_m$ is a constant accumulating at each step, depending only on $\alpha$ and $C_{R}$. This condition is satisfied when $r < C_m^{-\frac{1}{m-\alpha}} R_0$. 
    
    To explicitly construct this partition, we begin with the base case. By Zorn's lemma, there exists a maximal $(R_0/4)$-separated net of $X$, denoted by $\{x_i\}_{i=1}^F$. Assuming the space is separable, $F$ is at most countably infinite. By the maximality of the net, the union of balls $N_{x_i}(R_0/2)$ forms a cover of $X$: \begin{align}\bigcup_{i=1}^F N_{x_i}(R_0/2) = X.
    \end{align}
    
    Since the space $X$ has Assouad dimension $\beta < \alpha$, for any $x_i$ in the net, the ball $N_{x_i}(R_0)$ can be covered by a collection $\mathcal{B} = \{N_{y_j}(r) \mid j \in [C_R (R_0/r)^\alpha]\}$ of at most $C_R (R_0/r)^\alpha$ balls of radius $r < R_0/4$. We define an integer $w$ such that $wr \leq R_0/2 < (w+1)r$. Consider the concentric annuli:
    \begin{align}
        A_{n,i} = N_{x_n}(R_0 - ir) \setminus N_{x_n}(R_0 - (i+1)r).
    \end{align}
    The disjoint union of $w$ such annuli forms a thick shell:
    \begin{align}
        \bigsqcup_{i=0}^{w-1} A_{n,i} = N_{x_n}(R_0) \setminus N_{x_n}(R_0 - wr) \subset N_{x_n}(R_0) \setminus N_{x_n}(R_0/2).
    \end{align}
    Because any single ball in $\mathcal{B}$ can intersect at most 3 such annuli $A_{n,i}$, the pigeonhole principle guarantees the existence of an index $0 \leq w_n \leq w-1$ such that the specific annulus $A_{n,w_n}$ is covered by at most $\frac{3}{w} C_R (R_0/r)^\alpha$ balls from $\mathcal{B}$. Since $w = \lfloor R_0/(2r) \rfloor > R_0/(2r) - 1$, we have
    \begin{align}
        \frac{3}{w} C_R \left(\frac{R_0}{r}\right)^\alpha \leq \frac{3}{\frac{R_0}{2r} - 1} C_R \left(\frac{R_0}{r}\right)^\alpha \leq 12 C_R \left(\frac{R_0}{r}\right)^{\alpha-1}.
    \end{align}
    
    We then define the first residual space as $X_1 = \bigcup_{n=1}^F A_{n,w_n}$, and construct the disjoint sets:
    \begin{align}
        Y_n^1 = N_{x_n}(R_0 - (w_n+1)r) \setminus \left( X_1 \cup \bigcup_{i=1}^{n-1} N_{x_i}(R_0 - w_ir) \right).
    \end{align}
    By construction, the family $\{Y_n^1\}_{n=1}^F$ is $r$-separated and bounded by diameter $2R_0$. Furthermore, $X_1 \bigsqcup \left(\bigsqcup_{n=1}^F Y_n^1\right) = \bigcup_n N_{x_n}(R_0 - (w_n+1)r) \supset \bigcup_n N_{x_n}(R_0/4) = X$ forms a valid covering partition. Indeed, for every $ x\notin X_1$, there exists a smallest integer $n$ such that $x\in N_{x_n}(R_0-w_n r)$, which implies that $x\in Y^1_n$.
    
    Next, we bound the covering number of $X_1$. For any $x \in X$, if $A_{n,w_n} \cap N_x(R_0) \neq \emptyset$, it follows that $x_n \in N_x(2R_0)$. Since the net $\{x_n\}$ is $(R_0/4)$-separated, applying the covering condition to $N_x(2R_0)$ with balls of radius $R_0/9$ implies that the number of such net points $x_n$ within $N_x(2R_0)$ is bounded by $C_R \cdot 18^\alpha$. Therefore, $N_x(R_0)$ intersects at most $C_R \cdot 18^\alpha$ annuli $A_{n,w_n}$ in $X_1$. Consequently, $N_x(R_0) \cap X_1$ can be covered by at most $C_1 (R_0/r)^{\alpha-1}$ balls of radius $r$, where $C_1 = 12 \cdot 18^\alpha C_R^2$.
    
    This peeling process can be iterated for each subsequent residual space $X_k$, reducing the exponent of $(R_0/r)$ by $1$ at each step. The accumulated constants $C_k$ strictly depend only on $C_R$ and $\alpha$. The iteration terminates with the final partition $\bigsqcup_{k=1}^{m} \left( \bigsqcup_{n=1}^{N_k} Y_n^k \right)$, where $m=\lfloor\alpha\rfloor+1$ and the $C_{m}$ is the constant depending only on $\alpha$ and $C_{R}$. We set $C_{\alpha}\coloneqq \min\{\frac{1}{16},\frac{1}{4}C_{m}^{-\frac{1}{m-\alpha}}\}$, $R_0=\frac{1}{2C_{\alpha}}r$. Then for any $0<r<r_{\mathrm{Nagata}}=C_{\alpha}R$, we derive $R_0=\frac{r}{2C_{\alpha}}<\frac{R}{2}$. The initial assumption $\mathrm{diam}(Y_n^k) < 2R_0$ and $\mathrm{dist}(Y_n^k, Y_l^k) > r$ provide a Nagata division with $r$-separateness and $\frac{1}{C_{\alpha}}r$-boundness, completing the proof that the Nagata dimension is at most $\lfloor\alpha\rfloor$ at scale $C_\alpha R$. 
    
\end{proof}
 \begin{theorem}
     \label{A-N bound at scale}
     Let $\mathcal{X}=\{(G_{\lambda},d_{\lambda})\}_{\lambda\in I}$ be a family of graphs with shortest-path metric $d_{\lambda}$.
     If $\mathcal{X}$ has Assouad dimension $\beta$ up to scale $R$, then the Nagata dimension of $\mathcal{X}$ is bounded above by $\lfloor\beta\rfloor$ up to scale $c_{\beta} R$, where $c_{\beta}$ depends only on $\beta$ and the uniform constant $C$ in the definition of Assouad dimension .
 \end{theorem} 
 \begin{proof}
     Take $\alpha=\beta+\epsilon$, where $\epsilon>0$ is sufficiently small such that $\lfloor\beta\rfloor=\lfloor\alpha\rfloor$. Applying Lemma~\ref{Assouad--Nagata Reduction at scale $R$} for each $G_{\lambda},\lambda\in I$, we obtain a corresponding coefficient that is independent of the index $\lambda$.
 \end{proof}
 Given that the scale associated with the Assouad dimension can be chosen arbitrarily, the Nagata dimension is bounded across all scales, up to a constant that depends on $C$ and $\beta$, by taking the asymptotic limit $R \to +\infty$ in Theorem~\ref{A-N bound at scale}.
 \begin{corollary} \label{Assouad-Nagate bound} 
    Let $\mathcal{X}=\{(G_{\lambda},d_{\lambda})\}_{\lambda\in I}$ be a family of graphs with shortest-path metric $d_{\lambda}$. If the Assouad dimension of $\mathcal{X}$ at  all scales is $\beta$, then the Nagata dimension of $\{G_{\lambda}\}$ is bounded above by $\lfloor\beta\rfloor$ at all scales:
    \begin{align}\mathrm{dim}_{\mathrm{Nagata}}\mathcal{X}\leq \lfloor\beta\rfloor\leq \dimasd\mathcal{X}=\beta.
    \end{align}
\end{corollary}

 \begin{remark}
      Unlike the Assouad dimension, the Hausdorff dimension and the Nagata dimension are not comparable in general. For examples, even in simple cases, $\mathrm{dim}_{\mathrm{Nagata}}(\mathbb
     Q)=1>0=\mathrm{dim}_{\mathrm{Hausdorff}}(\mathbb{Q},|\cdot|_{\mathbb{R}})$ but $\mathrm{dim}_{\mathrm{Nagata}}(\mathbb{R})=1<2=\mathrm{dim}_{\mathrm{Hausdorff}}(\mathbb{R},|\cdot|_{\frac12}).$  
 \end{remark}

\section{Fractalized quantum codes}\label{app:Fractalized stabilizer code}

In this appendix, we discuss the intrinsic dimension of the fractalized stabilizer codes introduced in Ref.~\cite{Devakul2021fractalizingquantum}. The fractalization procedure can be viewed as a map that generates
higher-dimensional quantum codes from a set of linear cellular automaton rules. Our purpose is to show that the codes produced by this construction have integer Assouad dimension. We refer the reader to Ref.~\cite{Devakul2021fractalizingquantum} for details of the construction. Hence, from the perspective of intrinsic dimension, these codes do not provide a route to constructing fractional-dimensional code families that improve the parameter tradeoffs considered in the main text.

We firstly briefly summarize the construction, with notations following Ref.~\cite{Devakul2021fractalizingquantum}. The key observation of the fractalization map is that given an input $D$ dimensional stabilizer code and a set of LCA rules for creating an additional $m$ dimension, the output is a quantum code on a $D+m$ dimensional stabilizer code. 

The formalism is expressed in the language of polynomials. A Pauli operator on a \(D\)-dimensional translation-invariant stabilizer code can be represented by a polynomial over
\[
\mathbb{F}_2[x_1,\ldots,x_D]/(x_1^{L_1}-1,\ldots,x_D^{L_D}-1),
\]
where each monomial labels a site of the \(D\)-dimensional lattice. An \(X\) or $Z$-type Pauli operator anchored at \(\mathbf r_0\) is written as
\begin{equation}
    A_{\bm{r_0}}=\sigma_{X}\left[x^{\bm{r_0}}a(x)\right], B_{\bm{r_0}}=\sigma_{Z}\left[x^{\bm{r_0}}b(\overline{x})\right].
\end{equation}
The fractalization procedure extends the original \(D\)-dimensional code by introducing \(m\) additional lattice directions, denoted by
\(y=(y_1,\ldots,y_m)\). The behavior of the model along these added directions is determined by a collection of linear cellular automaton (LCA) update rules, which is denoted by
\begin{align}\mathbf f(y)=\{f^{(i)}(y)\}_{i=1}^{m}.
\end{align}
Following Ref.~\cite{Devakul2021fractalizingquantum}, the stabilizers are mapped from the original $A_{r}$ to a family of fractalized stabilizers on a higher dimensional lattice, 
\begin{align}
    &A^{\mathrm{frac}}_{\bm{r_0},\bm{s_0}}=\sigma_{X}\left[x^{\bm{r_0}}y^{\bm{s_0}} a(\bm{f}(y)\circ x)\right],\\
    &B^{\mathrm{frac}}_{\bm{r_0},\bm{s_0}}=\sigma_{Z}\left[ x^{\bm{r_0}}y^{\bm{s_0}}b(\bm{f}(\overline{y})\circ \overline{x})\right].
\end{align} where $\left[\bm{f}(y)\circ x\right]^{\bm{r}}=\left(\mathop{\prod}\limits_{i=1}^{m}f^{(i)}(y)^{r_i}\right)x^{\bm{r}}$. 


We now discuss several examples from Ref.~\cite{Devakul2021fractalizingquantum}. These examples illustrate that
fractalization can produce operators with visually fractal support, but the associated connectivity graphs still have ordinary integer-dimensional growth.

The simplest example is the fractalized one-dimensional Ising model with
linear cellular automaton rule \(f(y)=1+y\). The original and fractalized
Hamiltonians are

\begin{align}
    H_{\mathrm{Ising},1\mathrm{D}}
    &=
    -\sum_{r\in \mathbb{Z}/L\mathbb{Z}}
    Z_r Z_{r-1},
    \\
    H_{\mathrm{Ising},1\mathrm{D}}^{\mathrm{frac}}
    &=
    -\sum_{(r,s)}
    Z_{(r,s)} Z_{(r,s-1)} Z_{(r-1,s-1)} .
\end{align}

\begin{figure}[t]
    \centering
    \includegraphics[width=0.48\columnwidth]{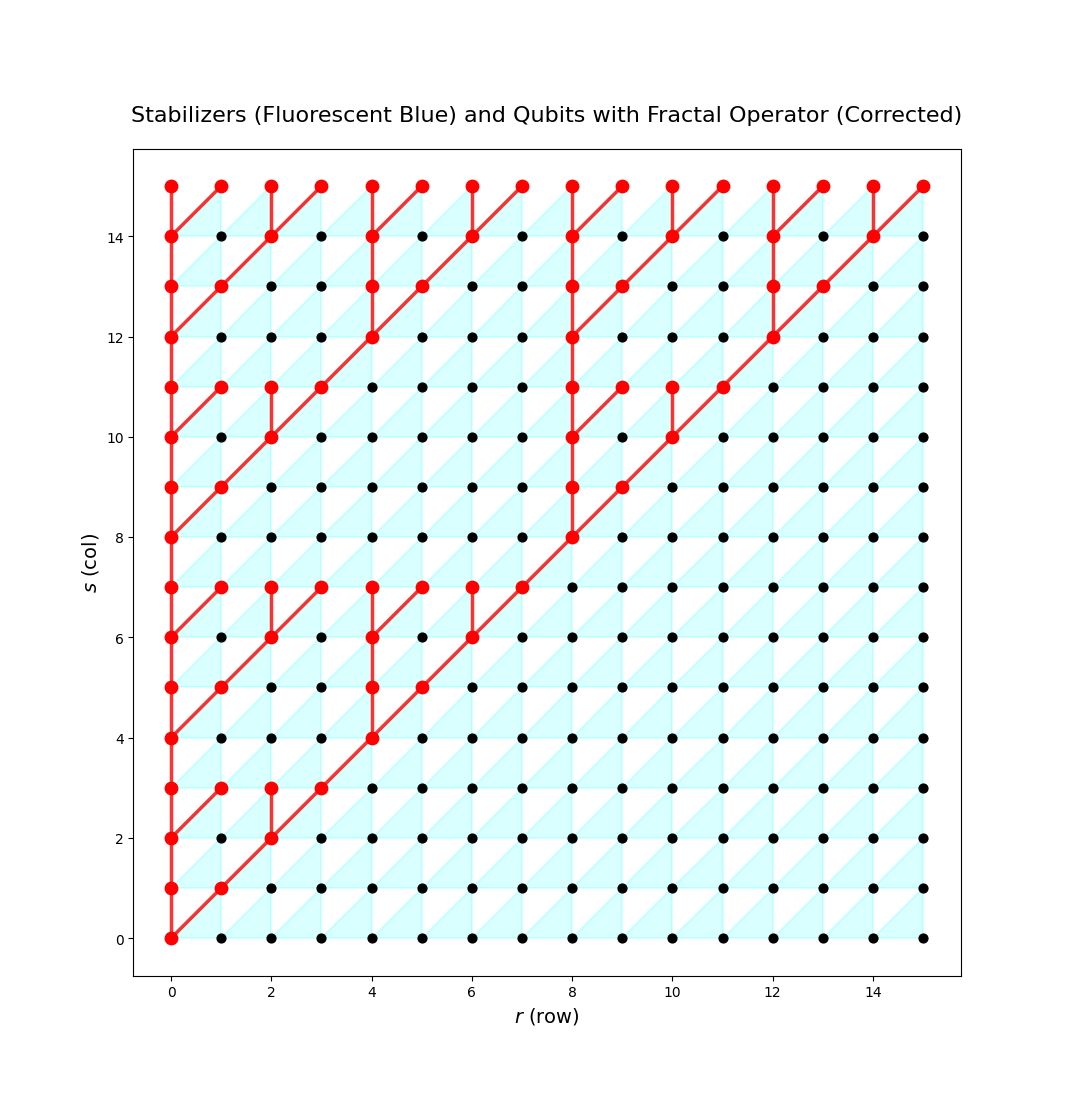}
    \caption{
    A logical operator in the fractalized one-dimensional Ising model. Although
    the operator has fractal support, the connectivity graph has integer
    Assouad dimension \(2\).
    }
    \label{fig:fractalized_1D_Ising}
\end{figure}

\begin{figure}[t]
    \centering
    \includegraphics[width=0.23\columnwidth]{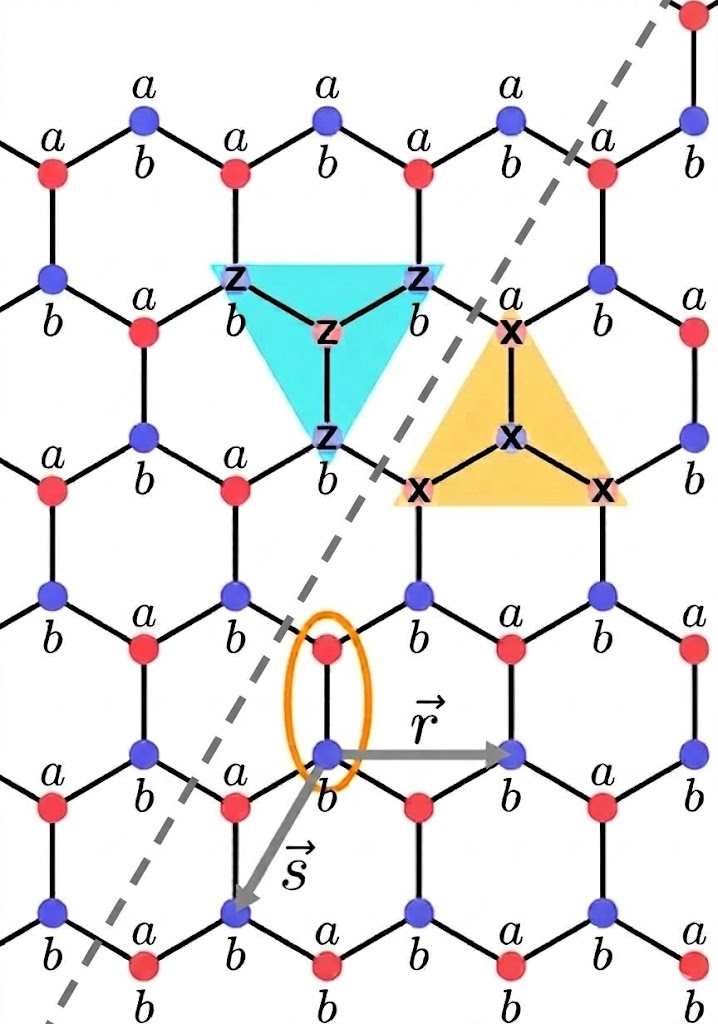}
    \caption{
    A two-dimensional layout of the fractalized one-dimensional cluster model.
    Each orange circle denotes one site containing two qubits, labeled \(a\)
    and \(b\).  
    }
    \label{fig:fractalized_1D_cluster}
\end{figure}
The corresponding connectivity graph is a two-dimensional graph consisting of triangles(see Fig.~\ref{fig:fractalized_1D_Ising}). Therefore, despite the fractal form of its logical operators, the graph has
integer Assouad dimension \(2\). In particular, it admits a transversal $\overline{\mathrm{CNOT}}$ gate, saturating Theorem~\ref{Intrinsic gate bound}.

The fractalized one-dimensional cluster model gives a similar example. For the
same rule \(f(y)=1+y\), the Hamiltonians before and after fractalization are
\begin{align}
    H_{\mathrm{Cluster},1\mathrm{D}}
    &=
    -\sum_{r}
    X_{r}^{(a)}X_{r}^{(b)}X_{r+1}^{(a)}
    -\sum_{r}
    Z_{r-1}^{(b)}Z_{r}^{(a)}Z_{r}^{(b)},
    \\
    H_{\mathrm{Cluster},1\mathrm{D}}^{\mathrm{frac}}
    &=
    -\sum_{r,s}
    X_{(r,s)}^{(a)}
    X_{(r+1,s)}^{(a)}
    X_{(r+1,s+1)}^{(a)}
    X_{(r,s)}^{(b)}
    -\sum_{r,s}
    Z_{(r,s)}^{(a)}
    Z_{(r,s)}^{(b)}
    Z_{(r-1,s)}^{(b)}
    Z_{(r-1,s-1)}^{(b)} .
\end{align}
This model admits a two-dimensional layout on a honeycomb
lattice( see Fig.~\ref{fig:fractalized_1D_cluster}). Hence its connectivity
graph again has integer Assouad dimension \(2\). The code supports transversal 
Clifford gate implemented by Hadamards \(H^{\otimes n}\), followed by a
product of SWAP gates between symmetrically positioned qubits across the
dashed axis in Fig.~\ref{fig:fractalized_1D_cluster}. 

A higher-dimensional example is obtained by fractalizing the two-dimensional
toric code. The Hamiltonians before and after fractalization are
\begin{align}
    H_{\mathrm{TC},2\mathrm{D}}
    &=
    -\sum_{(r_1,r_2)}
    \left[
    \sigma_X\!\left(
    x^{\bm r}
    \begin{pmatrix}
        1+x_1 \\
        1+x_2
    \end{pmatrix}
    \right)
    +
    \sigma_Z\!\left(
    x^{\bm r}
    \begin{pmatrix}
        1+\overline{x}_1 \\
        1+\overline{x}_2
    \end{pmatrix}
    \right)
    \right],
    \\
    H_{\mathrm{TC},2\mathrm{D}}^{\mathrm{frac}}
    &=
    -\sum_{(r_1,r_2,s)}
    \sigma_X\!\left(
    x^{\bm r} y^s
    \begin{pmatrix}
        1+x_1 f_1(y) \\
        1+x_2 f_2(y)
    \end{pmatrix}
    \right)
    -\sum_{(r_1,r_2,s)}
    \sigma_Z\!\left(
    x^{\bm r} y^s
    \begin{pmatrix}
        1+\overline{x}_1 f_1(\overline y) \\
        1+\overline{x}_2 f_2(\overline y)
    \end{pmatrix}
    \right).
\end{align}

For $f_1(y)^{L_1}=f_2(y)^{L_2}=1$, the three-dimensional
layout decomposes into \(L_3\) inclined two-dimensional translation invariant codes, encoding $2L_3$ logical qubits. For other algebraic relations in the cellular-automaton rules, several inclined layers can form a larger
connected component. However, the resulting connectivity graphs still have dimension 2 as its Assouad dimension. 

\end{document}